%% file: main.tex
\documentclass[mnsc,nonblindrev,hidelinks]{informs3_homepage}
\usepackage{tikz}
\usetikzlibrary{fit}
\usetikzlibrary{positioning}
\input{macros}

\begin{document}
    \RUNAUTHOR{}
	\RUNTITLE{Data-Driven Switchback Experiments}
 \TITLE{Data-Driven Switchback Experiments: \\ Theoretical Tradeoffs and Empirical Bayes Designs\footnote{The preliminary version of this paper is circulated under the title ``Bias-variance tradeoffs for designing simultaneous temporal experiments.'' This version has more general theoretical results, a new design approach, and an empirical application. We would like to thank Susan Athey, Mohsen Bayati, Haoge Chang, Dean Eckles, Adam Glynn, and Ramesh Johari for their useful comments and suggestions.}
}

    \ARTICLEAUTHORS{
    \AUTHOR{Ruoxuan Xiong}
		\AFF{Department of Quantitative Theory and Methods, Emory University, \EMAIL{ruoxuan.xiong@emory.edu}}
    \AUTHOR{Alex Chin}
		\AFF{Motif Analytics, \EMAIL{alexchin12@gmail.com}}
    \AUTHOR{Sean J. Taylor}
		\AFF{Motif Analytics, \EMAIL{seanjtaylor@gmail.com}}
    }
	
	\ABSTRACT{
            We study the design and analysis of switchback experiments conducted on a single aggregate unit. The design problem is to partition the continuous time space into intervals and switch treatments between intervals, in order to minimize the estimation error of the treatment effect. We show that the estimation error depends on four factors: carryover effects, periodicity, serially correlated outcomes, and impacts from simultaneous experiments. We derive a rigorous bias-variance decomposition and show the tradeoffs of the estimation error from these factors. The decomposition provides three new insights in choosing a design: First, balancing the periodicity between treated and control intervals reduces the variance; second, switching less frequently reduces the bias from carryover effects while increasing the variance from correlated outcomes, and vice versa; third, randomizing interval start and end points reduces both bias and variance from simultaneous experiments. Combining these insights, we propose a new empirical Bayes design approach. This approach uses prior data and experiments for designing future experiments. We illustrate this approach using real data from a ride-sharing platform, yielding a design that reduces MSE by 33\% compared to the status quo design used on the platform.

	}

	\maketitle

    \input{section_1}

    \input{section_2}

    \input{section_3}

    \input{section_4}

    \input{section_5}

    \input{section_6}

 \bibliographystyle{ormsv080}
	\bibliography{reference}

 \clearpage
\newpage
\renewcommand{\theHsection}{A\arabic{section}}
\bookmarksetup{startatroot}

\begin{APPENDICES}
 \input{appendix-empirical}

 \input{appendix-theoretical}

 \input{appendix-proof}
 \end{APPENDICES}

\end{document}

%% file: macros.tex
\usepackage{natbib}
\usepackage{adjustbox,float,comment,booktabs,diagbox,caption,subcaption,graphicx,bm, makecell,textcomp,mathtools,multirow,fancyvrb,upgreek,mathrsfs,bbm,hyperref
}
\usepackage{bookmark}

\TheoremsNumberedBySection

\usepackage[ruled,norelsize]{algorithm2e}
\def\*#1{\mathbf{#1}}
\def\+#1{\mathbb{#1}}

\newcommand{\total}{\mathrm{total}}
\newcommand{\cum}{\mathrm{cum}}
\newcommand{\s}{\mathrm{s}}

\newcommand{\RNum}[1]{\uppercase\expandafter{\romannumeral #1\relax}}

\newcommand{\gate}{\mathrm{gate}}

\makeatletter
\newcommand{\BBigg}{\bBigg@{3}}
\newcommand{\vast}{\bBigg@{4}}
\newcommand{\Vast}{\bBigg@{5}}
\makeatother

\newcommand{\dem}{\mathrm{dem}}

\makeatletter
\usepackage{tikz}
\usetikzlibrary{calc}
\newcommand*\circled[1]{\tikz[baseline=(char.base)]{
    \node[shape=circle, draw, inner sep=1pt, 
        minimum height={\f@size*1.6},] (char) {\vphantom{WAH1g}#1};}}
\makeatother

\let\oldsection\section
\renewcommand{\section}{
\setcounter{equation}{0}
  \renewcommand{\theequation}{\thesection.\arabic{equation}}
  \oldsection}

\usepackage{tikz}
\usetikzlibrary{decorations.pathreplacing,calc}

\tikzset{brace/.style={decorate, decoration={brace}},
	brace mirrored/.style={decorate, decoration={brace,mirror}},
}

\newcounter{brace}
\newcounter{arrow}
\newcommand{\R}{\mathbb{R}} 	

\renewcommand{\P}{\mathbf{P}} 	
\newcommand{\bias}{\operatorname{Bias}}   
\newcommand{\var}{\operatorname{Var}}   
\newcommand{\Var}{\operatorname{Var}}   
\newcommand{\cov}{\operatorname{Cov}} 	

\newtheorem{condition}{Condition}

\bibpunct[, ]{(}{)}{,}{a}{}{,}
\def\bibfont{\small}
\def\bibsep{\smallskipamount}

\def\bibhang{24pt}

\definecolor{Red}{rgb}{1,1,1}
\definecolor{Blue}{rgb}{0,0,1}
\definecolor{Olive}{rgb}{0.41,0.55,0.13}
\definecolor{Green}{rgb}{0,1,0}
\definecolor{MGreen}{rgb}{0,0.8,0}
\definecolor{DGreen}{rgb}{0,0.55,0}
\definecolor{Yellow}{rgb}{1,1,0}
\definecolor{Cyan}{rgb}{0,1,1}
\definecolor{Magenta}{rgb}{1,0,1}
\definecolor{Orange}{rgb}{1,.5,0}
\definecolor{Violet}{rgb}{.5,0,.5}
\definecolor{Purple}{rgb}{.75,0,.25}
\definecolor{Brown}{rgb}{.75,.5,.25}
\definecolor{Grey}{rgb}{.5,.5,.5}
\definecolor{Pink}{rgb}{1,0,1}
\definecolor{DBrown}{rgb}{.5,.34,.16}

\definecolor{Black}{rgb}{0,0,0}

\def\grey{\color{Grey}}

\def\cmtrm#1{{\grey{}}}

\def\gate{\mathrm{gate}}
\def\total{\mathrm{gate}}
\def\inst{\mathrm{inst}}
\def\co{\mathrm{co}}

\def\simul{\mathrm{simul}}

\usetikzlibrary{arrows.meta}

\definecolor{MyTriangle}{RGB}{164,212,255}

\definecolor{MyArrow}{RGB}{84,106,160}


%% file: section_1.tex
\section{Introduction}

Experimentation has become an increasingly popular and effective tool for testing and improving social and business policy in digitally-mediated economic and social settings. However, the scale and complexity of modern digital applications have given rise to scientific and statistical challenges for the design and analysis of experiments. 

The leading example considered in this paper is a ride-hailing platform where a product team would like to measure the effects of their product (e.g., matching or pricing algorithm) changes through an experiment. The experiment is run on a geographically determined market for two weeks. The product team chooses an experimental design ex-ante that determines when the current and test product versions are used during the experiment.

Such a product change may affect users' behavior in ways that create interference and alter outcomes for other users on both the rider and driver side of the marketplace. To mitigate user-level interference, companies often aggregate all users in a market into one unit and employ switchback designs on this unit.\footnote{See examples in Amazon \citep{masoero2023efficient,Cooprider2023amazon}, DoorDash \citep{Kastelman2018switchback}, Lyft \citep{chamandy2016experiment}, Tubi~\citep{silbert2022switchback}, LinkedIn and Netflix \citep{bojinov2020avoid}. } Prior to the recent applications in marketplaces, switchback designs were originally used in agriculture \citep{cochran1941double} and medicine \citep{mirza2017history}. 
These designs randomly switch between treatment and control for the same unit over time \citep{chamandy2016experiment}.Post experiment, a quantity we call the global average treatment effect (GATE) is commonly estimated. The GATE measures the difference in average outcomes over users and time between when the product change is deployed indefinitely (global treatment) versus when it is absent (global control).

Precise estimation of GATE is important in deciding whether to launch the product change indefinitely. The precision can be improved by using a better switchback design. Prior work (\cite{bojinov2020design,hu2022switchback} among others) has studied the design problem concerning the \textit{carryover effects}, which is a major source of the estimation error. Carryover effects measure the impact of past interventions on future outcomes. They are nonzero when interventions take time to change the marketplace to a new equilibrium state where the globally treated outcomes can be observed.\footnote{Carryover effects can be viewed as interference in the temporal dimension, which is distinct from interference across users.} To reduce the estimation error from carryover effects, switching less frequently can help.

In this paper, we show that the estimation error more generally depends on four factors: \emph{carryover effects} considered in prior work, \textit{periodicity}, \textit{correlated outcomes}, and \textit{simultaneous interventions}, as illustrated in Figure \ref{fig:four-factors}. Specifically, periodicity arises from the significant variation in rider demand and driver availability over time in a day and days in a week. Outcomes close in time are correlated due to weather, traffic, or other external factors such as supply and demand shocks. Other product teams may test other product changes simultaneously in the same market, complicating the measurement of the marginal effect of each one.

 \begin{figure}[t!]
	\centering
  \includegraphics[width=1\linewidth]{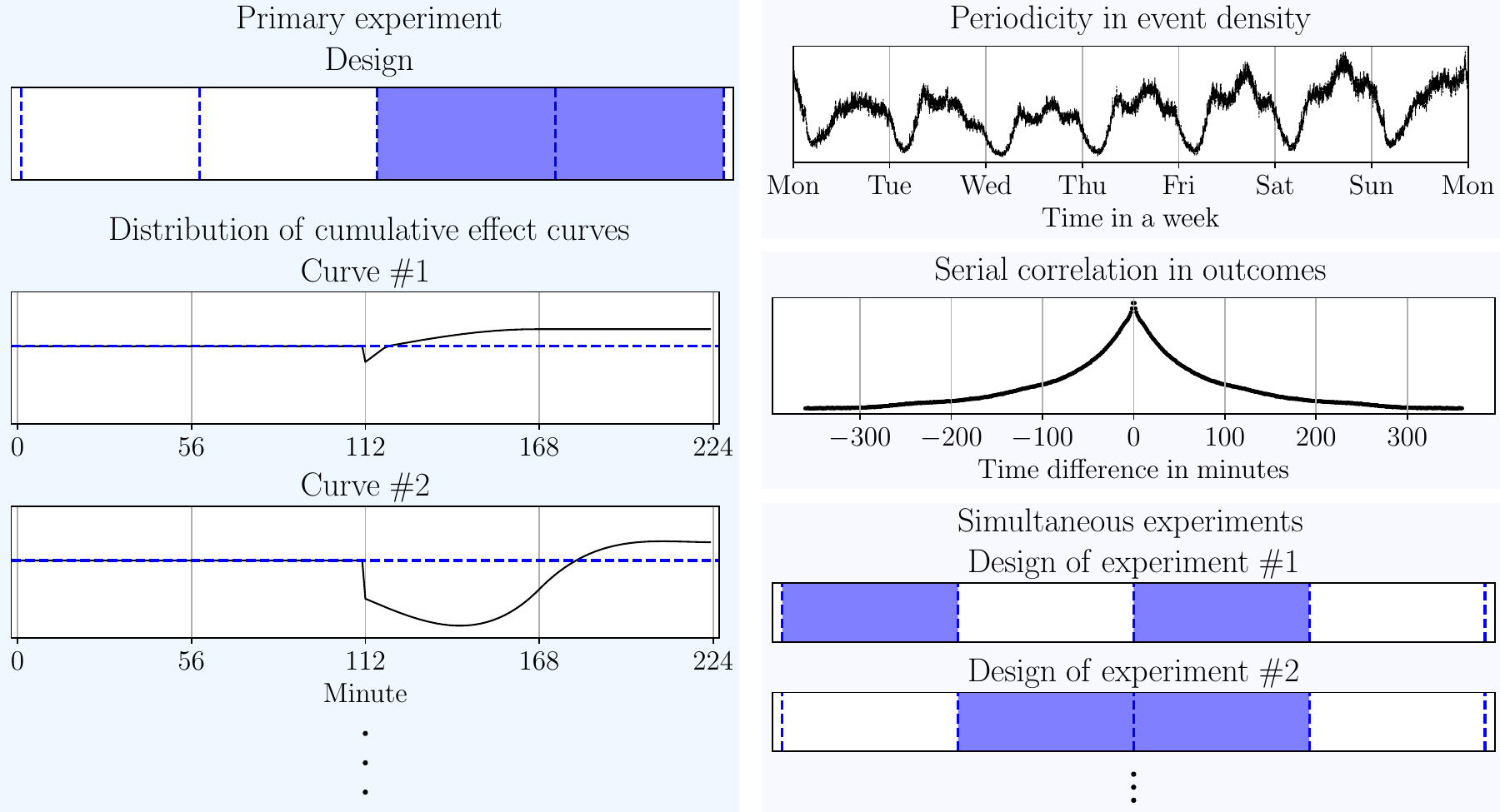}
	\caption{Illustration of four factors that affect the estimation error of GATE. This figure is generated using the data on a ride-sharing platform. In the illustration of (toy) switchback designs for both primary and simultaneous experiments, dash lines are switching points, and treated intervals are shaded. Each interval has a fixed length of $56$ minutes. There may be one or more experiments running simultaneously with the primary experiment. In the illustration of cumulative effect curves (CECs), the horizontal line is at the value of $0$. The CEC of the primary experiment is unknown ex-ante, but we can estimate a distribution of CECs from prior experiments. In the illustration of periodicity, the event density for the time in a week is shown, where the event can be a rider opening the app and checking the price. More illustrations of the periodic patterns are shown in Section \ref{sec:empirical}. In the illustration of serially correlated outcomes, the serial correlation in outcomes decreases with the absolute value of the time difference between two outcomes. 
 }
	\label{fig:four-factors}
\end{figure}

\subsection{Summary of Contributions}

We study the design and analysis of switchback experiments, accounting for all four factors. While we use the ride-hailing platform as the leading example, our work also applies to other contexts where some or all of these factors exist. 

For the design, we propose an empirical Bayes approach that uses knowledge from prior experiments to inform the design of new experiments. 
We illustrate this approach through a case study of a large corpus of prior experiments run in various geographically determined markets on a ride-sharing platform. 
This empirical Bayes approach first estimates an empirical distribution of cumulative effect curves (CECs), where each one is estimated from one experiment-market pair. 
The cumulative effect, as a measurement of carryover effects, is the total effect from current and past treatments,\footnote{Formally, the cumulative effect is the sum of instantaneous effect (effect of current treatment) and carryover effects (effects of past treatments).} which converges to GATE as the treatment duration grows to infinity. The shape of CEC shows the relationship between cumulative effect and treatment duration. In our case study, most CECs are non-monotonic, with 68\% of CECs changing signs as the treatment duration changes. See two examples in Figure \ref{fig:four-factors}. Non-monotonicity and changing signs can occur, for example, when testing a new surge pricing algorithm. The surge pricing may immediately reduce the demand and negatively affect the outcome (e.g., conversion rate). However, surge pricing attracts more drivers, gradually increases the supply, reduces the ETA, and then positively affects the outcome. With a non-monotonic curve similar to Curve \#2 in Figure \ref{fig:four-factors}, the estimated GATE can have larger errors from carryover effects than that with a curve like Curve \#1.

This empirical Bayes approach then chooses a switchback design by comparing the performance of candidate designs in synthetic experiments. These synthetic experiments are run on historical data, and the CEC is randomly drawn from the empirical distribution of CECs, ensuring that all four factors are present.
The best-performed design randomizes and, on average, doubles the switching lengths than the status quo design with fixed lengths, yielding a 33\% reduction in MSE.
Furthermore, the order of candidate designs encompasses a {\it hierarchical structure} in the effectiveness of three design principles: (a) {\it balancing} periodicity is the most effective; (b) carefully {\it selecting average switching periods} is moderately effective; (c) {\it tuning exact switching times} is mildly effective. 

For the theoretical analysis, we provide a rigorous decomposition of bias and mean-squared error (MSE) of the estimated GATE from the standard Horvitz–Thompson estimator, showing the tradeoffs of four factors.
The bias is decomposed into two sources of errors: (a) carryover effects from treatment at earlier times; (b) confounding effects from simultaneous interventions. The MSE is decomposed into squared bias and variance, where the variance is affected by three sources of randomness: (a) the measurement errors of outcomes and their covariance, determined by their distance in time; (b) the randomness of treatment assignments of primary and simultaneous interventions; (c) the randomness in event occurrence times. 

The decomposition together with a careful simulation study explains how the design principles reduce the estimation error: (a) balancing periodicity reduces all sources of variance; (b) switching less frequently reduces bias from carryovers; (c) switching more frequently reduces variance from the randomness in measurement errors and treatment assignments; (d) randomizing interval start and end points reduces both bias and variance from simultaneous interventions.
These insights manifest that the hierarchical structure identified in the ride-sharing setting is due to the large noise compared to the size of GATE; once the variance is reduced through balancing, bias dominates because of the non-monotonic CECs, and switching less frequently helps.

\subsection{Related Work}

Our work contributes to the literature on the design and analysis of experiments \citep{wu2011experiments}, and more specifically to literature on the time-based experiments or crossover trials \citep{jones2003design}.
Our work is most closely related to the growing literature on the design and analysis of switchback experiments, starting from \cite{bojinov2020design}. 
We complement this literature \citep{bojinov2020design,hu2022switchback,wu2022dynamic,ni2023design,li2023experimenting,chen2023switchback,masoero2023efficient} and show that, in addition to carryover effects considered in this literature, factors including periodicity, correlation event outcomes, and simultaneous interventions can affect the performance of switchback design. Due to these factors, we show the value of using prior data and experiments, and propose new designs to improve efficiency.

Our design is also related to a number of other designs in time-based experiments. One design is the staggered rollout design for panel experiments, where the design selects an initial (and possibly different) treatment time for each unit, either under temporal interference \citep{xiong2019optimal,basse2019minimax} or network interference \citep{cortez2022staggered,han2022detecting,boyarsky2023modeling}.
Another design is the synthetic control design for panel experiments, where the design selects units to be treated, allocates treatment to all of them in a single period, and forms a synthetic treated and control unit for treatment effect estimation \citep{doudchenko2019designing,doudchenko2021synthetic,abadie2021synthetic}. Concerning periodicity and nonstationarity, we complement \cite{wu2022non,wu2023non} and \cite{simchi2024non}, and provide a solution to balance periodicity when experimenting on an aggregate unit.

In our setup, we aggregate the units such that interference between users can be abstracted away; however, a growing literature directly tackles interference using novel experimental ideas. For example, on network data, cluster-randomized designs are commonly used for mitigating interference \citep{ugander2013graph,eckles2017design,candogan2021near,holtzreducing}, where the clusters are chosen to minimize edges that cut across clusters. Another popular method is the two-stage or multi-stage randomization, which has been used in public health \citep{hudgens2008toward,liu2014large}, digital platforms \citep{ye2023cold}, political science \citep{sinclair2012detecting}, and social science \citep{crepon2013labor,baird2018optimal,basse2018analyzing}. In the two-sided marketplace, recently proposed designs to mitigate interference bias include multiple randomization designs \citep{bajari2021multiple,johari2020experimental} and designs that perturb treatments near equilibrium outcomes \citep{wager2021experimenting}. \cite{li2021interference} characterize the bias and variance of such experiments and describe how the design can be optimized in such settings. Besides using novel designs, there is a growing literature on developing new treatment effect estimators and inferential theory accounting for the interference (\cite{chin2018central,chin2019regression,forastiere2021identification,qu2021efficient,yuan2021causal,yuan2023two,leung2023network,farias2022markovian} among others). Complementing this literature, our paper takes an agnostic approach to marketplace interference structure and leverages time-based experiments for treatment effect estimation.

Finally, this paper accounts for the impact of interventions tested simultaneously, which is common in practice. When multiple interventions are simultaneously applied to the same units, factorial design \citep{fisher1936design} is commonly used \cite{wu2011experiments}, which allows for estimating the effect of any treatment combination, especially useful in cross-sectional experiments with sufficiently low variance.
In general, design and analysis are more complex when the number of interventions is large \citep{dasgupta2015causal}.\footnote{This is because the number of treatment combinations increases exponentially in the number of interventions.} \cite{ye2023deep} provide a novel solution to this problem using debiased deep learning. In this paper, we study the setting of time-based experiments, which is more challenging. We focus on designing and analyzing a single primary intervention while being agnostic to how the other experiments were designed.

%% file: section_2.tex
\section{Problem Setup}\label{sec:setup}

Suppose a decision maker runs an experiment on a geographically determined market between time $0$ and time $T$ to study the effect of a new intervention, such as a new pricing or matching algorithm. In this paper, we adopt a continuous time framework, which naturally captures the types of event stream data that are commonly observed on ride-hailing platforms and similar settings.

Before the experiment starts, the decision maker chooses the treatment design for the whole experiment horizon, i.e., $\bm{W} = \{W_{t} \in \{0, 1\},\, \forall t \in [0,T]\}$. Here $W_{t} = 1$ indicates that all users in this market are exposed to the intervention (treatment) at time $t$, and $W_{t}=0$  indicates otherwise (control). Note that here, the treatment decisions are made for a continuous time interval. Therefore, we propose to use a two-step design procedure: the decision maker first partitions the experimental horizon $[0,T]$ into $M$ disjoint intervals and then randomly chooses the treatment assignment of each interval. The design procedure is fully general and can essentially yield any design with an appropriate choice of $M$ and treatment assignment probability. 
Let $0 \leq \iota_{0} \leq \iota_{1} \leq  \cdots \leq \iota_{M-1} \leq \iota_{M} = T$ be the endpoints that define the $M$ intervals, $\mathcal{I}_{m} = [\iota_{m-1}, \iota_{m}]$ be the $m$-th interval, and $|\mathcal{I}_{m}| = \iota_{m} - \iota_{m-1}$ be the length of the $m$-th interval.

As the treatment decisions are made at the interval level, the treatment assignments for all times within an interval are the same, i.e., 
\[W_{t} = W_{t^\prime}, \qquad \text{for all } t, t^\prime \in \mathcal{I}_{m}, ~~ \text{for all $m$}. \]
Any design that allows for varying treatment assignments across intervals is a switchback design \citep{bojinov2020design} and is the main focus of this paper. We will discuss switchback designs in more detail in Section \ref{subsec:design}.

In this paper, we directly analyze the event stream data, which is usually the most granular data available in an experiment. The use of granular data can provide more information for the analysis and design of switchback experiments. For example, we might be interested in the event where a rider opens the app and checks the price.

We suppose there are $n$ events occurring in the marketplace between time $0$ and time $T$. 
For each event $i$, the occurrence time $t_i$ is \emph{random} and, in the continuous time framework, can take any real value between $0$ and $T$. Let the outcome of event $i$ be $Y^{(i)}$.
For example, $Y^{(i)}$ could be a binary variable indicating whether the rider requests a ride or not. Let $f(t): [0,T]  \rightarrow \R^+$ be the density function from which events are sampled. We assume that $f(t)$ is bounded from below and from above for all $t$. For simplicity, it is possible to consider the uniform event density as in Example \ref{example:uniform-density}. However, in many realistic settings, the density of events will exhibit periodic patterns due to the seasonality of human behavior. For instance, in ride-hailing, many ride requests occur during commute times, and relatively few occur during the late evening on weeknights. For these settings, a periodic event density, as in Example \ref{example:period-density}, may be a more appropriate choice.

\begin{example}[Uniform event density]\label{example:uniform-density}
    If events are equally likely to occur at any time in the experiment, then $f(t) = 1/T$ for all $t \in [0, T]$.
\end{example}

\begin{example}[Periodic event density]\label{example:period-density}
    If event density has a periodic pattern,
     then a periodic function, such as $f(t) = a_1 \sin(a_2 t + a_3) + a_4$ for some constants $a_1$, $a_2$, $a_3$ and $a_4$ and $t \in [0, T]$, could capture the periodic event density. See an illustration in Figure \ref{fig:design-density}.
\end{example}

Besides the event outcome, we additionally define the marketplace outcome at time $t$ as $Y_t$. The marketplace outcome $Y_t$ can be viewed as the average outcome of all users in the marketplace, such as the average request rate at time $t$. Then the event outcome is a noisy measurement of the marketplace outcome, i.e., for all $i$, 
\[Y^{(i)} = Y_{t_i} + \varepsilon^{(i)} \, , \]
where the measurement error $\varepsilon^{(i)}$ has mean zero.
For example, when $Y^{(i)}$ is binary indicating whether rider $i$ requests a ride, we can model $Y^{(i)}$ as a random draw from the Bernoulli distribution with probability $\P(Y^{(i)} = 1) = Y_{t_i}$ of being one. Then $\varepsilon^{(i)}$ measures the difference between the binary outcome and the probability, i.e., $\varepsilon^{(i)}$ takes a value between $1-Y_{t_i}$ and $-Y_{t_i}$.

Importantly, measurement errors of events that are close in time can be correlated:
\[\cov(\varepsilon^{(i)}, \varepsilon^{(j)}) \neq 0 \qquad \text{for }~ t_i \neq t_j \, . \]
The correlation can be caused by external factors like weather, supply conditions, and traffic. This correlation creates a nuisance dependence between event outcomes, which can affect the resulting variance of treatment effect estimates.

We further account for the possibility that other decision makers may run experiments simultaneously on the same market to test the effect of other interventions. 
We refer to the experiment that the primary decision maker runs as the primary experiment. For the simplicity of exposition,
suppose only one experiment is run simultaneously in addition to the primary experiment in the main text. However, our theoretical results are shown for any number of simultaneous experiments. The insights for one and multiple simultaneous experiments are generally the same. 

Let the treatment design of the simultaneous experiment be $\bm{W}^\s$, where $\bm{W}^\s = \{W_t^\s \in \{0, 1\},\, \forall t \in [0,T]\}$. We assume that the treatment designs of the simultaneous experiment are chosen independently of the primary experiment. We primarily focus on the case where the primary decision maker is {\it agnostic} to the treatment decisions of the simultaneous experiment. We further make 
the non-anticipating outcome assumption, i.e., the outcome at time $t$ is only affected by the treatment assignments up to time $t$ (\cite{basse2019minimax,bojinov2020design} among others). We use $\bm{W}_t = \{W_u,~ \forall u \in [0, t]\}$ and $\bm{W}^\s_{t} = \{W_u^\s, ~\forall u \in [0, t]\}$ to denote the treatment assignments of the primary and simultaneous experiments up to time $t$. 

Then, accounting for the simultaneous experiment and non-anticipating outcomes, we define the potential outcomes of the marketplace at time $t$ as
\[ Y_t(\bm{w}_t, \bm{w}^\s_{t}) \, , \]
where $\bm{w}_t$ is a realization of $\bm{W}_t$ and $\bm{w}^\s_{t}$ is a realization of $\bm{W}^\s_{t}$.\footnote{Suppose both the primary and simultaneous interventions are not applied to times outside of the experiment, i.e., $w_t = 0$ and $w^\s_{t} = 0$ for $t \not \in  [0,T]$. Therefore, there are no carryover effects from treatments outside of the experiment, $ \R \backslash [0,T]$, to outcomes that occurred within the experiment, $[0,T]$. It is then reasonable to define potential outcomes only using treatment assignments within the experiment. } In this definition, potential outcomes are indexed by the treatment assignments of both primary and simultaneous experiments.
The marketplace outcome satisfies $Y_t = Y_t(\bm{W}_t, \bm{W}^\s_{t})$. Given treatment designs $\bm{W}_t$ and $ \bm{W}^\s_{t}$ and event occurrence time $t_i$, there is no randomness in $Y_{t_i}$ anymore, and the randomness in event outcome $Y^{(i)}$ purely comes from the measurement error $\varepsilon^{(i)}$.

Note that the definition above generalizes the standard, binary definition of potential outcomes under the stable unit treatment value assumption (SUTVA) in two aspects. First, this definition allows potential outcomes to be jointly affected by the primary and simultaneous interventions. Second, this definition allows for temporal interference, i.e., the potential outcome of $t$ is not only affected by the treatment status at $t$ but also the treatment assignments at earlier times. 

\begin{remark}[Marketplace and event outcomes]
    Here, we define the potential outcomes at the market level, while observations are at the event level. This starkly contrasts prior work \citep{bojinov2020design,hu2022switchback} where both potential and observed outcomes are at the market level. Note that we can preprocess event outcomes by averaging them by minute (or hour) to obtain observed marketplace outcomes. If we only analyze the observed marketplace outcomes, then our framework is the same as prior work. However, we choose to directly analyze event outcomes, because some factors affecting the estimation error, such as periodic event density and correlation in measurement errors, can be fleshed out and considered in switchback designs.
\end{remark}

\begin{remark}[Number of aggregate units]
    We study the design and analysis of switchback experiments on one aggregate unit (e.g., market), that is the same as \cite{bojinov2020design} and \cite{hu2022switchback}. If the experiments are run on multiple markets, then we can choose the design for each market separately and independently of the designs for other markets.
\end{remark}

\subsection{Estimands}\label{subsec:estimand}
Post-experiment, the primary decision maker uses the observed event outcomes $\{Y^{(i)}\}_{i\in \{1,\cdots, n\} }$ and treatment assignments $\bm{W}$ to estimate the effect of the primary intervention. The estimated effect will then be used to decide whether to deploy the intervention indefinitely. The estimand of primary interest for making this decision is 
the {\it global average treatment effect} (GATE), which measures the difference in average outcomes over time when an intervention is deployed indefinitely (global treatment) versus when it is absent (global control). 
The GATE is formally defined as
\[\delta^\gate  = \int \delta^\total_{t}  f(t)  dt \, , \]
which is the average of the total treatment effect $\delta^\total_{t}$ at time $t$ weighted by the event density $f(t)$. The total treatment effect $\delta^\total_{t}$ at time $t$ is defined as
\[  \delta^\total_{t} = Y_{t}(\bm{W}_t = \bm{1}_t, \bm{W}_t^\s = \bm{0}_t)  - Y_{t}(\bm{W}_t = \bm{0}_t, \bm{W}_t^\s =  \bm{0}_t)  \, ,\]
where $\bm{1}_t$ and $\bm{0}_t$ denote the marketplace being in the treatment and control state for a time duration of $t$, respectively. 
In the definition of GATE, simultaneous intervention is held in the global control state. This definition makes sense when the primary decision maker is interested in the effect of the primary intervention, while holding other conditions as the status quo. In cases where the primary decision maker is aware of the other intervention and wants to condition them in the treatment state, we can consider alternative definitions of $\delta^\gate$ and analyze them analogously.

Besides GATE, the decision maker may also want to estimate the \emph{cumulative effect}, which is the total effect from current and past treatments.
Formally, the cumulative effect is defined as the treatment effect at time $t$ when the treatment is employed from time $t-\Delta t$ to $t$
\[    \delta^\cum_{t}(\Delta t) = Y_{t}(\bm{W}_t = 
 (\bm{0}_{t-\Delta t}, \bm{1}_{\Delta t}), \bm{W}_t^\s = \bm{0}_t)  - Y_{t}(\bm{W}_t = \bm{0}_t, \bm{W}_t^\s = \bm{0}_t)  \, ,\]
where the notation $(\bm{0}_{t-\Delta t}, \bm{1}_{\Delta t})$ concatenates $\bm{0}_{t-\Delta t}$ and $\bm{1}_{\Delta t}$, meaning that the marketplace is in the control state for a duration of $t-\Delta t$ and then in the treatment state for a duration of $\Delta t$. The definition of $\delta^\cum_{t}(\Delta t)$ is visualized in Figure \ref{fig:cumulative-effect}. 
It is easy to see that the cumulative effect $\delta^\cum_{t}(\Delta t)$ converges to the total treatment effect $\delta^\gate_{t}$ as the treatment duration $\Delta t$ grows to infinity, i.e., 
\[\delta^\gate_{t} = \lim_{\Delta t \rightarrow \infty} \delta^\cum_{t}(\Delta t) \, . \]

However, an infinite treatment duration may not be necessary for the cumulative effect to stabilize and converge to the total treatment effect. Gaining insights into both the necessary duration for convergence and the dynamics of the cumulative effect (i.e., the CEC) is valuable for the decision maker when designing and analyzing the switchback experiments.

\begin{figure}[t]
	\centering
\begin{tikzpicture}[domain=0:1.8, xscale = 6, yscale = 1.5]

\draw[->]  (-0.3,0)  --  (1.5,0)  node[right]  {time};
\node(1) at (0.,0) [circle,draw, fill, scale = 0.5 pt]{};
\node(3) at (0.7,0) [circle,draw, fill, scale = 0.5pt]{};
\node(5) at (1.2, 0) [circle,draw, fill, scale = 0.5pt]{};

\node[below of = 1, yshift = 0.4 cm](leftlabel){$0$};
\node[below of = 3, yshift = 0.4 cm](middlelabel){$t-\Delta t$};
\node[below of = 5, yshift = 0.4 cm](rightlabel){$t$};

\draw [decorate,decoration={brace,amplitude=10pt}, yshift = 0.2cm]
(.02,0) -- (.68,0) node (curly_bracket)[black,midway, yshift = 0.7 cm] 
{$\bm{0}_{t-\Delta t}$};
\draw [decorate,decoration={brace,amplitude=10pt}, yshift = 0.2cm]
(.72,0) -- (1.18,0) node (curly_bracket)[black,midway, yshift = 0.7 cm] 
{$\bm{1}_{\Delta t}$};
\node[below right of = 3, xshift = -3 cm, yshift = -0.5
cm](switch){switch from control to treatment};
\draw[latex-] (3) to[out=-135,in=45,looseness=1.] (switch);
\node[above right of = 1, xshift = -1 cm, yshift = 1.2
cm](effect){$\bm{W}_t = (\bm{0}_{t-\Delta t}, \bm{1}_{\Delta t})$};
\node[above right of = 5, xshift = -1 cm, yshift = 1.2
cm](effect){cumulative effect is $\delta^\cum_{t}(\Delta t)$};
\draw[->]  (1.2,1.)  --  (1.2,.2)  node[above]  {};

\end{tikzpicture}
\caption{An illustration of cumulative effect $\delta^\cum_{t}(\Delta t)$ at time $t$ after being treated for a duration of $\Delta t$ in the primary experiment, while holding the simultaneous intervention in the control state ($\bm{W}^s_t = \bm{0}_t$). When $\Delta t$ grows to infinity, $\delta^\cum_{t}(\Delta t)$ converges to $\delta^\gate_t$. }
\label{fig:cumulative-effect}
\end{figure}
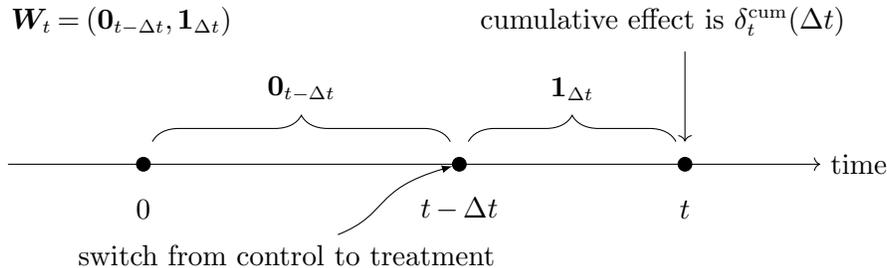

    \subsection{Post-Experiment Estimation}\label{subsec:estimation}

 In this paper, we primarily focus on analyzing the widely used Horvitz-Thompson (HT) estimator for the treatment effect \citep{horvitz1952generalization}. The HT estimator of $\delta^\gate$ using the $n$ observed event outcomes and treatment design takes the form of

		\begin{equation}\label{eqn:ht-estimator}
  \hat{\delta}^\gate = \frac{1}{n} \sum_{i=1}^n \left[\frac{W_{t_i}  Y^{(i)} }{\pi} -  \frac{(1-W_{t_i}) Y^{(i)} }{1-\pi} \right] \, ,
    \end{equation}
    where 
    \[\pi =  \int_{t \in [0, T]} \+E[W_{t}]  f(t) d t \]
    is the fraction of times in the treatment state. If time $t$ has 50\% probability of being treated, then $\P(W_{t} = 1)  = 1/2$ and $\pi = 1/2$. If an event occurs in a treated interval, then its outcome is used to estimate the average outcome under global treatment; otherwise, its outcome is used to estimate the average outcome under global control.

    We focus on the HT estimator instead of alternative treatment effect estimators (discussed in Section \ref{subsec:alternative-estimator}) for four reasons. First, the implementation of the HT estimator is simple. It relies neither on assumptions about carryover and correlation mechanisms of event outcomes, nor on the knowledge of treatment assignments of simultaneous interventions. Second, as shown in our case study, the HT estimator generally performs better than alternative estimators. Third, the HT estimator is simpler to analyze than alternative estimators while still conveying important design insights to improve estimation efficiency. Fourth, our case study shows that these design insights can not only reduce estimation errors for the HT estimator but also for alternative estimators.

	\subsection{Empirical Bayes Switchback Design}\label{subsec:design}

 Before the experiment starts, the decision maker chooses the number of intervals $M$ and the interval switching points $0 \leq \iota_{1} \leq \cdots \leq \iota_{M-1} \leq T$, aiming to reduce the estimation error of GATE, post-experiment. An important metric for the decision maker is the MSE of $\hat{\delta}^\gate$, defined as
 \begin{align}\label{eqn:mse-expression}
     \+E_{W,\varepsilon,t}\Big[ \big(\hat{\delta}^\gate - \delta^\gate\big)^2\Big] \, ,
 \end{align}
 where the expectation is taken with respect to the treatment designs $\bm{W}$ and $\bm{W}^\s$, the measurement errors in event outcomes $\varepsilon^{(1)}, \cdots, \varepsilon^{(n)}$, and the event occurrence times $t_1, \cdots, t_n$. Here we focus on the randomized designs, where each time period is equally likely to be treated or untreated, i.e., $\P(W_{t} = 1)  = 1/2$ and $ \P(W^\s_t = 1) = 1/2$ for all $t$.\footnote{We focus on the $1/2$ probability because this probability is shown to be efficient in many design problems.}

 In Section \ref{sec:analysis}, we provide the expression of MSE as a function of the interval endpoints, which essentially applies to any switchback design. The expression shows that the MSE depends on assumptions on carryovers, outcome covariance, periodicity, and simultaneous interventions.

 We propose a new \emph{empirical Bayes} approach for switchback design that leverages historical data to account for these four factors in choosing a design for a new experiment.
 This approach has four phases, as illustrated in Figure \ref{fig:empirical-bayes}. In phase one, we analyze historical data and model the data-generating process. Specifically, the event density is estimated and then can be used as the input for the design. The periodic pattern is analyzed and then can be used in design to balance the periodicity. In phase two, we estimate an empirical distribution of CECs from prior experiments. Although the CEC of the new intervention is not known ex-ante, a reasonable prior is the empirical distribution of CECs. In phase three, we choose a design by comparing the performance of candidate designs through synthetic experiments. These experiments are run on historical data, and the CEC of the synthetic intervention is randomly drawn from the empirical distribution of CECs. Such experiments ensure that all factors affecting the MSE are present. In phase four, we select the best-performed design in synthetic experiments. The selected design is referred to as the \emph{empirical Bayes} design.

 In this paper, we focus on empirically evaluating the following three types of heuristic designs, in which the decision maker only needs to choose two parameters, and the design problem is more tractable. The comparison of these three designs fleshes out the important design insights to improve estimation efficiency. Figure \ref{fig:design-density} illustrates the switching points of the three designs under the periodic event density. 
 The first type is the fixed duration switchback (Example \ref{example:fixed-duration-switchback}), which has constant interval lengths and is usually the design used in practice.

 \begin{figure}[t!]
	\centering
 \resizebox{1\columnwidth}{!}{%
\begin{tikzpicture}
    \begin{scope}
        \draw[very thick, MyTriangle!40] (0,-0.03) -- (20,-0.03);
        \draw[very thick, MyTriangle!55] (0,-1.5) -- (20,-1.5);
        \draw[very thick, MyTriangle!70] (0,-3) -- (20,-3);
        \draw[very thick, MyTriangle!85] (0,-4.5) -- (20,-4.5);
        \draw[very thick, MyTriangle!85] (0,-6) -- (20,-6);
        \draw[very thick, MyTriangle] (20,-0.02) -- (20,-6);
        \draw[very thick, MyTriangle] (0,-0.02) -- (0,-6);
        \filldraw[very thick,white,fill=MyTriangle!45] (20,0) -- (0,0) --  (0,-6) -- (20,-6);
        \filldraw[very thick,white,fill=MyTriangle!70] (20,0) -- (0,0) --  (0,-4.5) -- (20,-4.5);
        \filldraw[very thick,white,fill=MyTriangle!85] (20,0) -- (0,0)  -- (0,-3) -- (20,-3);
        \filldraw[very thick,white,fill=MyTriangle] (20,0) -- (0,0) -- (0,-1.5) -- (20,-1.5);
        \node[text width=20cm] at (11,-0.75) {\large $\bullet$ Phase 1: Analyze historical data and model data generating process};
        \node[text width=20cm] at (11,-2.25) {\large $\bullet$ Phase 2: Estimate the empirical distribution of CECs from prior experiments};
        \node[text width=20cm] at (11,-3.75) {\large $\bullet$ Phase 3: Run synthetic experiments on historical data to compare candidate designs};
        \node[text width=20cm] at (11,-5.25) {\large $\bullet$ Phase 4: Choose the best-performed design in synthetic experiments };
    \end{scope}

\end{tikzpicture}
}
\caption{Empirical Bayes approach for switchback designs. See Section \ref{subsec:design} for more details.  }
\label{fig:empirical-bayes}
\end{figure}

\begin{example}[Fixed duration switchback]\label{example:fixed-duration-switchback}
    The first interval starts at time $\iota_{0} = q$ for an offset parameter $q < T/M$. The length of all the intervals beside the last one is $p = T/M$. The endpoints are then equal to $\iota_{m} = m \cdot p + q$ for all $m$. 
\end{example}

The second type is the Poisson duration switchback (Example \ref{example:poisson-duration-switchback}), where the interval length is random and generated by the Poisson distribution. The Poisson duration switchback has similar, but slightly different, interval lengths compared to the fixed duration switchback. We study this design motivated by our finding in the MSE decomposition that randomizing interval lengths can reduce the error from the confounding effect of simultaneous interventions.

\begin{example}[Poisson duration switchback]\label{example:poisson-duration-switchback}
    The first interval starts at time $\iota_{0} = q$. The length of each interval $\iota_{m} - \iota_{m-1}$ is randomly drawn from the Poisson distribution with the mean parameter $\lambda = T/M$. We sum the lengths of the first to the $m$-th intervals to obtain the value of the endpoint $\iota_{m}$.\footnote{If the endpoints of some intervals are bigger than the experiment duration (i.e., there exists an $\bar{M}$ such that $\iota_{m^\prime} > T$ for $m^\prime \geq \bar{M}$), then we set the endpoints of these intervals to $T$ (i.e., set $\iota_{m^\prime}$ to $T$ for $m^\prime \geq \bar{M}$) and the lengths of these intervals are zero. }
\end{example}

The third type is the change-of-measure switchback (Example \ref{example:change-of-measure-design}), which has constant interval lengths after changing the measure of nonuniform event density to uniform density. Under the uniform event density, the change-of-measure switchback is the same as the fixed duration switchback. However, under the periodic density, the change-of-measure design has much shorter interval lengths in times of high density and much longer interval lengths in times of low density than in the other two designs. We can determine the interval lengths in change-of-measure switchback using empirical density estimated from historical data. We study this design because we find that such a pattern of interval lengths can reduce an important variance term in MSE.

\begin{example}[Change-of-measure switchback]\label{example:change-of-measure-design}
    The first interval starts at time $\iota_{0} = q$ for some $q$ that satisfies $\int_0^q f(t) dt < 1/M$. For the remaining endpoints, they are chosen in a way that the event occurrence probability is the same across intervals, that is, $\int_{\iota_{m}}^{\iota_{m+1}} f(t) dt = 1/M$. 
\end{example}

 \begin{figure}[t!]
	\centering
  \includegraphics[width=1\linewidth]{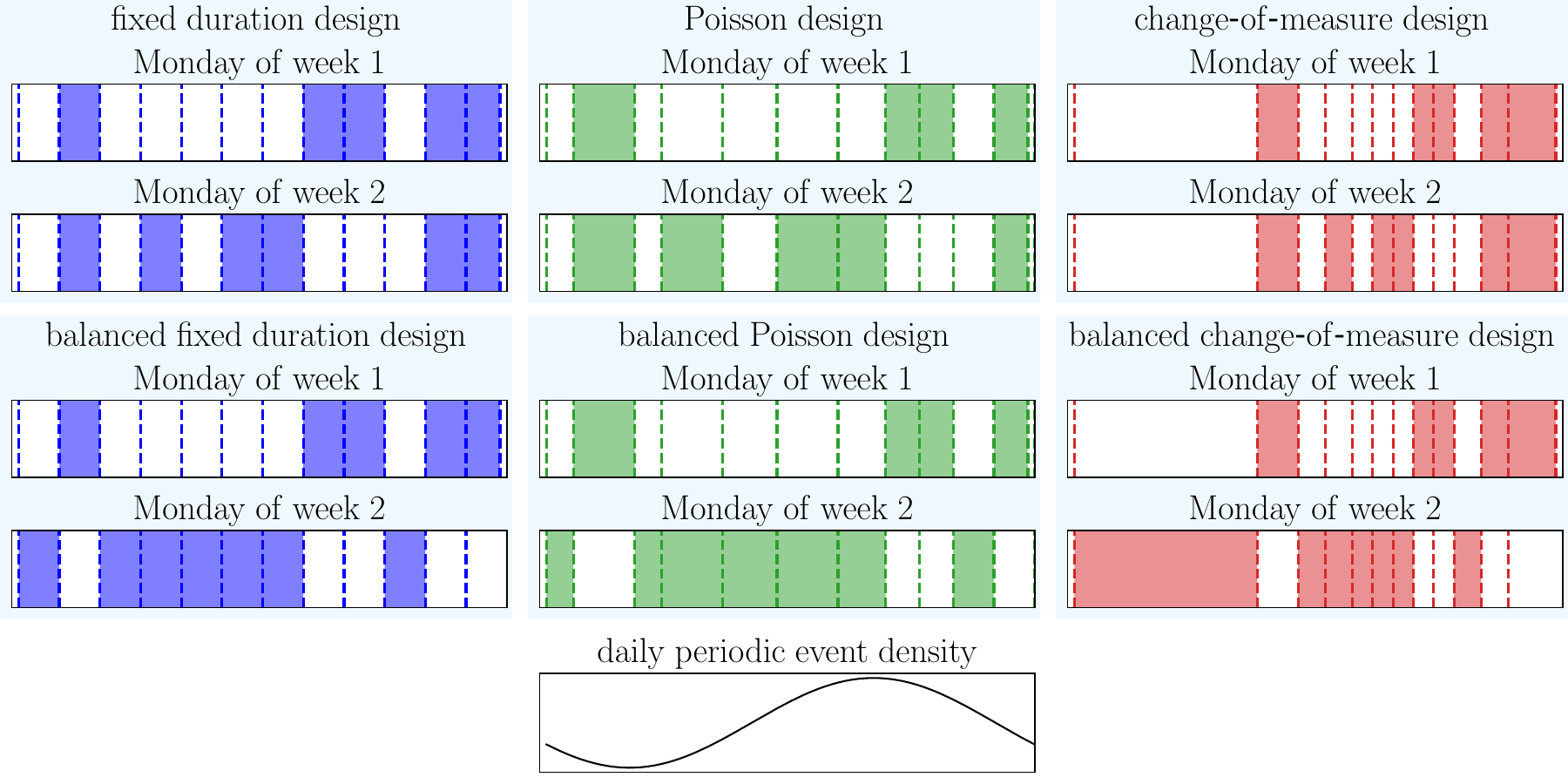}
	\caption{Illustration of various designs and daily periodic event density (dash lines are switching points, and treated intervals are shaded) on Mondays in a two-week experiment. The design on other days in a week is analogous. If a design is not balanced, the design on each day is independent of the design on other days; otherwise, the design in the second week mirrors the design in the first week, and in the first, the design on each day is independent of the design on other days.
 }
	\label{fig:design-density}
\end{figure}

In addition, we consider the following balanced design that imposes restrictions on the randomness of treatment assignments. This balanced design is relevant only if we explicitly account for the periodicity in the analysis and design, which is a key distinction from prior work on switchback designs. We study this design because, based on the MSE decomposition, we find that balancing is very effective for variance reduction, especially when the relative size of the treatment effect compared to noise is small.

\begin{example}[Balanced randomized design]\label{example:balanced-design}
The balanced design balances the heterogeneity in $f(t)$ and $Y_t(\bm{w}_t, \bm{w}^\s_t)$ between control and treated intervals. In a special case where $f(t) = f(t+T/2)$ and $Y_t(\bm{w}_t, \bm{w}^\s_t) = Y_{t + T/2}(\bm{w}_{t+T/2}, \bm{w}^\s_{t+T/2})$ for $t\leq T/2$, the switchback design is balanced if the treatment assignment at time $t$ is opposite to that at time $t - T/2$, that is, $W_t = 1 - W_{t - T/2}$.  See a visualization of such a balanced design in Figure \ref{fig:design-density}.
\end{example}

In the example of a two-week experiment, when a balanced randomized switchback is used, the treatment assignments of the second week are opposite to those of the first week. As the heterogeneities in the outcomes tend to have a periodic pattern by week, the balanced design in a two-week experiment creates matched pairs for the same time in a week.\footnote{Balancing the periodicity is found effective in reducing variance in the switchback experiments via multiple randomization designs \citep{masoero2023efficient} and nonstationary a/b tests \citep{wu2022non}. 
} We note that the balancing approach in Example \ref{example:balanced-design} is used in practice and studied in this paper for illustration purposes. One could consider more complex balanced designs, such as balancing the treatment assignments of two consecutive days, hours, etc., while the insights would generally be the same.

%% file: section_3.tex
\section{A Case Study on a Ride-Sharing Platform}\label{sec:empirical}

In this section, we analyze historical data from a ride-sharing platform and explore strategies for designing more efficient experiments. We have access to two sets of historical data. The first data set consists of the event-level data of the top 50 regions between June 2022 and March 2023, referred to as the historical control data hereafter. The second data set includes the event-level data of a large corpus of experiments conducted between June 2021 and March 2023, hereafter referred to as the historical experimental data. In both data sets, each event represents a rider session started from opening the app. The outcome is binary, denoting whether the rider requested the ride ($Y_{t_i} = 1$) or not ($Y_{t_i} = 0$). 

We illustrate our proposed empirical Bayes design approach on the historical data to select a switchback design. By comparing the performance of designs introduced in Section \ref{subsec:design}, we identify a hierarchical structure of the effectiveness of design principles in reducing the estimation error.

\subsection{Analysis of Historical Data}\label{subsec:analysis-historical-data}

In this subsection, we show the estimated event density, mean control outcomes, and variance of measurement errors from historical control data. These estimates demonstrate the behavioral changes based on the time of day and day of the week. We also show the empirical distribution of CECs estimated from prior experiments. 

\subsubsection{Estimates from Historical Control Data}
\paragraph{Event density}

The estimated event density $f(t)$ for every minute of a week is shown in Figure \ref{fig:event-density-error}. 
There are two main observations. First, the event density has a periodic pattern, with high density during the peak hours, such as 6 PM, and low density during the off-peak times, such as 3 AM. Second, the event density is higher during peak hours on weekends (Fridays, Saturdays, and Sundays) than during peak hours on weekdays.

\paragraph{Global control outcome}

The estimated and standardized global control outcome (i.e.,  $Y_t(\bm{0}_t, \bm{0}_t)$ subtracted by its mean and divided by its standard deviation) for every minute in a week is shown in Figure \ref{fig:event-density-error}. The average global control outcome has a periodic pattern and is generally higher during the daytime. 

\paragraph{Heteroskedastic measurement errors}

The estimated and standardized standard error of measurement errors for every minute in a week is shown in Figure \ref{fig:event-density-error}. We estimate the standard errors by assuming that the binary event outcome is randomly drawn from a Bernoulli distribution, with the probability of getting one being the average conversation rate (i.e., global control outcome). The standard error is heteroscedastic and has a periodic pattern, similar to that of the event density and global control outcome. Notably, the standard error of measurement errors tends to be lower in times of high event density, and higher in times of low event density. This is because, during periods of high event density, the mean control outcome is generally larger and lies in a regime where the variance is negatively correlated with the mean control outcome.

\begin{figure}[t!]
    \centering
    \includegraphics[width=1.\textwidth]{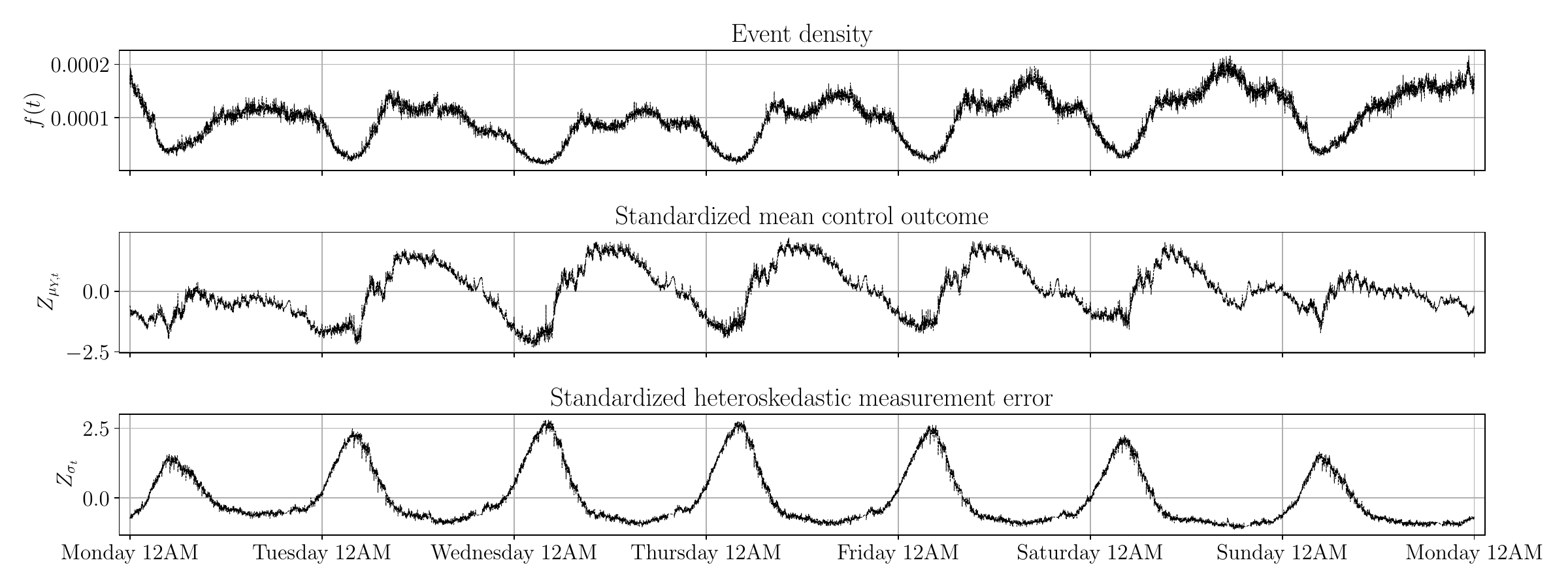}
    \caption{Event density, standardized mean control outcome (denoted by $Z_{\mu_{Y,t}}$), and standardized heteroscedastic measurement error  (denoted by $Z_{\sigma_{t}}$) from Monday 12 AM to Sunday 11:59 PM. 
    }
    \label{fig:event-density-error}
\end{figure}

\subsubsection{Empirical Distribution of CECs from Prior Experiments}
We analyze a large corpus of the experiments run between June 2021 and March 2023 to construct a prior on CEC. We have $149$ two-week experiments run across $114$ markets. Different experiments were run in different subsets of the markets. There are $890$ distinct experiment-market pairs in total.
These experiments employ a balanced, fixed-duration switchback design with a constant interval length of $56$ minutes (referred to as status quo design). 

We first estimate a CEC for each experiment-market pair by taking the difference in outcomes between treated and control intervals independently for every minute since the switch. Then, we obtain a $56$-dimensional CEC vector whose $j$-th entry is the average cumulative effect given that the market has been treated for $j$ minutes.

Figure \ref{fig:fitted-poly2} shows several representative estimated CECs. A notable observation is that the estimated cumulative effects vary substantially minute-to-minute with the treatment duration, indicating a high variance in the treatment effect estimation. Therefore, to reduce variance, we smooth out the estimated CECs by fitting a natural cubic spline to the vector of points. The natural cubic splines are particularly suitable for incorporating our priors on the shape of CECs, including a low tendency of cumulative effects to vary minute-to-minute with treatment duration, and the convergence of cumulative effects to GATE at the right boundary. In Appendix \ref{subsec:curve-fitting-details}, we document and interpret the constraints imposed in curve fitting. The smoothed curves by natural cubic splines are overlaid on the (raw) estimated CECs in Figure \ref{fig:fitted-poly2}. The $890$ smoothed CECs constitute the \emph{empirical distribution of CECs}.

\begin{figure}[t!]
    \centering
    \includegraphics[width=\textwidth]{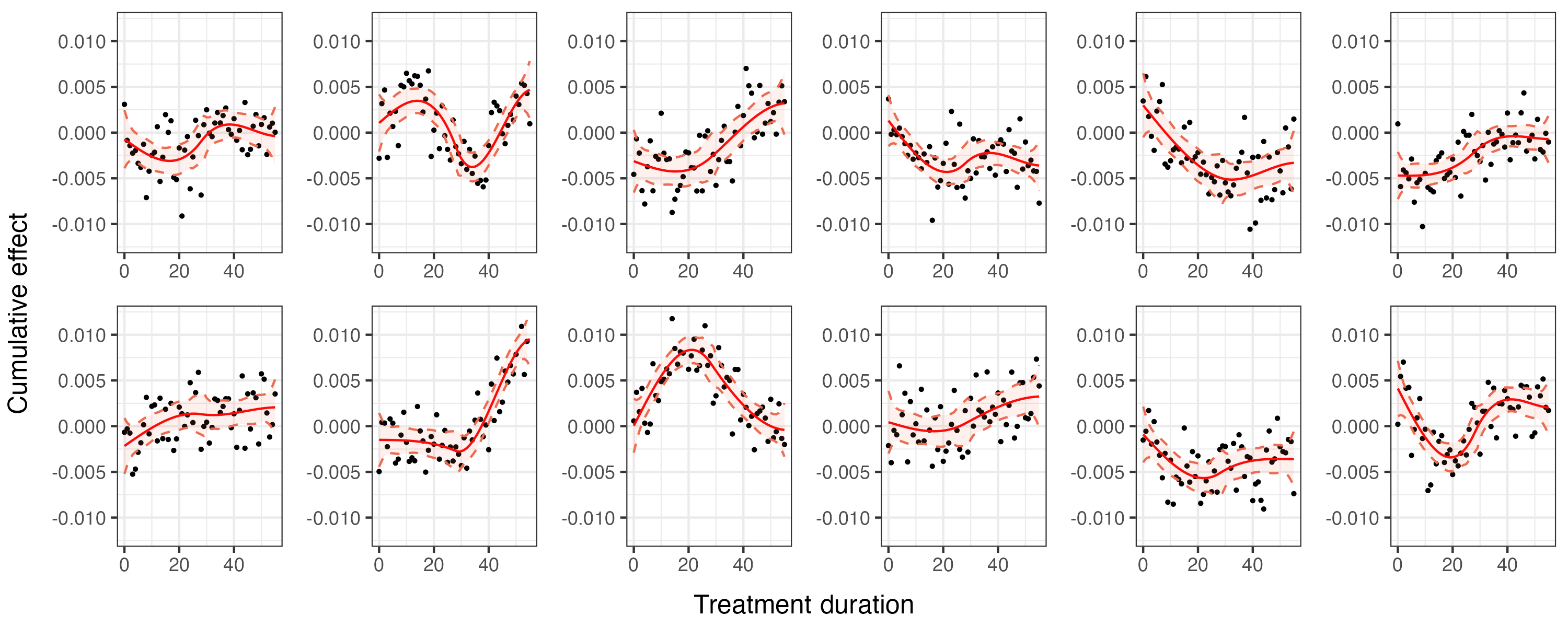}
    \caption{$12$ representative estimated CECs for the treatment duration $\Delta t$ of $\{1, \cdots, 56\}$ minutes (in black dots) and their smooth curves by natural cubic splines (in red). The natural cubic splines have zero gradients at the right boundary. Most CECs are non-monotonic in the treatment duration and can switch signs. The shaded area for each curve shows the $95\%$ confidence interval of the cumulative effect given a treatment duration and the natural cubic spline as the smoothing method.
    }
    \label{fig:fitted-poly2}
\end{figure}

Two important observations arise from the smoothed CECs. First, for $68\%$ of the smoothed CECs, cumulative effects change signs as the treatment duration changes. This suggests that the treatment may first be effective in increasing the conversion rate and then become ineffective. The opposite case, where the treatment is first ineffective and then becomes effective, is also likely to happen. Second, in most cases, the cumulative effects do not vary monotonically with the treatment duration. For example, the cumulative effects can first increase and then decrease, or vice versa, as the treatment duration increases. Overall, it often takes quite a bit of time for the cumulative effects to converge to GATE. 
Because of these observations, we further consider more persistent carryovers and fit the $112$-minute CECs, as shown in Figure \ref{fig:112-various-poly} in Appendix \ref{appendix:empirical}. The two observations for the $56$-minute curves continue to hold for the $112$-minute curves.

To validate our model fitting procedure, we compare the cross-validation error of smoothing by the natural cubic spline with smoothing by alternative methods, such as polynomial regression and local regression. We document our cross-validation procedure in Appendix \ref{subsec:cross-validation}, with the results shown in Figure \ref{fig:model-comparison} in Appendix \ref{appendix:empirical}. The cross-validation error of the natural cubic spline is comparable to all the alternative methods. As the natural cubic spline is well suited to incorporate our priors on the curve shape and can capture complex curve shapes, we focus on it in the main text. See Appendix \ref{subsec:cross-validation} for more discussion. As a robustness check, we have run synthetic experiments using the smoothed effects by alternative methods, and the best-performed design is generally robust to the choice of a flexible smoothing method.

\subsection{Empirical Bayes Design through Synthetic Experiments}\label{subsec:synthetic-experiment}

We illustrate our proposed empirical Bayes approach to choosing a switchback design by running synthetic experiments on historical data. Decision makers can use a similar approach to select specific switchback designs tailored to their specific settings.

\subsubsection{Setup of Synthetic Experiments}\label{subsec:setup-syn-exp}

We consider the following six switchback designs to run two-week synthetic experiments: fixed duration (FD), balanced fixed duration (bal. FD), change-of-measure (CM), balanced change-of-measure (bal. CM), Poisson duration (Poisson), and balanced Poisson duration (bal. Poisson) switchback designs. For each design, we vary the average interval length across three specifications: $28$, $56$, and $112$ minutes. In the Poisson duration switchback, the average interval length corresponds to the mean parameter $\lambda$. In the change-of-measure design, the average interval length equals to the event occurrence probability in an interval multiplied by the experiment duration. 

We conduct two types of synthetic experiments. In the first type, no other experiment is run simultaneously with the primary synthetic experiment. To implement this type, we randomly draw a market from the $50$ markets in the historical control data and then randomly draw two consecutive weeks of historical control data in this market. In the second type, one experiment starts and ends simultaneously with the primary synthetic experiment. To implement the second type, we randomly select one experiment-market pair and use its two-week historical experimental data. In both types, we assume that the primary synthetic intervention has not been applied to the historical (control or experimental) data.

Next, we randomly draw one CEC from the empirical distribution of CECs, and assume that the cumulative effects of the synthetic intervention follow this drawn curve. Then we use this curve to generate the two-week synthetic experimental data. Specifically, given a switchback design and every time point, we calculate the total effect from treatments received on and prior to that time point. Then we add the total effect to the two weeks of historical data drawn from the previous step, and obtain the synthetic experimental data. 

Finally, we apply the HT estimator to the synthetic experimental data to estimate GATE and compute the estimation error of GATE using the drawn curve. We repeat this procedure and obtain the estimation error for each of the six switchback designs with each of the three average interval lengths, given a draw of two weeks of historical data and a CEC. We repeatedly draw the two-week historical data and CEC for $500$ times in total. Then for each switchback design with each average interval length, we obtain $500$ estimation errors in total.

\begin{figure}[t!]
    \centering
    \includegraphics[width=\textwidth]{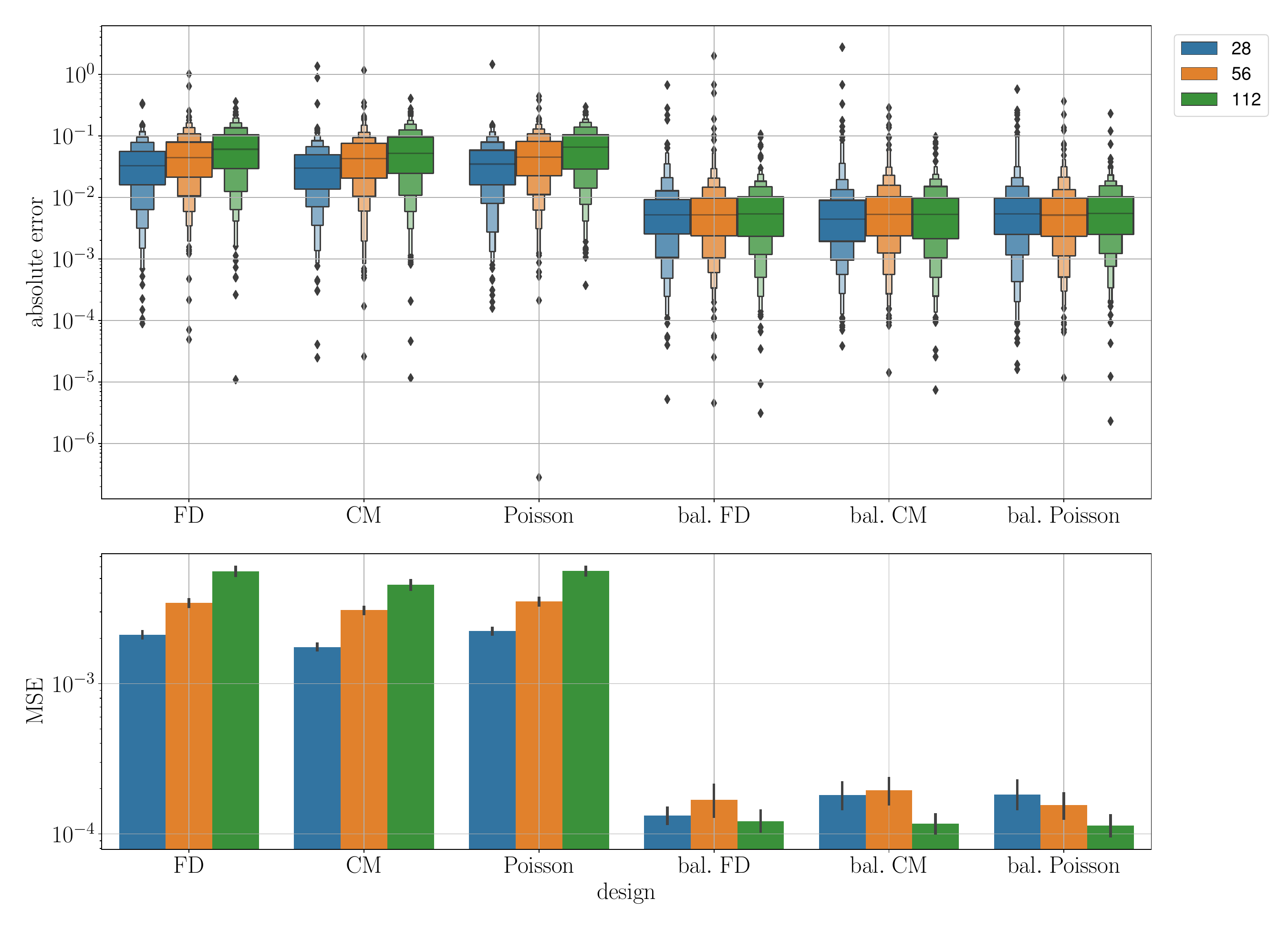}
    \caption{Boxenplot of the absolute estimation error and MSE from two-week synthetic experiments using various switchback designs. The boxenplot is obtained based on $500$ draws of two-week historical data (serving as the observed data absent of synthetic intervention) and a CEC (serving as the treatment effects of the synthetic intervention). In this simulation, one experiment starts and ends simultaneously with the primary synthetic experiment. The corresponding results when no experiments start and end simultaneously with the primary synthetic experiment are shown in Figure \ref{fig:56min-mse} in Appendix \ref{appendix:empirical}. The differences in results between when one versus no experiments run simultaneously are demonstrated in Figure \ref{fig:56min-diff-mse}  in Appendix \ref{appendix:empirical}. The legend denotes the average interval length, and  ``bal.'' is an abbreviation of ``balanced''.  
    }
    \label{fig:syn-exp-fitted-poly2}
\end{figure}

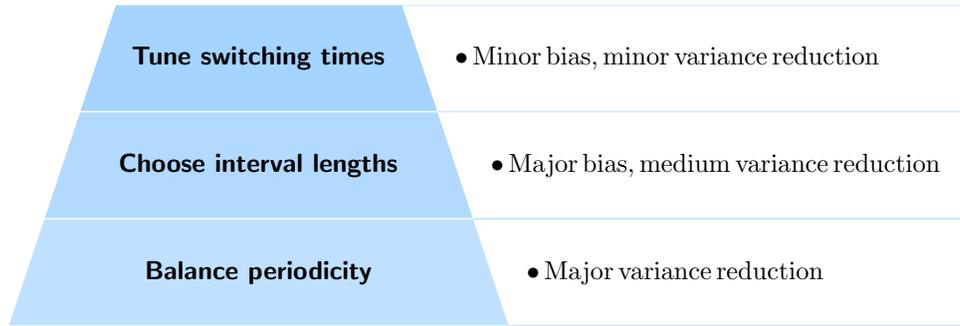
\begin{figure}[t!]
	\centering
 \resizebox{0.8\columnwidth}{!}{%
\begin{tikzpicture}
    \begin{scope}
        \draw[very thick, MyTriangle!40] (0,-0.03) -- (20,-0.03);
        \draw[very thick, MyTriangle!55] (0,-3) -- (20,-3);
        \draw[very thick, MyTriangle!70] (0,-6) -- (20,-6);
        \draw[very thick, MyTriangle!85] (0,-9) -- (20,-9);
        \draw[very thick, MyTriangle] (20,-0.02) -- (20,-9);
        \filldraw[very thick,white,fill=MyTriangle!70] (4,0) -- (-4,0) --  (-7,-9) -- (7,-9);
        \filldraw[very thick,white,fill=MyTriangle!85] (4,0) -- (-4,0)  -- (-6,-6) -- (6,-6);
        \filldraw[very thick,white,fill=MyTriangle] (4,0) -- (-4,0) -- (-5,-3) -- (5,-3);
        \node at (0,-1.5) {\bfseries\sffamily\scshape\LARGE {Tune switching times} };
        \node at (0,-4.5) {\bfseries\sffamily\scshape\LARGE Choose interval lengths};
        \node at (0,-7.5) {\bfseries\sffamily\scshape\LARGE Balance periodicity};
        \node[text width=13cm] at (12,-1.5) {\LARGE $\bullet$ Minor bias, minor variance reduction};
        \node[text width=13cm] at (13,-4.5) {\LARGE $\bullet$ Major bias, medium variance reduction};
        \node[text width=13cm] at (14,-7.5) {\LARGE $\bullet$ Major variance reduction};
    \end{scope}

\end{tikzpicture}
}
\caption{Hierarchical structure in the effectiveness of design principles in reducing MSE in the case study on a ride-sharing platform. Balancing periodicity is the most effective. Carefully choosing average interval lengths is moderately effective. Tuning exact switching times is mildly effective.  }
\label{fig:pyramid}
\end{figure}

\subsubsection{Results}\label{subsec:empirical-results}

Figure \ref{fig:syn-exp-fitted-poly2} visualizes the estimation error of GATE for each switchback design with each average interval length when one experiment runs simultaneously. If we use the balanced Poisson duration switchback with an average interval length of $112$ minutes, then the MSE is reduced by $33$\% compared to the status quo fixed-duration switchback design with an interval length of $56$ minutes.  Figure \ref{fig:56min-mse} in Appendix \ref{appendix:empirical} shows the corresponding results with no experiment running simultaneously. Based on the results in Figures \ref{fig:syn-exp-fitted-poly2} and \ref{fig:56min-mse}, we summarize the effectiveness of various design principles in reducing the estimation error into a hierarchical structure, as shown in Figure \ref{fig:pyramid}. 

As shown in Figures \ref{fig:syn-exp-fitted-poly2} and \ref{fig:56min-mse}, balancing periodicity, achieved by using balanced switchback designs in Example \ref{example:balanced-design}, is the most effective in reducing the estimation error. This observation indicates that variance from periodicity is a major source of error and can be reduced using a balanced design.

Carefully choosing the average interval lengths can also effectively reduce the estimation error. The effect is two-fold. On the one hand, when switchback designs are not balanced, fast switching reduces variance and then reduces the estimation error because variance is the major source of error for an unbalanced design. On the other hand, when switchback designs are balanced, long switching can reduce the chance of a large error and make the designs more robust to various shapes of CEC. We also note that when an experiment starts and ends simultaneously, the $56$-minute interval length can have a larger MSE than the $28$-minute one. This is because of the confounding effect of the simultaneous experiment. The reason is supported by Figure \ref{fig:56min-mse} in Appendix \ref{appendix:empirical} that, without the simultaneous experiment, the $56$-minute interval length then has a lower error.

After the average interval lengths are selected, choosing the exact switching times can further reduce the MSE. A simple approach is to choose between fixed duration, change-of-measure, and Poisson duration designs depending on our prior for CEC. Figure \ref{fig:syn-exp-fitted-poly2} shows that the balanced Poisson switchback with long interval lengths has smaller MSE than other balanced designs, when the prior follows the empirical distribution of CECs. The results indicate that randomizing switching times (e.g., through Poisson switchback) can mitigate the confounding effects of simultaneous experiments, which can effectively reduce the estimation error in practice.

We also directly compare results with and without one simultaneous experiment in Figure \ref{fig:56min-diff-mse} in Appendix \ref{appendix:empirical}. The comparison shows that the presence of a simultaneous experiment generally increases the estimation error. The increase is more obvious if the primary and simultaneous experiments use the same $56$-minute interval length. Therefore, when impacts from simultaneous experiments are nonnegligible, it can be important to use designs that mitigate confounding effects. For example, effective approaches include choosing different interval lengths, staggering switching times across experiments, or launching experiments at different times.

Finally, we run synthetic experiments using the $112$-minute cumulative effect curves smoothed by natural cubic splines. As shown in Figure \ref{fig:112min-mse} in Appendix \ref{appendix:empirical}, the Poisson duration switchback with an average interval length of $112$ minutes continues to have the lowest MSE, and the hierarchical structure continues to hold.

%% file: section_4.tex
	\section{Analysis of Switchback Design}\label{sec:analysis}
	
	In this section, we present the precise decomposition of the bias and MSE for any switchback design. The decomposition shows how the four factors affect the MSE and provides the theoretical foundation to explain how various switchback designs affect the MSE.

	In Section \ref{subsec:assumption}, we first lay out the assumptions required for the decomposition. Subsequently, we introduce several interval-level statistics in Section \ref{subsec:block-stats}. These interval-level statistics are then used in the decomposition of bias and MSE in Section \ref{subsec:mse-result}. The decomposition in this section assumes only one simultaneous intervention, but this is only for the purpose of exposition. The general results with multiple simultaneous interventions are stated in Appendix \ref{subsec:multi-simul-exp}, and the proof in Appendix \ref{appendix:proof} allows for multiple simultaneous interventions.

	\subsection{Assumptions}\label{subsec:assumption}

 We begin by assuming that the sampling of events is exogenous and independent of the treatment assignments of both the primary and simultaneous interventions. This assumption holds true for interventions that potential riders cannot discern a difference before opening the app and checking prices, and consequently do not affect the event density, such as changes to pricing or matching algorithms. 

\begin{assumption}[Exogeneity of events]\label{ass:exogeneity}
	Events are sampled randomly and independently from the density function $f(t)$, and $f(t)$ is independent of the treatment assignments of primary and simultaneous experiments, $\bm{W}$ and $\bm{W}^\s$.
\end{assumption}
 
 Moreover, for the purpose of exposition, we impose an assumption on the structure of the treatment effects. This assumption is based on the observation that the treatment effects can be decomposed as
 \begin{align*}
    & Y_{t}(\bm{w}_t, \bm{0}_t) - Y_{t}(\bm{0}_t, \bm{0}_t) = w_t \cdot (\underbrace{Y_{t}(e_{t}, \bm{0}_t)- Y_{t}(\bm{0}_{t}, \bm{0}_t)}_{\delta^\inst_{t}})  +
    \underbrace{Y_{t}(\bm{w}_t, \bm{0}_t) - Y_{t}(e_{t}, \bm{0}_t)}_{\delta^\co_{t}(\bm{w}_t) } \, ,
\end{align*}
where $e_t$ is a one-hot-encoded vector with the entry of time $t$ to be $1$ and all the remaining entries to be $0$, $\delta^\inst_{t}$ is the {\it instantaneous effect} of treatment at time $t$ on the outcome at time $t$, and  $\delta^\co_{t}(\bm{w}_t)$ is the {\it carryover effect} of treatment assignments $\bm{w}_t$ prior to time $t$ on the outcome at time $t$. For notation simplicity, we let $\delta^\co_t \coloneqq  \delta^\co_{t}(\bm{1}_t)$ be the carryover effect under global treatment. Then, for a sufficiently large $t$ and with a sufficiently long time in the treatment state, the total treatment effect has $\delta^\gate_t = \delta^\inst_t + \delta^\co_t$. 

This assumption, as stated in Assumption \ref{ass:interference-structure} below, imposes a structure on carryover effects; that is, they are additive and can be parameterized by a carryover kernel. 
This assumption is for exposition, which can be relaxed at the cost of more complex notations in the main result, while the insights are generally the same. 
	
	\begin{assumption}[Carryover effects]\label{ass:interference-structure}
	    For every $t$, there exists a carryover kernel $d^\co_{t}(t^\prime)$ that measures the intensity of carryover effect from $t^\prime$ to $t$ and satisfies $\int  d^\co_{t}(t^\prime)  f(t^\prime)  d t^\prime  = 1$, such that
	    \begin{align*}
	       \delta^\co_{t}(\bm{w}_t)
	       =&  \delta^\co_{t} \cdot \int w_{t^\prime} \cdot d^\co_{t}(t^\prime)  f(t^\prime) d t^\prime \, .
	    \end{align*}
	\end{assumption}
	
 The carryover kernel $d^\co_{t}(t^\prime)$ can be quite general. Below, we provide two examples of commonly used carryover kernels in practice. In the case of the non-anticipating outcome, $d_{t}(t') = 0$ for all $t' > t$. Moreover, if the treatment assignments can only affect the outcomes for a duration of $h < \infty$ in the future (as assumed in \cite{bojinov2020design,basse2019minimax,xiong2019optimal} among others)), then $d_{t}(t') = 0$ for all $t' < t - h$ and can take arbitrary values for $t - h \leq t^\prime \leq t$. Alternatively, if the duration is infinity, but carryover effects decay at a geometric rate (as assumed in \cite{hu2022switchback}), then $d^\co_{t}(t^\prime)$ decays with $\exp[-(t - t^\prime)]$ for all $t^\prime < t$.

\begin{example}[Uniform carryover kernel]\label{example:uniform-carryover-kernel}
If the carryover intensity is uniform in $t^\prime \in [t-h, t]$, but is zero outside this interval, then $d^\co_{t}(t') \propto 1/h$ for all $t' \in [t-h, t]$ and $d^\co_{t}(t') = 0$ for all $t' \not\in [t-h, t]$.
\end{example}
\begin{example}[Linear decay carryover kernel]\label{example:linear-decay-carryover-kernel}
If the carryover intensity decays linearly in $t - t^\prime$ for $t^\prime \in [t-h, t]$, and is zero outside this interval, then $d^\co_{t}(t') \propto t-t^\prime$ for all $t' \in [t-h, t]$ and $d^\co_{t}(t') = 0$ for all $t' \not\in [t-h, t]$. See examples in Figure \ref{fig:setup-spillovers}.
\end{example}

Below, we introduce a condition under which the effects of primary and simultaneous interventions are additive. This represents a special case of the confounding between primary and simultaneous interventions. We show how this condition simplifies the decomposition of MSE and makes the decomposition more interpretable.

\begin{condition}[Additivity of Intervention Effects]\label{cond:additive}
	The effects of primary and simultaneous interventions are additive, i.e., 
	\begin{align*}
		Y_t(\bm{w}^\prime_t, \bm{w}^\s_t) - Y_t(\bm{w}_t, \bm{w}^\s_t) 
		= 	 Y_t(\bm{w}^\prime_t, \bm{w}^{\s\prime}_t) - Y_t(\bm{w}_t,\bm{w}^{\s\prime}_t)\, ,
	\end{align*}
	where $\bm{w}^\prime$ and $\bm{w}$ are two treatment assignments of the primary intervention, and $\bm{w}^\s$ and $\bm{w}^{\s\prime}$ are two treatment assignments of the simultaneous intervention.
    
\end{condition}

This condition excludes intervention effects from being synergistic (combining two interventions leads to a larger effect than expected) or antagonistic (combining two interventions leads to a smaller effect than expected). Condition \ref{cond:additive} is reasonable for certain classes of distinct interventions; for example, we may often assume that a pricing change and a routing change act via different mechanisms and are thus additive.

    \subsection{Interval-Level Statistics}\label{subsec:block-stats}
    
    We introduce several interval-level statistics that quantify carryover effects, correlations in measurement errors, confounding effects from simultaneous interventions, and other components at the interval level. These interval-level statistics serve as building blocks of the bias and MSE decomposition in Section \ref{subsec:mse-result}.

    \paragraph{Fraction of events.}

    Let
    \[\mu^{(m)} =  \int_{t \in \mathcal{I}_{m} } f(t) d t \]
    be the fraction of events occurring in the interval $\mathcal{I}_{m}$. $\mu^{(m)} $ ranges from 0 to 1 and the sum of $\mu^{(m)}$ over $m$ equals to $1$. 

    \begin{example}
        If event density $f(t)$ is uniform in time $t$, then $\mu^{(m)} = |\mathcal{I}_{m}|/T$. Moreover, if the fixed duration switchback is used, then $\mu^{(m)} = 1/M$.
    \end{example}

            	\begin{figure}[t]
	\centering
	\begin{subfigure}[b]{0.33\textwidth}
			\centering
			\includegraphics[height=4cm, keepaspectratio]{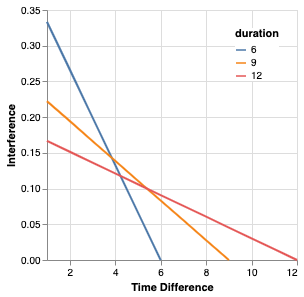}
			\caption{Carryover kernel}   	
 			\label{fig:setup-spillovers}
	\end{subfigure}%
		\begin{subfigure}[b]{0.33\textwidth}
			\centering
			\includegraphics[height=4cm, keepaspectratio]{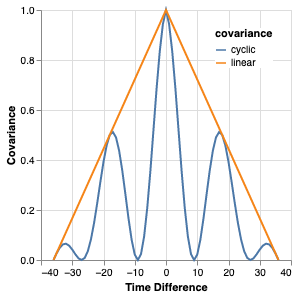}
			\caption{Covariance kernel}   	
 			\label{fig:setup-covariance}
	\end{subfigure}%
	\caption{Illustration of carryover and covariance kernels. In the carryover kernel $d^\co_{t}(t^\prime)$, the time difference denotes $t^\prime - t$. If $t^\prime - t < 0$, then $d^\co_{t}(t^\prime) = 0$. The interpretation of time difference is analogous to the covariance kernel. The covariance kernel can capture the serial correlation demonstrated in Figure \ref{fig:four-factors}.}
 	\label{fig:setup}
\end{figure}

    \paragraph{Mean control outcome.} Let

    \[ \mu^{(m)}_{Y^\mathrm{ctrl}}  = \int_{t \in \mathcal{I}_{m}}  Y_{t}(\bm{0}_t, \bm{0}_t) f(t) dt\]
    be the integrated global control outcome $Y_{t}(\bm{0}_t, \bm{0}_t)$ over time $t$ in the interval $\mathcal{I}_{m}$. 

    \paragraph{Variance and covariance of measurement errors.}
    Let the integrated variance of the measurement error for events in the interval $\mathcal{I}_{m}$ be 
    \[V^{(m)} = \int_{t_i \in \mathcal{I}_{m} } \+E_{\varepsilon}\left[(\varepsilon^{(i)})^2  \mid  t_i \right]  f(t_i) d t_i \, , \]
    where $\+E_{\varepsilon}\left[(\varepsilon^{(i)})^2  \mid  t_i \right]$ represents the variance of measurement error for event $i$ occurring at time $t_i$ (as measurement error has mean zero).

    \begin{example}
        Suppose measurement errors are homoscedastic, that is, $\+E_{\varepsilon}\left[(\varepsilon^{(i)})^2  \mid  t_i \right]= \sigma^2$ for all time $t_i$. If the event density $f(t)$ is uniform in time $t$, then $V^{(m)} = \sigma^2 |\mathcal{I}_{m}|/T$. Additionally, if the fixed duration switchback is used, then $V^{(m)} = \sigma^2/M$. 
    \end{example}

    Next, let the integrated covariance between measurement errors of events in the interval $\mathcal{I}_{\ell m}$ be 
    \[C^{(m)} =  \int_{t_i, t_j \in \mathcal{I}_{m} } \+E_{\varepsilon}\left[\varepsilon^{(i)} \varepsilon^{(j)} \mid  t_i, t_j \right] f(t_i) f(t_j) d t_i d t_j \, ,\]
    where $\+E_{\varepsilon}\left[\varepsilon^{(i)} \varepsilon^{(j)} \mid  t_i, t_j \right] $ is the covariance between the measurement errors of event $i$ occurred at time $t_i$ and event $j$ occurred at time $t_j$.
    
    In practical settings, patterns often exist in how the covariance $\+E_{\varepsilon}\left[\varepsilon^{(i)} \varepsilon^{(j)} \mid  t_i, t_j \right] $ varies with time $t_i$ and time $t_j$, e.g., decays monotonically or periodically with the distance between time $t_i$ and time $t_j$. Therefore, a kernel function can be used to parameterize and capture the patterns in $\+E_{\varepsilon}\left[\varepsilon^{(i)} \varepsilon^{(j)} \mid  t_i, t_j \right] $. See examples in Figure \ref{fig:setup-covariance}.

    \paragraph{Integrated total treatment effects.}

    Let the integrated total treatment effect over time $t$ in the interval $\mathcal{I}_m$ be
    \begin{align*}
        \Xi^{(m)} = \int_{t \in \mathcal{I}_{m}} \delta^\gate_{t}  f(t) d t \, .
    \end{align*}
    Following the definition of $\delta^\gate$, the sum of $\Xi^{(m)}$ over $m$ equals $\delta^\gate$. Moreover, if $\delta^\total_{t}$ is constant in time $t$, then $\Xi^{(m)} = \delta^\gate \mu^{(m)}$ for any $m$.

    \paragraph{Integrated carryover effects.} Let the integrated carryover effect of primary treatments in $\mathcal{I}_k$ to outcomes in $\mathcal{I}_m$ be
        \[I^{(m,k)} = \int_{t \in \mathcal{I}_{m}}  \left[\delta^\co_{t} \int_{t^\prime \in \mathcal{I}_{k}} d^\co_{t}(t^\prime) f(t^\prime) d t^\prime \right]   f(t)   dt \, . \]
    For simplicity in notation, we let $I^{(m)} = I^{(m,m)}$ be the integrated carryover effect of treatments on outcomes in the same interval. The integrated carryover effect $I^{(m,k)}$ increases with the length of both $\mathcal{I}_{m}$ and $\mathcal{I}_{k}$, and increases with the size of carryover effect $\delta^\co_{t}$ for time $t \in \mathcal{I}_{m}$. The sum of $I^{(m,k)}$ over both $m$ and $k$, which is the integrated carryover effect of primary treatments of all intervals on the outcomes of all intervals, is equal to the average of carryover effect $\delta^\co_t$ over time $t$ weighted by event density $f(t)$ (which is denoted by $\delta^\co$ hereafter). Moreover, if the carryover effect $\delta^\co_t$ is constant in $t$, then the sum of $I^{(m,k)}$ over $k$, which is the integrated carryover effect of the treatment of all intervals on the outcomes in the interval $\mathcal{I}_{m}$, is equal to $\delta^\co \mu^{(m)}$. Therefore, $I^{(m,k)}$ can be viewed as the building block of $\delta^\co$.  

    \paragraph{Confounding effects from the simultaneous intervention.} For any time $t$, let
    \[\delta^\simul_t(\bm{W}_t) = \+E_{\bm{W}^\s_t}\left[Y_t(\bm{W}_t, \bm{W}^\s_t) - Y_t(\bm{W}_t, \bm{0}_t) \mid \bm{W}_t, t \right] \]
    be the expected treatment effects from the simultaneous intervention at time $t$, conditional on $\bm{W}_t$. Here, the expectation is taken with respect to the distribution from which the treatment design of the simultaneous intervention is drawn. If the simultaneous intervention has nonzero treatment effects, then $\delta^\simul_t(\bm{W}_t)$ is generally nonzero, which can then bias the HT estimator of $\delta^\gate$. 
    
    We introduce a quantity below that measures the integrated bias from the treatment effects of simultaneous intervention
    \[S^{(m)} = \int_{t \in \mathcal{I}_{m}} \Phi^\simul_t f(t) dt \, , 
    \]
    where $\Phi^\simul_t$ is defined as
    \[\Phi^\simul_t = \+E_{\bm{W}^{(-m)} }\big[ \delta^\simul_t(\bm{W}^{(-m)},  W^{(m)} = 1 ) \big]  - \+E_{\bm{W}^{(-m)} }\big[ \delta^\simul_t(\bm{W}^{(-m)},  W^{(m)} = 0 ) 
        \big] \, .\]
    Here $W^{(m)}$ and $\bm{W}^{(-m)}$ are the treatment assignments of interval $\mathcal{I}_m$ and of all the intervals excluding $\mathcal{I}_m$, respectively. 
    As treatments are assigned at the interval level, $\bm{W}_t$ is uniquely determined given $(W^{(m)},\bm{W}^{(-m)})$, so we can also write $\delta^\simul_t(\bm{W}_t)$ as $\delta^\simul_t(\bm{W}^{(-m)},  W^{(m)})$.

    The quantity $\Phi^\simul_t$ measures the imbalance of expected treatment effects from the simultaneous intervention at time $t$ between when the $m$-th interval is treated ($W^{(m)} = 1$) versus when this interval is not treated ($W^{(m)} = 0$). Here, the expectation is taken with respect to the distribution from which $\bm{W}^{(-m)}$ is drawn. If the imbalance $\Phi^\simul_t$ is larger, then the integrated bias from simultaneous interventions $S^{(m)}$ is generally larger. 

    Note that $\Phi^\simul_t$ can be zero in some special cases, such as when the effects of primary and simultaneous interventions are additive (Condition \ref{cond:additive} holds), as illustrated in the following example. 

    \begin{example}[Additive effects]\label{example:additive-simul}
        When Condition \ref{cond:additive} holds, we have 
        \begin{align*}
            & Y_t((\bm{W}^{(-m)},  W^{(m)} = 1 ) , \bm{W}^\s_t) - Y_t((\bm{W}^{(-m)},  W^{(m)} = 0 )  , \bm{W}^\s_t) \\
            =& Y_t((\bm{W}^{(-m)},  W^{(m)} = 1 ) , \bm{0}_t) - Y_t((\bm{W}^{(-m)},  W^{(m)} = 0 ) , \bm{0}_t) \, .
        \end{align*}
        Moving the terms from the right-hand side to the left-hand side and then taking the expectation over $\bm{W}^\s_t$, we obtain 
        \[\delta^\simul_t(\bm{W}^{(-m)},  W^{(m)} = 1 ) - \delta^\simul_t(\bm{W}^{(-m)},  W^{(m)} = 0 ) = 0 \, , \qquad \text{for any } \bm{W}^{(-m)} \, , \]
        which implies that $\Phi^\simul_t$ is zero.
    \end{example}

    However, when the effects of primary and simultaneous interventions are not additive (e.g., two interventions are useful only when both are present), $\Phi^\simul_t$ is generally nonzero and is a source of bias of the HT estimator.

    \subsection{Main Results}\label{subsec:mse-result}
    
    We provide the decomposition of the bias and MSE of $\hat{\delta}^\gate$ from the HT estimator in terms of the interval-level statistics in Theorems \ref{theorem:switchback-bias} and \ref{theorem:bias-variance-switchback} below. The decomposition shows how different components in the outcome affect the estimation error of $\hat{\delta}^\gate$. In the theorem statements, we use $W^{\s(m)}$ to denote the $m$-th interval's treatment assignment of the simultaneous intervention.

    \begin{theorem}[Estimation Bias]\label{theorem:switchback-bias}
        Suppose Assumptions \ref{ass:exogeneity} and \ref{ass:interference-structure} hold, $W^{(m)}$ is independent in $m$ with $\P(W^{(m)} = 1)  = 1/2$, and $W^{\s(m)}$ is independent in $m$ with $ \P(W^{\s(m)} = 1) = 1/2$. Moreover, $\bm{W}$ is independent of $\bm{W}^\s$. The estimation bias of $\hat{\delta}^\gate$ is
        \begin{align*}
        \+E_{W,\varepsilon,t}\left[\hat{\delta}^\gate  - \delta^\gate\right] =& \bias(\mathcal{E}_\mathrm{carryover}) + \bias(\mathcal{E}_\mathrm{simul}) 
        \, ,
    \end{align*}
    where 
    \begin{align*}
        &\bias(\mathcal{E}_\mathrm{carryover}) = \sum_{m = 1}^M I^{(m)} - \delta^\co \\
	&\bias(\mathcal{E}_\mathrm{simul}) = \sum_{m = 1}^M S^{(m)} \, .
    \end{align*}
    \end{theorem}

    Theorem \ref{theorem:switchback-bias} shows there are two sources of bias in the HT estimator. The first source of bias comes from the carryover effects and is measured by $\bias(\mathcal{E}_\mathrm{carryover})$. If carryover effects are zero, (i.e., $\delta^\co_{t} = 0$ for all time $t$), then $\bias(\mathcal{E}_\mathrm{carryover}) = 0$. If carryover effects are nonzero, then $\bias(\mathcal{E}_\mathrm{carryover})$ quantifies the bias in the HT estimator arising from using direct treated and control outcomes to approximate globally treated and control outcomes, respectively.

The second source of bias comes from the confounding effects of simultaneous interventions and is measured by $\bias(\mathcal{E}_\mathrm{simul})$. When the effects of main and simultaneous interventions are additive, as illustrated in Example \ref{example:additive-simul}, $S^{(m)}$ is zero and then $\bias(\mathcal{E}_\mathrm{simul})$ is also zero. However, when effects are not additive, $\bias(\mathcal{E}_\mathrm{simul})$ is generally nonzero and tends to increase with the effect size of the simultaneous intervention.

Both sources of bias can be reduced by properly choosing a switchback design. To mitigate the bias from carryover effects, switching less frequently helps. For example, in the setting of uniform event density and fixed-duration switchback, $I^{(m)}=\delta^\co(1/M - h/(2T))$ as shown in Example \ref{example:carryover-value}, and then the carryover bias is $|\bias(\mathcal{E}_\mathrm{carryover})| = \delta^\co Mh/(2T)$, which increases with the number of intervals $M$. To reduce the bias from simultaneous experiments, randomizing and staggering the switching times helps, such as by using Poisson duration switchback. This can be clearly seen in the toy Example \ref{example:simul-misalign} in Appendix \ref{subsec:additional-examples} and simulation results in Figure \ref{subfig:simul-vary-design} below. 

Next, we show the decomposition of the MSE of $\hat{\delta}^\gate$.

    \begin{theorem}[Mean-Squared Error]\label{theorem:bias-variance-switchback}
Suppose the assumptions in Theorem \ref{theorem:switchback-bias} hold.
The MSE of $\hat{\delta}^\gate$ is
    \begin{align*}
    \+E_{W,\varepsilon,t}\left[\left(\hat{\delta}^\gate  - \delta^\gate \right)^2\right] =& \var(\mathcal{E}_\mathrm{meas}) +  \bias(\mathcal{E}_\mathrm{carryover}) ^2  + \var(\mathcal{E}_\mathrm{inst}+\mathcal{E}_\mathrm{carryover}) \\
    & + \+E[\mathcal{E}_\mathrm{simul}^2] + 2 \+E[(\mathcal{E}_\mathrm{inst}+\mathcal{E}_\mathrm{carryover})\cdot \mathcal{E}_\mathrm{simul}] \, ,
    \end{align*}
    where 
    \begin{align*}
    &\var(\mathcal{E}_\mathrm{meas}) = 4 \sum_{m=1}^M \left(V^{(m)}/n  + C^{(m)} \cdot (n-1)/n \right) \\
    &\var(\mathcal{E}_\mathrm{inst}+\mathcal{E}_\mathrm{carryover})  = \sum_{m = 1}^M \left(\Xi^{(m)}  + 2  \mu^{(m)}_{Y^\mathrm{ctrl}}   \right)^2   + \sum_{m=1}^M \sum_{m^\prime \neq m} \left(\left[I^{(m,m^\prime)}\right]^2 + I^{(m,m^\prime)} I^{(m^\prime, m)}\right)
    \end{align*}
    and 
    \begin{align*}
    &\+E[\mathcal{E}_\mathrm{simul}^2] = \sum_{m=1}^M \sum_{m^\prime = 1}^M S^{(m,m^\prime)}_{\mathrm{var}} \\
    &\+E[(\mathcal{E}_\mathrm{inst}+\mathcal{E}_\mathrm{carryover}) \cdot \mathcal{E}_\mathrm{simul}]  = \sum_{m = 1}^M \sum_{m^\prime = 1}^M S^{(m,m^\prime)}_{\mathrm{cov}}
    \end{align*}
     with $S^{(m,m^\prime)}_{\mathrm{var}}$ and $S^{(m,m^\prime)}_{\mathrm{cov}}$ defined in Equations \eqref{eqn:S-m-mprime-var} and \eqref{eqn:S-m-mprime-cov}, respectively, in Appendix \ref{subsec:notations}.
    \end{theorem}

    Theorem \ref{theorem:bias-variance-switchback} demonstrates that, in addition to the bias terms, the MSE is affected by three sources of variance. The first source of variance arises from the randomness in event measurement errors and is quantified by $\var(\mathcal{E}_\mathrm{meas})$. Note that $\var(\mathcal{E}_\mathrm{meas})$ consists of two parts: the first part $V^{(m)}$, which measures the variance of event measurement error, and the second part $C^{(m)}$, which measures the covariance between measurement errors of two events. As the number of events grows, the first part diminishes and the second part dominates. If the correlation in measurement errors is persistent, then $C^{(m)}$ is larger, leading to a larger MSE. 

    The second source of variance stems from the randomness in treatment assignments of the primary intervention and is measured by $\var(\mathcal{E}_\mathrm{inst}+\mathcal{E}_\mathrm{carryover})$. The expression of $\var(\mathcal{E}_\mathrm{inst}+\mathcal{E}_\mathrm{carryover})$ shows that: (1) it increases with the size of the instantaneous effect, following that the term $\Xi^{(m)}$ increases with the instantaneous effect; (2) it increases with the size of the carryover effect, following that both $\Xi^{(m)}$ and $I^{(m,m^\prime)}$ increase with the carryover effect; and (3) it increases with the scale of the mean outcome, following that the term $\mu^{(m)}_{Y^\mathrm{ctrl}} $ increases with the mean outcome.

    The third source of variance arises from the randomness in treatment assignments of the simultaneous intervention and affects $\+E[\mathcal{E}_\mathrm{simul}^2]$. In Proposition \ref{prop:additive-one-simul-effect} below, we present the expression for $\+E[\mathcal{E}_\mathrm{simul}^2]$ under the condition of additive primary and simultaneous effects. This expression shows that this source of variance increases with the magnitude of instantaneous and carryover effects of simultaneous intervention. Furthermore, this term increases with the overlap between intervals of the primary and simultaneous interventions.

    In addition to the bias and variance terms, the MSE includes a cross term $\+E[(\mathcal{E}_\mathrm{inst}+\mathcal{E}_\mathrm{carryover}) \cdot \mathcal{E}_\mathrm{simul}] $. Proposition \ref{prop:additive-one-simul-effect} also provides the expression for this term under the condition of additive primary and simultaneous effects.\footnote{A preliminary version of this proposition appears in \cite{xiong2023bias}. Theorems \ref{theorem:switchback-bias} and \ref{theorem:bias-variance-switchback} are more general as they both allow for non-additive effects of the primary and simultaneous interventions. } Generally, if the variance terms from primary and simultaneous interventions are larger, then this cross term is larger.

        \begin{proposition}\label{prop:additive-one-simul-effect}
        Under the assumptions in Theorem \ref{theorem:bias-variance-switchback}, if Condition \ref{cond:additive} holds, then the bias from the simultaneous intervention $\bias(\mathcal{E}_\mathrm{simul}) $ is zero
         and the variance from the simultaneous intervention is 
             \begin{align*}
         \+E[\mathcal{E}_\mathrm{simul}^2] 
         =&  \sum_{m = 1}^M  \bigg( 
\int_{t\in \mathcal{I}_{m}}  \delta^{\s.\gate}_t  f(t) dt \bigg)^2   \\ & + \sum_{m = 1}^M  \sum_{m^\prime =1}^M  \bigg( 
\int_{t\in \mathcal{I}_{m} \cap \mathcal{I}^\s_{m^\prime} }  \delta^{\s.\inst}_t  f(t) dt + \int_{t \in \mathcal{I}_{m}, t^\prime \in \mathcal{I}^\s_{m^\prime} }  \delta^{\s.\co}_{t} d^{\s.\co}_{t}(t^\prime)  f(t) f(t^\prime)  dt d t^\prime\bigg)^2 \, ,
    \end{align*}
    and
    \begin{align*}
     \+E[(\mathcal{E}_\mathrm{inst}+\mathcal{E}_\mathrm{carryover}) \cdot \mathcal{E}_\mathrm{simul}] =& \sum_{m=1}^M \left(\Xi^{(m)} + 2 \mu^{(m)}_{Y^\mathrm{ctrl}}\right)\bigg( 
\int_{t\in \mathcal{I}_{m}} \delta^{\s.\gate}_t  f(t) dt \bigg) \, ,
 \end{align*}
 where $\delta^{\s.\gate}_t$, $\delta^{\s.\inst}_t$, $\delta^{\s.\co}_t$, and $d^{\s.\co}_t$ are the total treatment effect, instantaneous effect, carryover effect, and carryover kernel of the simultaneous intervention at time $t$, respectively, and $\mathcal{I}^\s_{m}$ is the $m$-th interval of the simultaneous intervention. 
    \end{proposition}

    The MSE decomposition makes it clear how the design affects the MSE.
    First and foremost, the switching interval length has a mixed effect on MSE. Switching frequently reduces most variance and cross terms because the number of interval-level observations increases, thereby increasing the effective sample size.
    
    This can be seen clearly by analyzing the order of the terms in the setting of uniform event density and fixed-duration switchback. Specifically, in this setting, $\var(\mathcal{E}_\mathrm{meas}) = O(1/M)$ decreases with $M$, following that $C^{(m)} = O(1/M^2)$ as shown in Example \ref{example:covariance-value}. Moreover, the term $\sum_{m = 1}^M \big(\Xi^{(m)} + 2  \mu^{(m)}_{Y^\mathrm{ctrl}}\big)^2 = O(1/M)$ also decreases with $M$, given that $\Xi^{(m)} = O(1/M)$ and $\mu^{(m)}_{Y^\mathrm{ctrl}} = O(1/M)$. Analogously, for the variance from the simultaneous intervention, both $\+E[\mathcal{E}_\mathrm{simul}^2] $ and $\+E[(\mathcal{E}_\mathrm{inst}+\mathcal{E}_\mathrm{carryover}) \cdot \mathcal{E}_\mathrm{simul}]$ are $O(1/M)$ under the additivity condition. 
    
    On the other hand, switching frequently increases the carryover bias.
    This is because the carryover effects from treatments in other intervals increase with switching frequency. Therefore, with the tradeoff involved, the optimal value of $M$ depends on the relative size of instantaneous and carryover effects, the scale of the global control outcomes, and the duration and mechanism of carryover effects.

    Second, the choice of interval endpoints also affects the MSE. To see this more clearly, consider the setting with constant $\delta^\gate_{t}$ in $t$ and constant $Y_t(\bm{0}, \cdots, \bm{0})$ in $t$ (equals to $\bar{Y}^{\mathrm{ctrl}}$). Then $\sum_{m = 1}^M \big(\Xi^{(m)} \big)^2 = (\delta^\gate + 2 \bar{Y}^{\mathrm{ctrl}})^2 \sum_{m = 1}^M \big( \mu^{(m)} \big)^2$, which is minimized at $\mu^{(m)} = 1/M$ for all $m$ with the same fraction of events in each interval. This implies that when event density $f(t)$ varies with $t$, this term can be reduced by switching more frequently in times of high event density and less frequently in times of low event density. The change-of-measure switchback has such switching patterns and is therefore considered in our case study.
    
    In addition, the second term of $\+E[\mathcal{E}_\mathrm{simul}^2] $ is affected by how much the intervals of primary and simultaneous interventions overlap, and is the largest when the interval endpoints of primary and simultaneous interventions are the same. Therefore, randomizing and staggering the switching times of different interventions can reduce $\+E[\mathcal{E}_\mathrm{simul}^2] $. This then motivates us to consider the Poisson switchback in our case study.
    
    Last but not least, balancing periodicity can reduce the MSE. Recall the setting of a two-week experiment with the balanced design in Example \ref{example:balanced-design}. Under mild assumptions, the mean control outcome $\mu^{(m)}_{Y^\mathrm{ctrl}}$ can be canceled out in both the variance term  $\var(\mathcal{E}_\mathrm{inst}+\mathcal{E}_\mathrm{carryover}) $ and the cross term $\+E[(\mathcal{E}_\mathrm{inst}+\mathcal{E}_\mathrm{carryover}) \cdot \mathcal{E}_\mathrm{simul}] $. The variance term is then equal to
        \[\var(\mathcal{E}_\mathrm{inst}+\mathcal{E}_\mathrm{carryover})  = \sum_{m = 1}^M \left(\Xi^{(m)}  \right)^2   + \sum_{m=1}^M \sum_{m^\prime \neq m} \left(\left[I^{(m,m^\prime)}\right]^2 + I^{(m,m^\prime)} I^{(m^\prime, m)}\right) \, .\]
    Furthermore, under Condition \ref{cond:additive}, the cross term is equal to
             \begin{align*}
     \+E[(\mathcal{E}_\mathrm{inst}+\mathcal{E}_\mathrm{carryover}) \cdot \mathcal{E}_\mathrm{simul}] =& \sum_{m=1}^M \Xi^{(m)} \bigg( 
\int_{t\in \mathcal{I}_{m}} \delta^{\s.\gate}_t  f(t) dt \bigg) \, .
 \end{align*}
    Canceling out $\mu^{(m)}_{Y^\mathrm{ctrl}}$ can substantially reduce the variance when the relative treatment effects are small (not more than a few percent of the outcome mean). This is exactly the case in our case study on a ride-sharing platform. 
    In fact, the HT estimator coincides with the Hajek estimator when the balanced designs are used, but this is not the case for unbalanced designs. The Hajek estimator is generally more efficient than the HT estimator; this gives another perspective as to why balanced designs are more efficient than unbalanced designs.  

    \subsection{Inference of Treatment Effects}

    To draw inferences on the treatment effect, one approach is to run a randomization-based test for the sharp null of no effect from the primary intervention, holding the treatment assignments of the simultaneous experiment as the realized ones: 
    \[\mathcal{H}_0 : Y_t(\bm{w}_t, \bm{W}^\s_t) = Y_t(\bm{0}_t, \bm{W}^\s_t ) \, , \qquad \forall \bm{w}_t  ~\text{and}~ t \, . \]
    Under this sharp null hypothesis, the outcomes would have been identical to the observed ones regardless of what treatment assignments $\bm{w}_t$ are used in the primary experiment, and $\delta^\gate = 0$. Then we can calculate the $p$-value of this hypothesis using the same procedure as \cite{bojinov2020design}. Specifically, we construct an empirical distribution of the estimated GATE under the sharp null and use it to compute the $p$-value. To construct the empirical distribution, suppose we draw the treatment designs and re-calculate the GATE for $J$ times. In the $j$-th iteration, let $\bm{W}^{[j]}$ be the drawn treatment design, which is used with the observed outcomes $Y^{(1)}, \cdots, Y^{(n)}$ to estimate $\delta^\gate$
    \begin{equation*}
        \hat{\delta}^{\gate,[j]} = \frac{1}{n} \sum_i \left(\frac{W^{[j]}_{t_i} }{\pi} - \frac{1-W^{[j]}_{t_i} }{1-\pi} \right) Y^{(i)}  \, .
    \end{equation*}

    We then use $\hat{\delta}^{\gate,[1]}, \cdots, \hat{\delta}^{\gate,[J]}$ from the $J$ iterations to estimate the $p$-value 
    \[\hat{p} = \frac{1}{J} \sum_{j = 1}^J \bm{1}\left(|\hat{\delta}^\gate| \geq |\hat{\delta}^{\gate,[j]}| \right) \, , \]
    where $\bm{1}(|\hat{\delta}^\gate| \geq |\hat{\delta}^{\gate,[j]}| )$ is an indicator function that takes value of one if $|\hat{\delta}^\gate| \geq |\hat{\delta}^{\gate,[j]}|$, and zero otherwise. 
    Besides calculating the $p$-value, we can invert the randomization-based test similar to \cite{tritchler1984inverting} to construct the confidence interval of $\hat{\delta}^\gate$ under the sharp null.

    \subsection{Alternative Estimators}\label{subsec:alternative-estimator}

    It is possible to reduce the estimation error of GATE by using alternative estimators that leverage prior knowledge of carryover and correlation mechanisms and information on other interventions. Specifically, to reduce the carryover bias, we can specify and leverage the structure of the carryover mechanisms when estimating GATE. For example, a commonly made assumption is that the fixed-duration carryover effect, i.e., the effect of treatment carries over to no more than $h$ periods in the future, for some $h$, such as in Examples \ref{example:uniform-carryover-kernel} and \ref{example:linear-decay-carryover-kernel}. Under this assumption, we can burn the first $h$ periods since switch (either from control to treatment, or from treatment to control) \citep{hu2022switchback}, and use the following modified HT estimator
    \begin{equation}\label{eqn:ht-estimator-burnin}
        \hat{\delta}^{\gate} = \frac{1}{n} \sum_i \left[\frac{W_{t_i} W_{t_i - h} Y^{(i)} }{\P(W_{t_i} = 1, W_{t_i - h} = 1)} - \frac{(1-W_{t_i}) (1 - W_{t_i - h}) Y^{(i)} }{\P(W_{t_i} = 0, W_{t_i - h} = 0) } \right] 
  .
    \end{equation}
    This estimator only uses the observations at times that have been in the treatment or control states for at least time $h$ and, therefore, encompasses a bias-variance tradeoff: It is unbiased under the assumption of a fixed-duration carryover effect, but it can have a larger variance due to a reduced sample size. In Figure \ref{fig:alternative-estimator} in Appendix \ref{appendix:empirical}, we compare the estimation error with various $h$ under the same synthetic experiment setup as that in Section \ref{subsec:empirical-results}. The hierarchical structure shown in Figure \ref{fig:pyramid} continues to hold with various $h$ and the balanced Poisson design with long switches has the lowest MSE for all $h$. However, it turns out that the estimator with \emph{zero burn-ins}, i.e., the HT estimator analyzed in this paper, has the lowest MSE regardless of which design is used. This empirical evidence supports our choice of estimator, and occurs because our data is noisy.

    To reduce the variance from correlated outcomes, we can specify and use the structure of the correlation mechanisms to reweigh the event outcomes. Moreover, if we know the design of the simultaneous experiment, we can use the generalized least squares (GLS) estimators to jointly estimate the instantaneous and carryover effects for all interventions, while taking advantage of the inverse error covariance weighting. GLS can be more efficient but requires the correct model specification. 
    Nevertheless, the design insights from the HT estimator can carry over to the GLS estimator, with a few numerical examples shown in \cite{xiong2019optimal}.

    Another possible future direction to improve estimation efficiency is to model the demand and supply, and use machine learning to estimate the demand and supply \citep{adam2023machine}. The estimated model of demand and supply could be used as the input in the estimation of GATE.

%% file: section_5.tex
\section{Simulation}\label{sec:simulation}

In this section, we compute values for error components in the MSE in the settings where error components are on a similar scale, fleshing out the interesting tradeoffs in designing a switchback experiment. The computation is based on Theorems \ref{theorem:switchback-bias} and \ref{theorem:bias-variance-switchback}. The results confirm the findings in our case study.  We use fixed-duration designs in the base case to characterize the tradeoffs involved, and the general insights carry over to other heuristic designs. In the base case, a linear decay carryover kernel is used with a bandwidth of $h_{\mathrm{carryover}} = 60$ (i.e., the duration of carryover effects is 60 minutes). Moreover, a linear decay covariance kernel is used with a bandwidth of $h_{\mathrm{covariance}} = 60$ (i.e., two event outcomes are correlated only if they are within 60 minutes apart). Both the instantaneous and carryover effects are constant in time and equal to $\delta^\inst_{t} = 1$ and $\delta^\co_{t} = 1$ for all $t$. In the base case, there is no simultaneous experiment, and the primary experiment lasts one day, i.e., $T = 1,440$ minutes. We vary the parameters in the base case and illustrate how error components vary and affect the performance of various switchback designs.

\subsection{Instantaneous and Carryover Effects Only} 

Figure \ref{fig:total-effect} illustrates the tradeoff between the bias from carryover effects, $\bias(\mathcal{E}_\mathrm{carryover}) $, and variance from instantaneous and carryover effects, $\var(\mathcal{E}_\mathrm{inst}+\mathcal{E}_\mathrm{carryover})$. If the carryover effects last longer, i.e., $h_{\mathrm{carryover}}$ is larger, then the bias $\bias(\mathcal{E}_\mathrm{carryover}) $ tends to be larger, and switching less frequently reduces the MSE from instantaneous and carryover effects, as shown in the comparison between Figures \ref{subfig:inst-carryover-base} and \ref{subfig:inst-carryover-vary-duration}. The relative size between instantaneous and carryover effects also matters for the tradeoff. If the instantaneous effect is relatively larger than the carryover effect, then switching more frequently reduces the MSE, as shown in the comparison between Figures \ref{subfig:inst-carryover-base} and \ref{subfig:inst-carryover-vary-effect-size}. 

\begin{figure}[h!]
    \centering
    
    \begin{subfigure}[b]{0.2\textwidth}
        \caption{base case}\label{subfig:inst-carryover-base}
        \includegraphics[width=\linewidth]{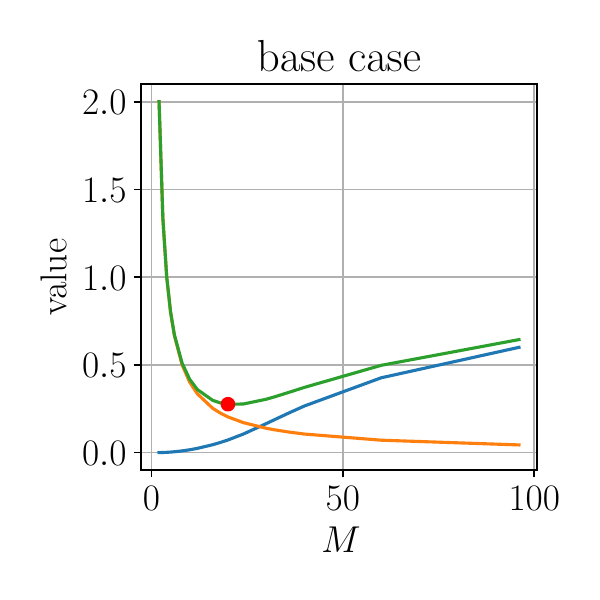}
    \end{subfigure}%
    \begin{subfigure}[b]{0.4\textwidth}
        \caption{vary $h_{\mathrm{carryover}}$ }\label{subfig:inst-carryover-vary-duration}
        \includegraphics[width=\linewidth]{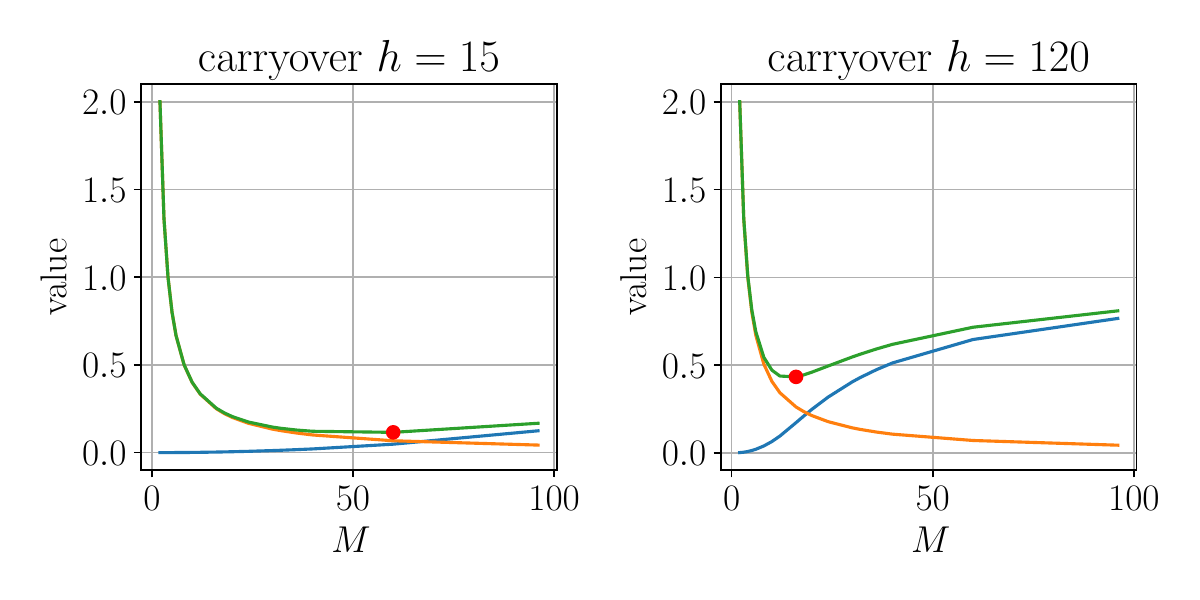}
    \end{subfigure}%
    \begin{subfigure}[b]{0.4\textwidth}
        \caption{vary $\delta^\inst/\delta^\co$}\label{subfig:inst-carryover-vary-effect-size}
        \includegraphics[width=\linewidth]{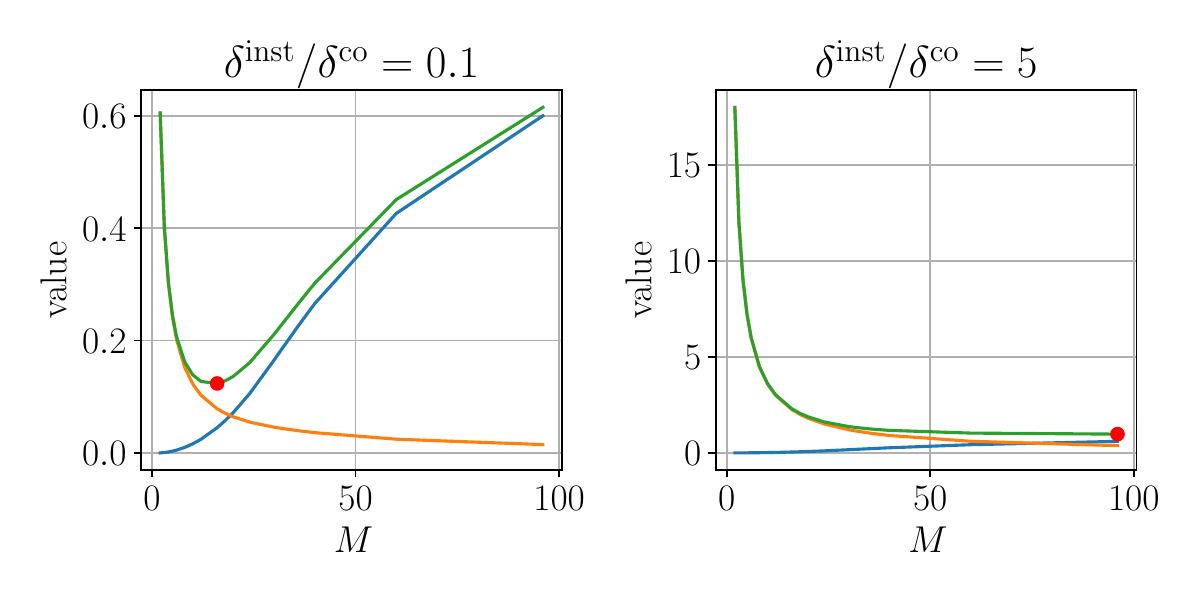}
    \end{subfigure}
    \begin{subfigure}[b]{\textwidth}
    \centering
        \includegraphics[width=0.6\linewidth]{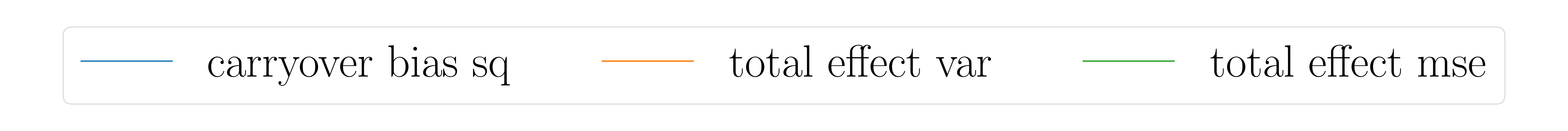}
    \end{subfigure}
    \caption{Tradeoff between error components from instantaneous and carryover effects. ``carryover bias sq'' (blue) denotes $\bias(\mathcal{E}_\mathrm{carryover})^2$, ``total effect var'' (orange) denotes $\var(\mathcal{E}_\mathrm{inst}+\mathcal{E}_\mathrm{carryover})$, and ``total effect mse'' (green) denotes the sum of the two. Figure \ref{subfig:inst-carryover-vary-duration} varies $h_{\mathrm{carryover}}$ and Figure \ref{subfig:inst-carryover-vary-effect-size} varies $\delta^\inst/\delta^\co$, while holding other parameters at the base level. The orange dot denotes the minimum ``total effect mse''. }
    \label{fig:total-effect}
\end{figure}

\subsection{Total Treatment Effects with Measurement Errors}
Figure~\ref{fig:total-effect-meas} summarizes the tradeoffs between error components from randomness in treatment assignments and measurement errors. Switching frequently generates more comparisons, reducing variance from measurement errors and increasing carryover bias from previous intervals. If correlation in measurement errors are persistent (i.e., large $h_{\mathrm{covariance}}$), then $\var(\mathcal{E}_\mathrm{meas}) $ has a larger impact on the MSE, and switching more frequently helps to reduce the MSE, as shown in the comparison between Figures \ref{subfig:inst-carryover-meas-base} and \ref{subfig:inst-carryover-meas-vary-duration}. Moreover, if event outcomes are noisier with a larger $\Var_{\sigma,t}$, then switching more frequently helps, as shown in the comparison between Figures \ref{subfig:inst-carryover-meas-base} and \ref{subfig:inst-carryover-meas-var}.

\begin{figure}[h]
	\centering
 \begin{subfigure}[b]{0.2\textwidth}
        \caption{base case}\label{subfig:inst-carryover-meas-base}
        \includegraphics[width=\linewidth]{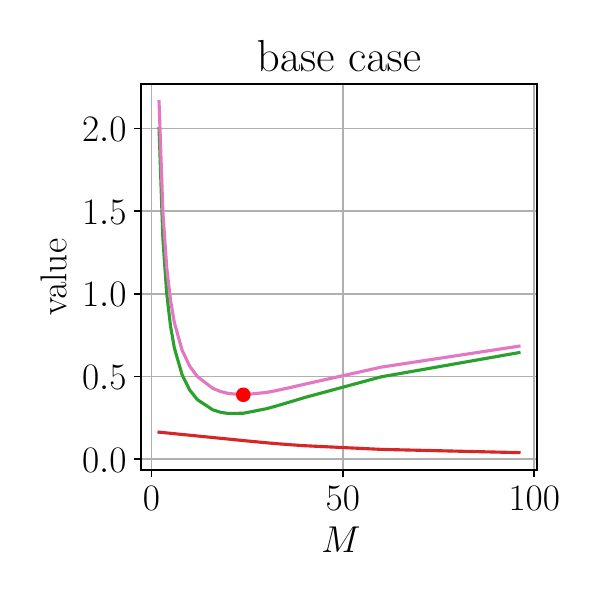}
    \end{subfigure}%
    \begin{subfigure}[b]{0.4\textwidth}
        \caption{vary $h_{\mathrm{covariance}}$}\label{subfig:inst-carryover-meas-vary-duration}
        \includegraphics[width=\linewidth]{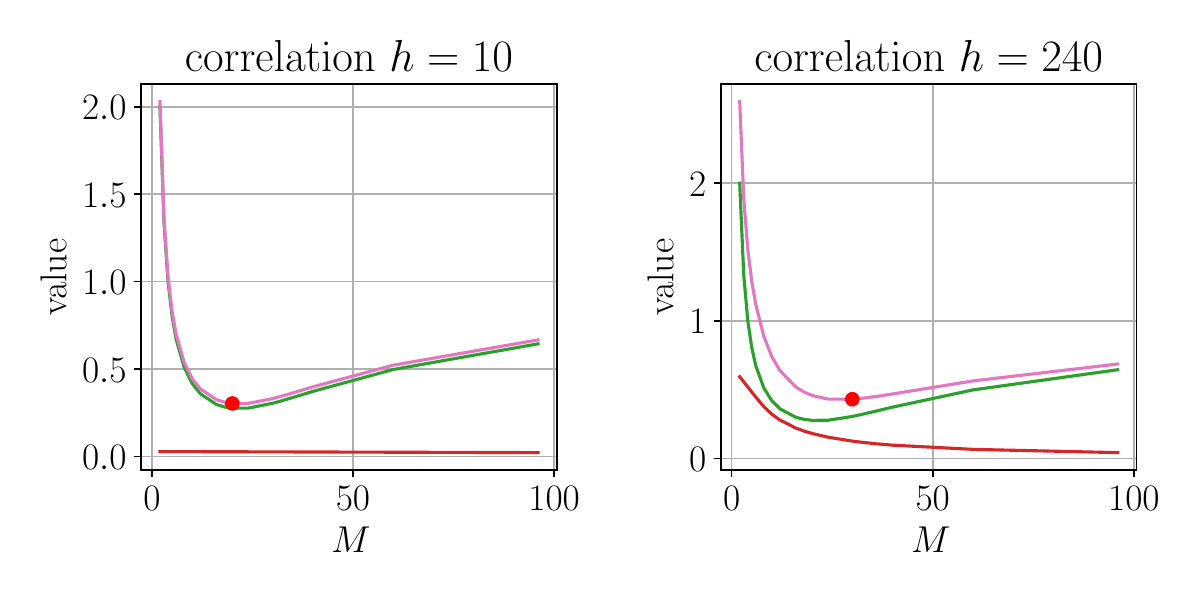}
    \end{subfigure}%
    \begin{subfigure}[b]{0.4\textwidth}
        \caption{vary $\Var_{\sigma,t}$}\label{subfig:inst-carryover-meas-var}
        \includegraphics[width=\linewidth]{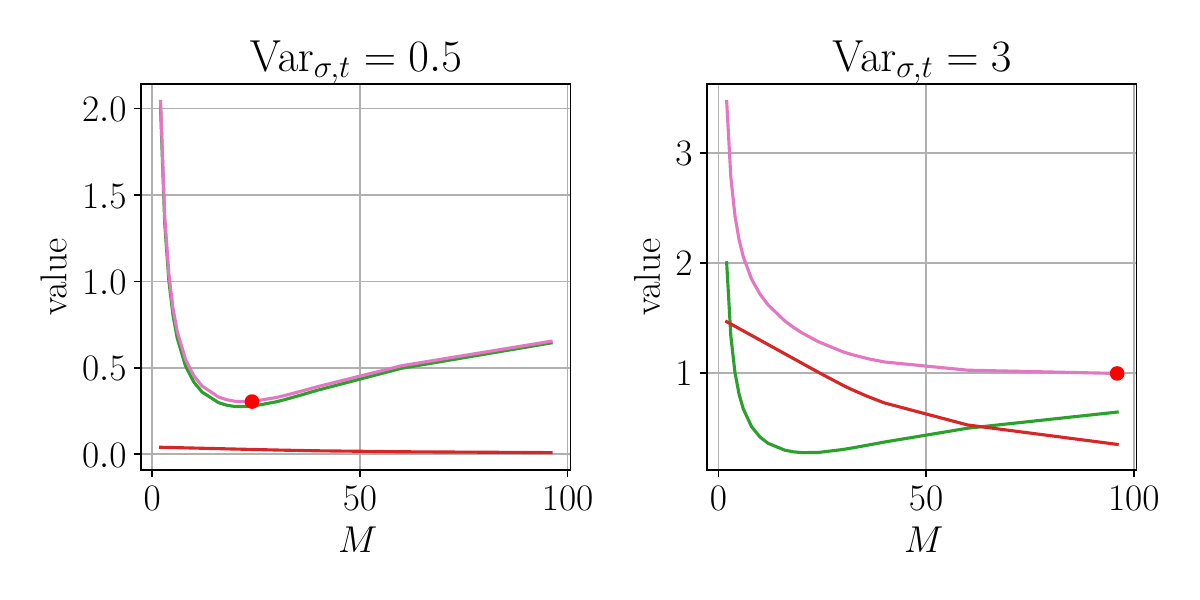}
    \end{subfigure}
    \begin{subfigure}[b]{\textwidth}
    \centering
        \includegraphics[width=0.6\linewidth]{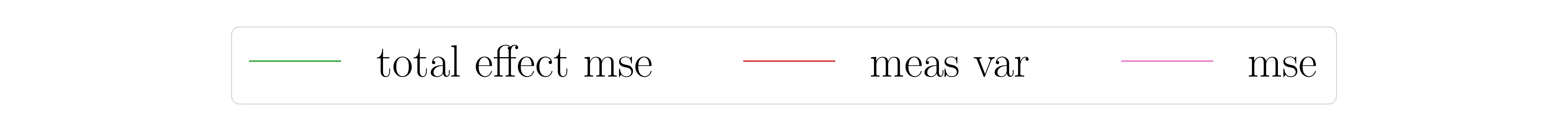}
    \end{subfigure}
	\caption{Tradeoffs between error components from randomness in treatment assignments and measurement errors. ``total effect mse'' (green) denotes the sum of $\bias(\mathcal{E}_\mathrm{carryover})^2$ and $\var(\mathcal{E}_\mathrm{inst}+\mathcal{E}_\mathrm{carryover})$,  ``meas var'' (red) denotes $\var(\mathcal{E}_\mathrm{meas}) $, and ``mse'' (pink) denotes the sum of ``total effect mse'' and ``meas var''. Figure \ref{subfig:inst-carryover-meas-vary-duration} varies $h_{\mathrm{covariance}}$ and Figure \ref{subfig:inst-carryover-meas-var} varies $\Var_{\sigma,t}$, while holding other parameters at the base level. The orange dot denotes the minimum ``mse''. }
	\label{fig:total-effect-meas}
\end{figure}

\subsection{Simultaneous Interventions}

Figure~\ref{fig:simul-exp} shows the tradeoffs involved in the presence of simultaneous interventions. The MSE of the primary intervention is affected by three factors: the number of simultaneous interventions $K$, switching frequency, and the offset in switching times between simultaneous experiments. When more experiments are run simultaneously, the optimal switching frequency increases. This is because switching more frequently helps reduce the confounding effects of simultaneous interventions. Moreover, properly staggering the switching times of the primary and simultaneous interventions also decreases the MSE of the primary intervention, with the effect being more pronounced with a low switching frequency due to the increased finite-sample correlation between the designs. As shown in Figure \ref{subfig:simul-vary-design}, Poisson switchbacks, which stagger through randomizing switching times, can be more effective unless the fixed duration designs are explicitly staggered well.

\begin{figure}[h]
	\centering
 \begin{subfigure}[b]{0.2\textwidth}
        \caption{base case}\label{subfig:simul-base}
        \includegraphics[width=1\linewidth]{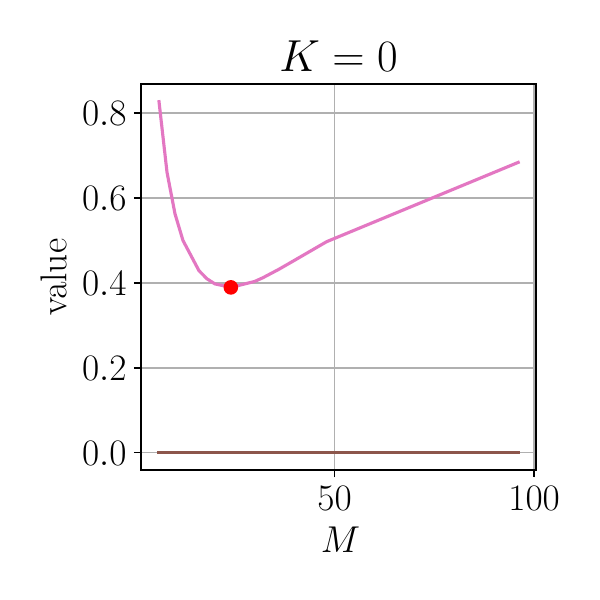}
    \end{subfigure}%
     \begin{subfigure}[b]{0.4\textwidth}
        \caption{vary $K$}\label{subfig:simul-vary-k}
        \includegraphics[width=1\linewidth]{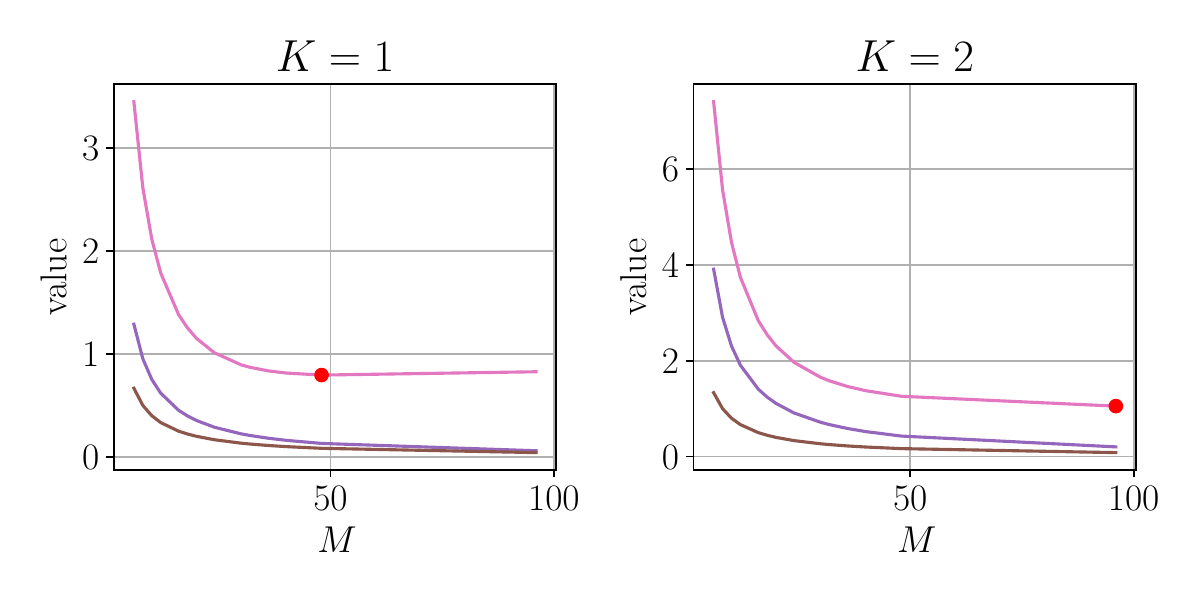}
    \end{subfigure}%
    \begin{subfigure}[b]{0.4\textwidth}
        \caption{compare designs (simul mse)}\label{subfig:simul-vary-design}
        \includegraphics[width=1\linewidth]{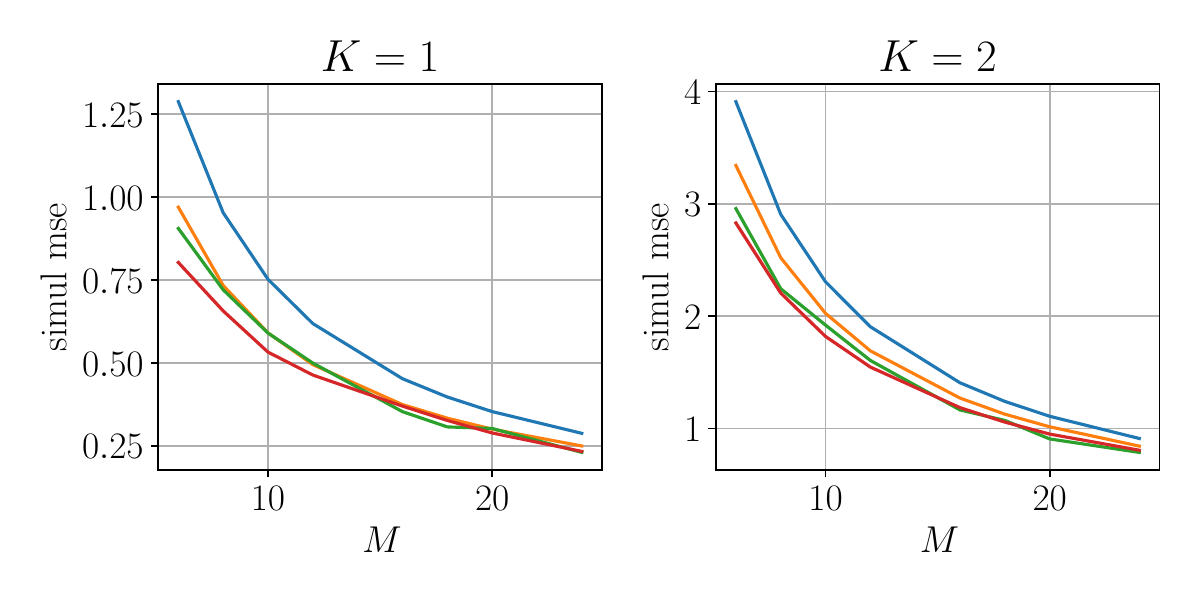}
    \end{subfigure}
    \begin{subfigure}[t]{0.49\textwidth}
        \includegraphics[width=0.9\linewidth]{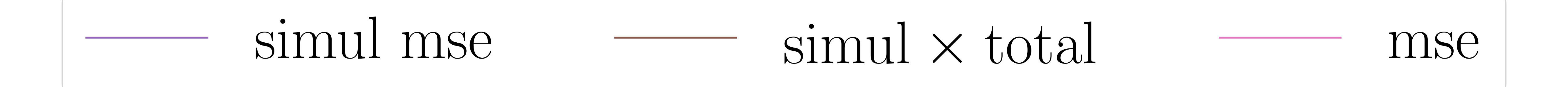}
    \end{subfigure}\hfill
    \begin{subfigure}[t]{0.49\textwidth}
        \raggedleft
        \includegraphics[width=0.9\linewidth]{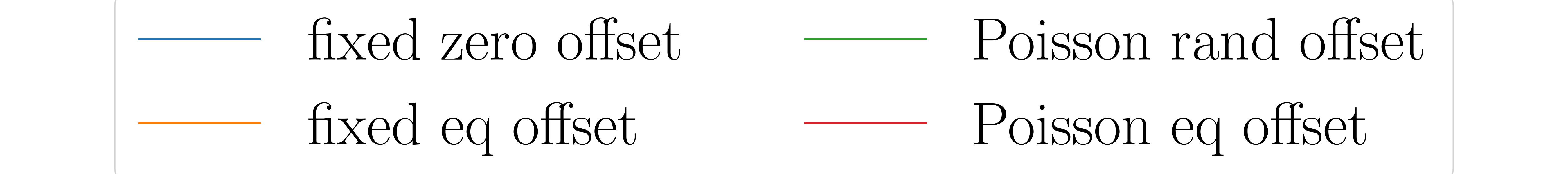}
    \end{subfigure}
		
	\caption{MSE with simultaneous interventions. ``simul mse'' (purple) denotes $\+E[\mathcal{E}_\mathrm{simul}^2]$, ``simul$\times$total'' (brown) denotes $\+E[(\mathcal{E}_\mathrm{inst}+\mathcal{E}_\mathrm{carryover}) \cdot \mathcal{E}_\mathrm{simul}] $, and ``mse'' (pink) denotes $ \+E_{W,\varepsilon,t}[(\hat{\delta}^\gate  - \delta^\gate)^2] $. Figure \ref{subfig:simul-vary-k} varies the number of simultaneous experiments $K$, while holding other parameters at the base level and switching at the same times for all interventions. Figure \ref{subfig:simul-vary-design} compares four designs: fixed duration switchbacks with offset $q = 0$ for all interventions (i.e., fixed zero offset), fixed duration switchbacks with offset $q= p \cdot j/(K + 1)$ for the $j$-th simultaneous intervention (i.e., fixed eq offset), Poisson switchbacks with offset $q \sim \mathrm{Poisson}(T/M)$ (i.e., Poisson rand offset), and Poisson switchbacks with offset $q= p \cdot j/ (K + 1)$ for the $j$-th simultaneous intervention (i.e., Poisson eq offset).}
	\label{fig:simul-exp}
\end{figure}

\subsection{Periodic Event Density} \label{subsec:density} 
Figure~\ref{fig:density} shows results from a periodic density using a fixed duration switchback. When the design has a period that aligns with density, the offset parameter $q$ determines how the alignment alters the bias and variance. Switching at times of high density yields a design with low variance from measurement errors $\var(\mathcal{E}_\mathrm{meas}) $. On the other hand, switching at times of low density reduces carryover bias from the preceding interval $\bias(\mathcal{E}_\mathrm{carryover}) $ and reduces MSE from treatment effects. Combining measurement errors and treatment effects, the optimal switching times are somewhere in between high and low density times. Therefore, knowledge of the density of events can improve the efficiency of the design by leveraging the best absolute times for bias- or variance-minimizing switching points.

\begin{figure}[h]
	\centering
  \includegraphics[width=1\linewidth]{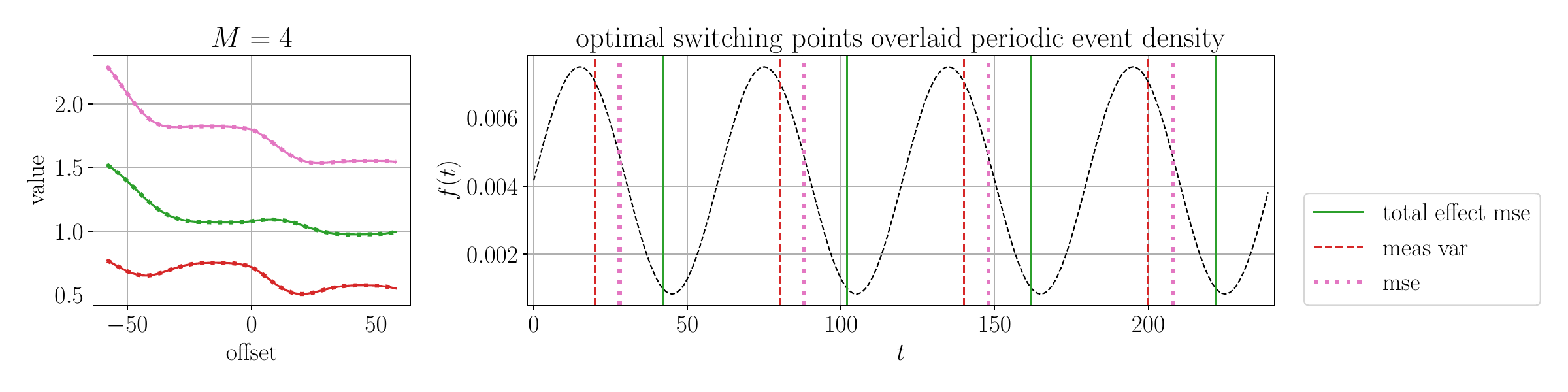}
	\caption{Role of offset parameter $q$ in periodic event density. For the simplicity of exposition, $T = 240$ and $M = 4$, while other parameters are set at the base level.}
	\label{fig:density}
\end{figure}

%% file: section_6.tex
\section{Discussion and Conclusion}\label{sec:discussion}

This paper studies the design and analysis of simultaneous switchback experiments from both a theoretical and practical point of view. We provide a theoretical analysis of how the bias and variance of the HT estimator of the GATE are affected by four factors: carryover effects, periodicity, correlations in event outcomes, and effects of simultaneous interventions. Empirical and simulation studies show how these factors trade off each other and provide insights into how one can design efficient switchback experiments. Based on these insights, we propose an empirical Bayes design that utilizes historical data to model and select the parameters of the new experiment's design.

Perhaps the most general conclusion we can draw is that designing experiments in this context involves considering a complex set of tradeoffs and critically depends on the assumptions experimenters would make using prior knowledge. Our case study on a ride-sharing platform illustrates this point, showcasing the value of leveraging prior experiments in the empirical Bayes approach when designing new experiments.  

While we motivate this study by applications in the ride-sharing setting, the theory and practical guidelines presented can find broader applications in other contexts. Indeed, various settings exist where cross-sectional interventions are not possible or outcomes cannot be easily attributed to treatment decisions. Estimating the effectiveness of traditional media advertising aligns well with our problem setup, and a privacy-friendly approach to online advertising might employ temporal variation in campaign spending linked to sales through timestamps only. Prior work has explored time-varying interventions in financial or cryptocurrency markets  \citep{krafft2018experimental} or in self-experimentation for personalized medicine \citep{karkar2016framework}. An important goal of this work is to expand the use of time-based experiments to settings where they are not currently used.

%% file: appendix-empirical.tex
\section{Supplementary Empirical Results}\label{appendix:empirical}

\subsection{Curve Fitting of Natural Cubic Splines }\label{subsec:curve-fitting-details}

We scale the time since the switch from $[0, 56]$ to $[0, 1]$. Let $x$ be the scaled time since the switch. In the natural cubic splines, we specify a knot located in the middle of the interval, i.e., at $k = 1/2$. Then, the natural cubic splines can be written as 
\begin{equation}
    \begin{aligned}
        g(x) =& \bm{1}(x < k) \cdot (a_0  + a_1 x + a_2 x^2 + a_3 x^3) \\
        & + \bm{1}(x \geq k) \cdot (b_0  + b_1 x + b_2 x^2 + b_3 x^3) 
    \end{aligned}
\end{equation}
with eight parameters $a_0, a_1, a_2, a_3, b_0, b_1, b_2, b_3$. These eight parameters are fitted from the estimated cumulative effects with treatment duration from $0$-th to $55$-th minutes. Through this fitting, we reduce the number of parameters from $56$ to $8$, thereby reducing the variance in the treatment effect estimation. Furthermore, we impose additional constraints in the spline fitting based on our prior on the shape of the curve:
\begin{align*}
    g^\prime(1) =&  0 \tag{gradient at the right boundary is zero} \\
    g(k_{-}) =& g(k_{+}) \tag{left and right values are equal at the knot} \\
    g^\prime(k_{-}) =& g^\prime(k_{+}) \tag{left and right gradients are equal at the knot} \\
    g^{\prime\prime}(0) =& 0 \tag{Hessian at the left boundary is zero}
\end{align*}

With the first constraint, the gradient approaches zero close to the right boundary of the interval, so it incorporates our prior that the cumulative effect converges to GATE as the treatment duration increases. With this constraint, we can set the GATE as $g(1)$. The second and third constraints ensure that the value and gradient evaluated at the left and right of the knot are the same. The last constraint ensures that the second-order derivative at the left boundary of the interval is zero, so that the cumulative effect varies linearly with the treatment duration at the left boundary. The last constraint is aligned with the standard constraint of natural cubic splines that the function is linear beyond the boundary knots. 

\subsection{Leave-Two-Out Cross-Validation of Smoothing Methods}\label{subsec:cross-validation}

We perform leave-two-out cross-validation to compare various smoothing methods. The smoothing method can be natural cubic spline, polynomial regression, or local regression. We refer to natural cubic spline as the benchmark method and other methods as alternative methods. For each specific experiment-region pair and a smoothing method, the calculation of the cross-validation error consists of five steps. The first step is to split the two weeks of experimental data into the training and holdout sets. The holdout set consists of two intervals, while the training set has all the remaining intervals. As the experiment uses the balanced design with a fixed switching duration of $56$ minutes and lasts for two weeks, for a given $i$, the treatment assignment of the $i$-th interval in the first week differs from that of the $i$-th interval in the second week. We choose these two intervals in the holdout set, so that CEC can be properly estimated using the observations in the holdout set. The second step is to estimate the CEC, i.e., $\delta^\cum_t(\Delta t)$ for $\Delta t= \{1, \cdots, 56\}$, separately on the training and holdout sets. We let the estimated CEC on the training and holdout sets be $\hat\delta^{\cum,\mathrm{train}}(\Delta t)$ and $\hat\delta^{\cum,\mathrm{holdout}}(\Delta t)$ for $\Delta t= \{1, \cdots, 56\}$.\footnote{The estimated CEC is averaged over $t$.} The third step is to apply the smoothing method to the estimated CEC $\hat\delta^{\cum,\mathrm{train}}(\Delta t)$ on the training set. Let the effect from the smoothing method be $\hat{g}(\Delta t)$. The fourth step is to calculate the estimation error of GATE on the holdout set, which is $\hat\delta^{\cum,\mathrm{holdout}}(\Delta t) - \hat{g}(\Delta t)$ for $\Delta t = 56$-th minute. In this step, $\hat{g}(56)$, the value at the last minute of an interval, is an estimate of GATE on the training set, and $\hat\delta^{\cum,\mathrm{holdout}}(56)$ is the GATE on the holdout set. In the last step, we iterate the first to fourth steps by varying the intervals in the holdout set and calculate the estimation error of GATE $\hat\delta^{\cum,\mathrm{holdout}}(56) - \hat{g}(56)$. We calculate the mean squared error (MSE), which is the average of $(\hat\delta^{\cum,\mathrm{holdout}}(56) - \hat{g}(56))^2$ over the holdout intervals. This MSE is referred to as the leave-two-out cross-validation MSE. 

The histograms in Figure \ref{fig:model-comparison} show the distribution of the ratio of the leave-two-out cross-validation MSEs over all experiment-region pairs. The denominator of the ratio is the MSE using the natural cubic spline. The numerator of the ratio is the MSE using an alternative method, such as polynomial regression of degree $d$ for $d \in \{0, 1, 2, 3, 4\}$ or local regression of degree $d \in \{0,1,2\}$. 

There are four findings from the leave-two-out cross-validation. First and foremost, the ratio of MSE lies between $0.98$ and $1.02$ for most of the experiment-region pairs and for all alternative methods. This implies that the natural cubic spline is comparable to any alternative method. Second, the cross-validation MSE increases with model flexibility. This makes sense as the estimation of GATE has a high variance, and a restrictive model can effectively reduce the variance. Third, the natural cubic spline has a smaller MSE than the polynomial regression of degree $3$. This finding is interesting as imposing a prior in the natural cubic spline can generally reduce the out-of-sample error. Fourth, if we allow for non-constant and non-monotonic smoothed CEC, natural cubic spline has the lowest cross-validation MSE, as compared to polynomial regression of degree $d \geq 2$ and local regression of degree $d \geq 1$.

We additionally note that, among all the smoothing methods, only the smoothed effect from natural cubic spline, polynomial or local regression of degree $0$, and natural cubic spline at the last minute can have a natural interpretation of GATE. This is because the first-order derivative of these three smoothing methods at the right boundary (i.e., the last minute) is zero. For all the other methods, the first-order derivative is nonzero, and extrapolated effect after the $56$-th minute would be different from that of the $56$-th minute. This would make it unnatural to use the effect of the $56$-th minute as an estimate of GATE. If polynomial or local regression of degree $0$ is used for smoothing, then it is likely that a fast-switching design would have a low MSE. However, we focus on the natural cubic spline in the main text because it can capture richer dynamics of the cumulative effects. In addition, from an adversarial perspective, the natural cubic spline and long-switching design are more robust to complex dynamics of the cumulative effects. They are likely to yield a low-bias estimate of GATE, possibly at the expense of a slight increase in variance. However, the increase in variance is small, as the MSE of natural cubic spline is generally within 2\% more than the MSE of polynomial or local regression with degree $0$.

\begin{figure}
    \centering
    \includegraphics[width=1\textwidth]{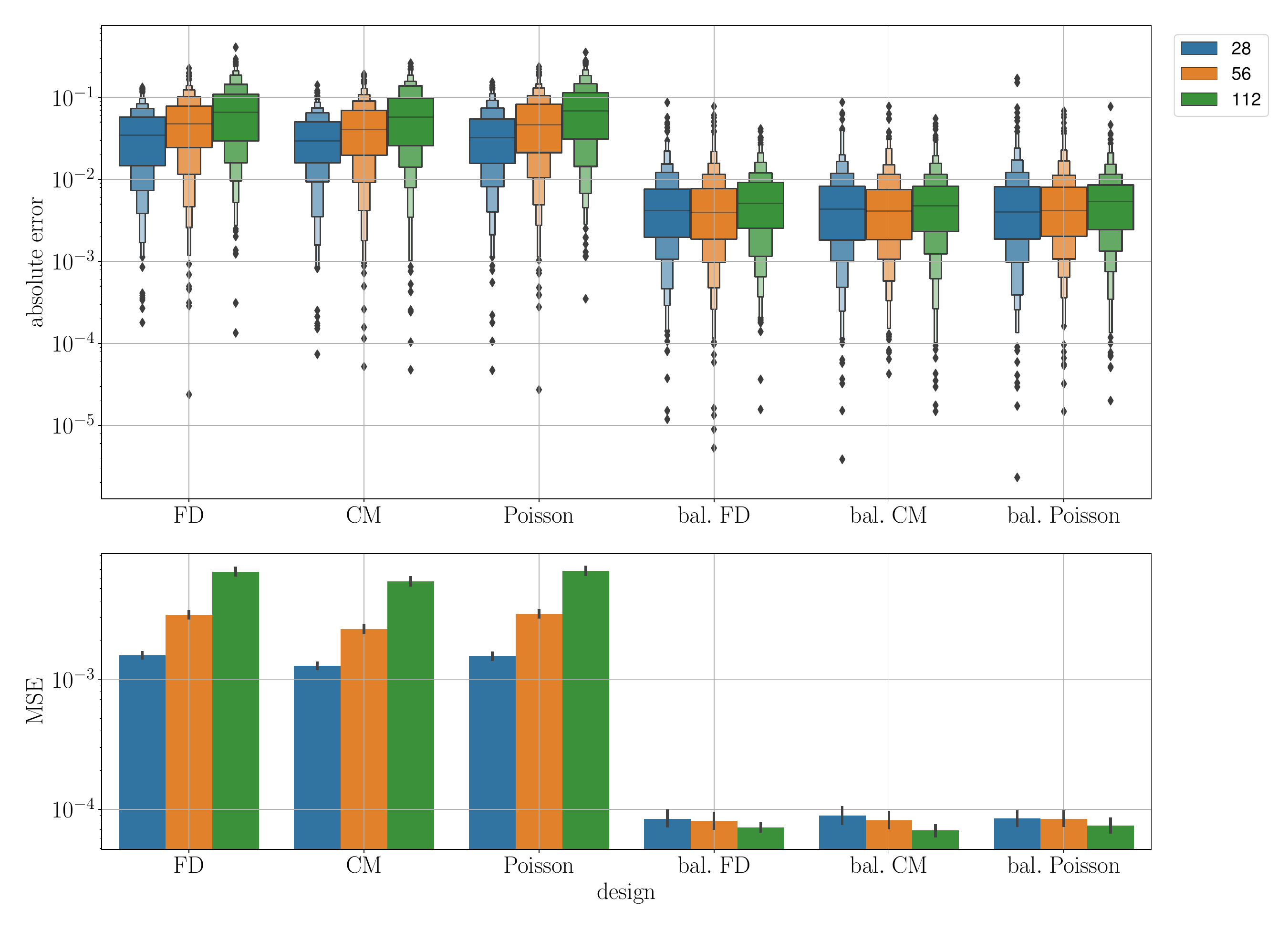}
    \caption{Boxenplot of the absolute estimation error and MSE from two-week synthetic experiments using various designs. The boxenplot is obtained based on $500$ draws of two-week historical control data (serving as the observed data absent of synthetic intervention) and a $56$-minute CEC from the empirical distribution of smoothed CECs (serving as the treatment effects of the synthetic intervention). In this simulation, no experiment is run simultaneously with the primary synthetic experiment. The drawn CECs are identical to those in Figure \ref{fig:syn-exp-fitted-poly2}. The MSE in the second row is the average of squared errors after removing the 1\% outliers of the estimation errors. This is to ensure that the MSE is not dominated by the values of the outliers. }
    \label{fig:56min-mse}
\end{figure}

\begin{figure}
    \centering
    \includegraphics[width=1\textwidth]{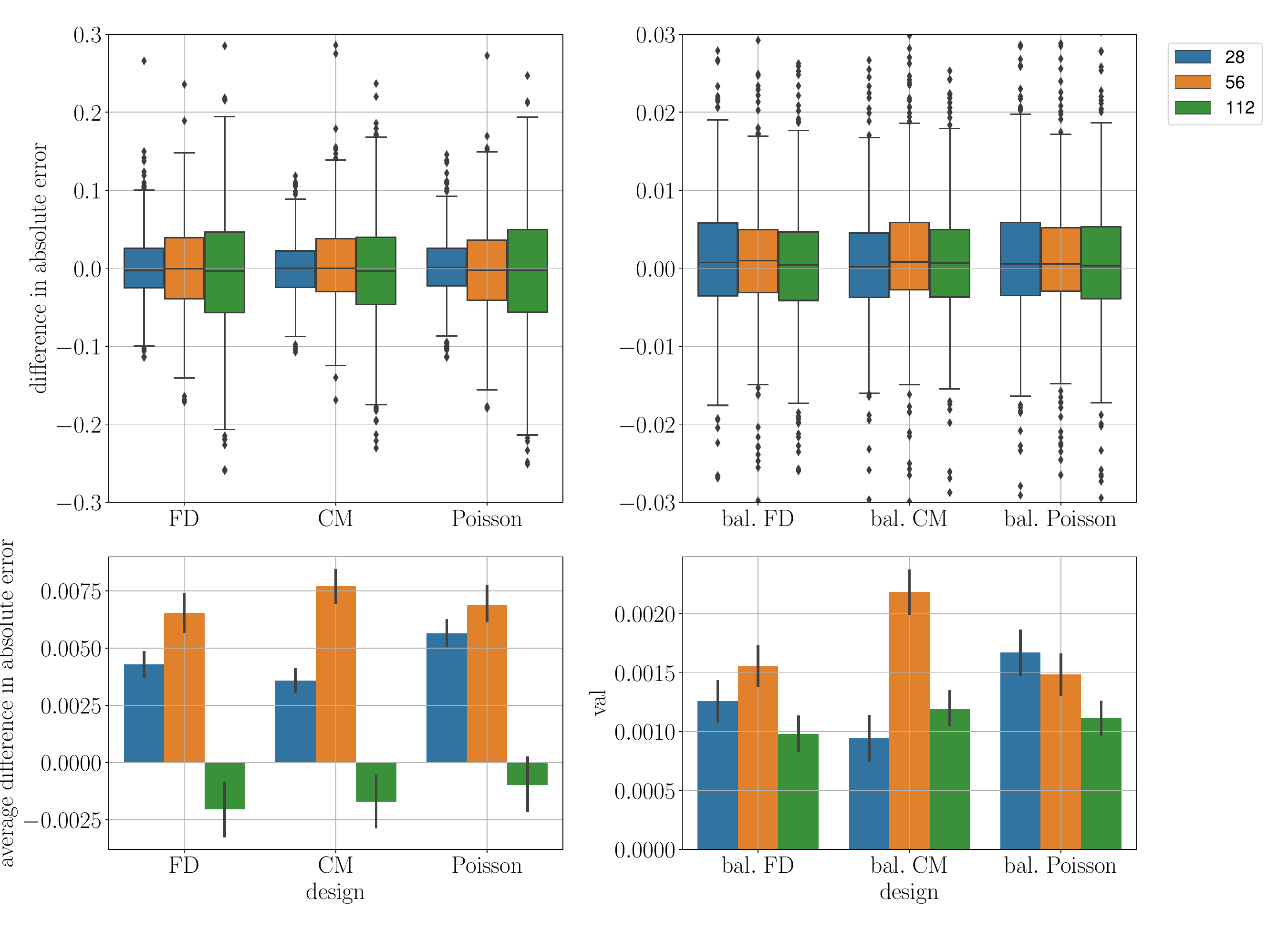}
    \caption{Difference in absolute estimation error when one (in Figure \ref{fig:syn-exp-fitted-poly2}) versus no (in Figure \ref{fig:56min-mse}) experiment is run simultaneously. The average difference in the second row is the average of errors after removing the top 1\% and bottom 1\% outliers of the differences. This is to ensure that the average difference is not dominated by the values of the outliers. }
    \label{fig:56min-diff-mse}
\end{figure}

\begin{figure}[h]
    \centering
    \begin{subfigure}[b]{\textwidth}
        \includegraphics[width=\linewidth]{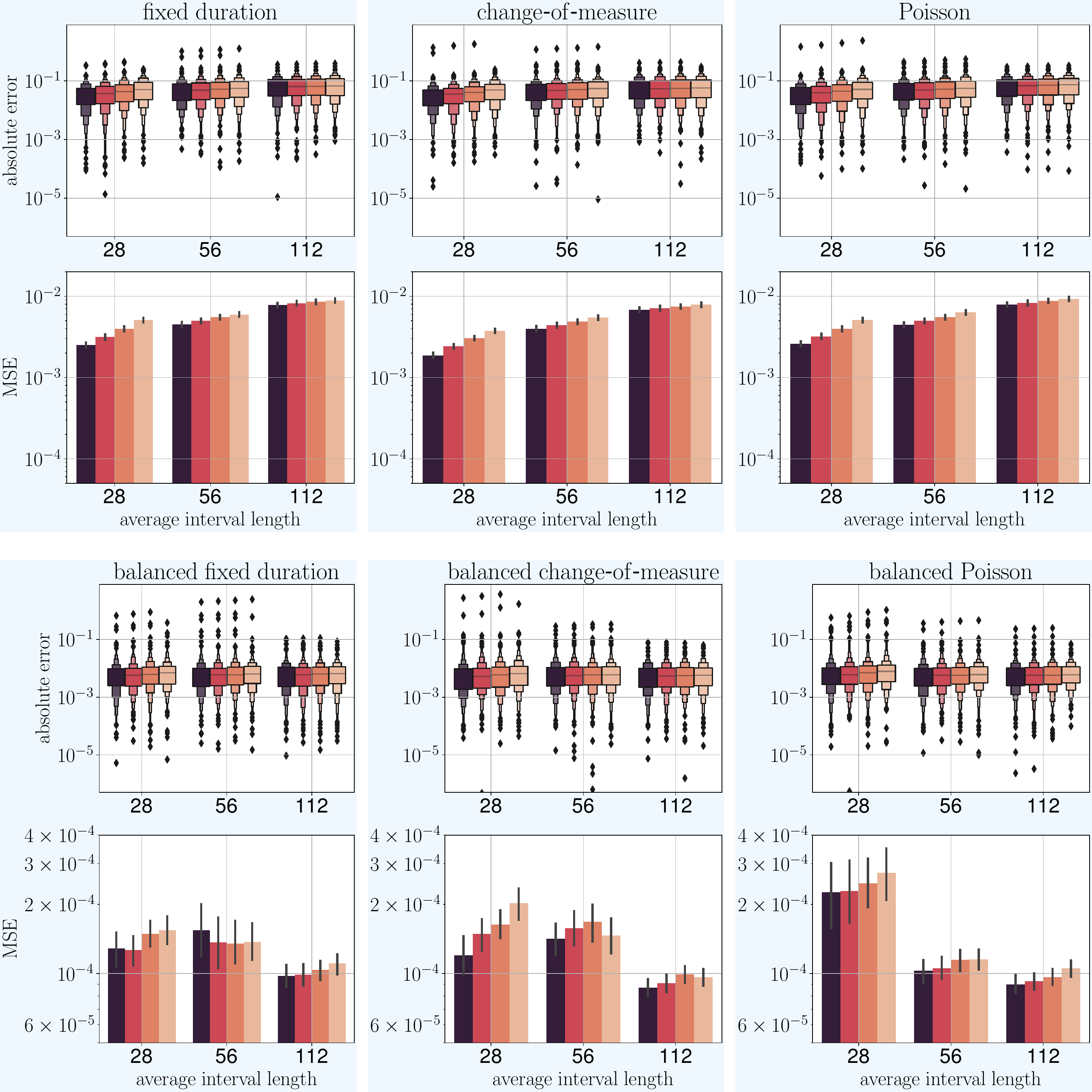}
    \end{subfigure}
    \begin{subfigure}[b]{\textwidth}
    \centering
        \includegraphics[width=0.9\linewidth]{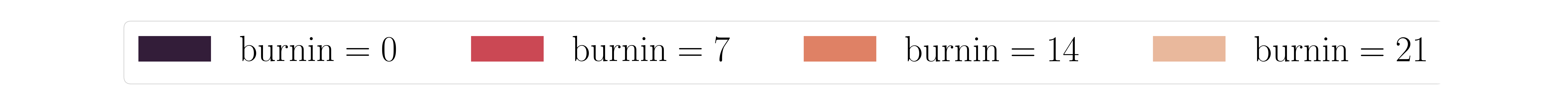}
    \end{subfigure}
    \caption{MSE of Horvitz–Thompson estimators with various burn-ins (i.e., $h$ in Estimator \eqref{eqn:ht-estimator-burnin}) in simulated experiments on \textbf{historical experimental data} (i.e., one experiment running simultaneously), where the CEC is drawn from the empirical distribution of smoothed $56$-minute CECs.
   }
    \label{fig:alternative-estimator}
\end{figure}

\begin{figure}
    \centering
    \includegraphics[width=\textwidth]{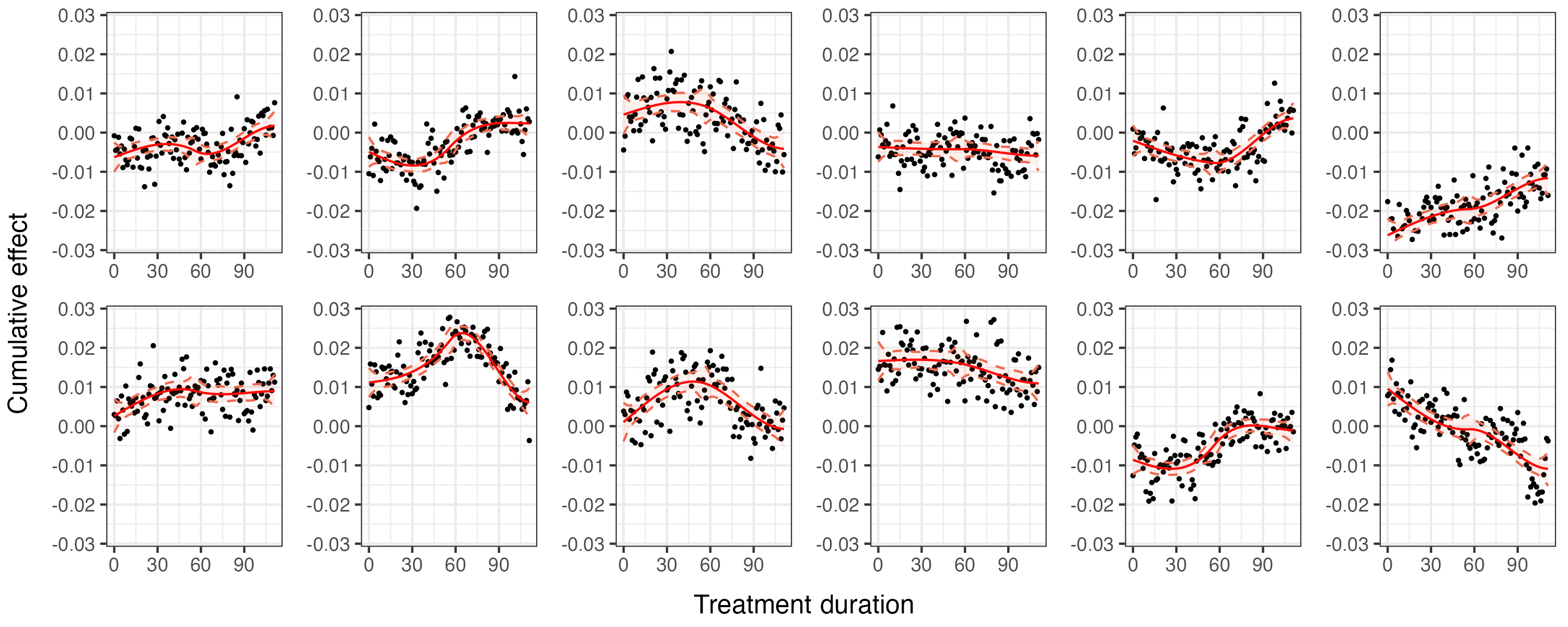}
    \caption{$12$ representative CECs for the treatment duration of $\{1, \cdots, 112\}$ minutes (in black dots), and their smooth curves by natural cubic splines (in red). The shaded area for each curve shows the $95\%$ confidence interval of the cumulative effect given a treatment duration and the natural cubic spline as the smoothing method.}
    \label{fig:112-various-poly}
\end{figure}

\begin{figure}
    \centering
    \includegraphics[width=1\textwidth]{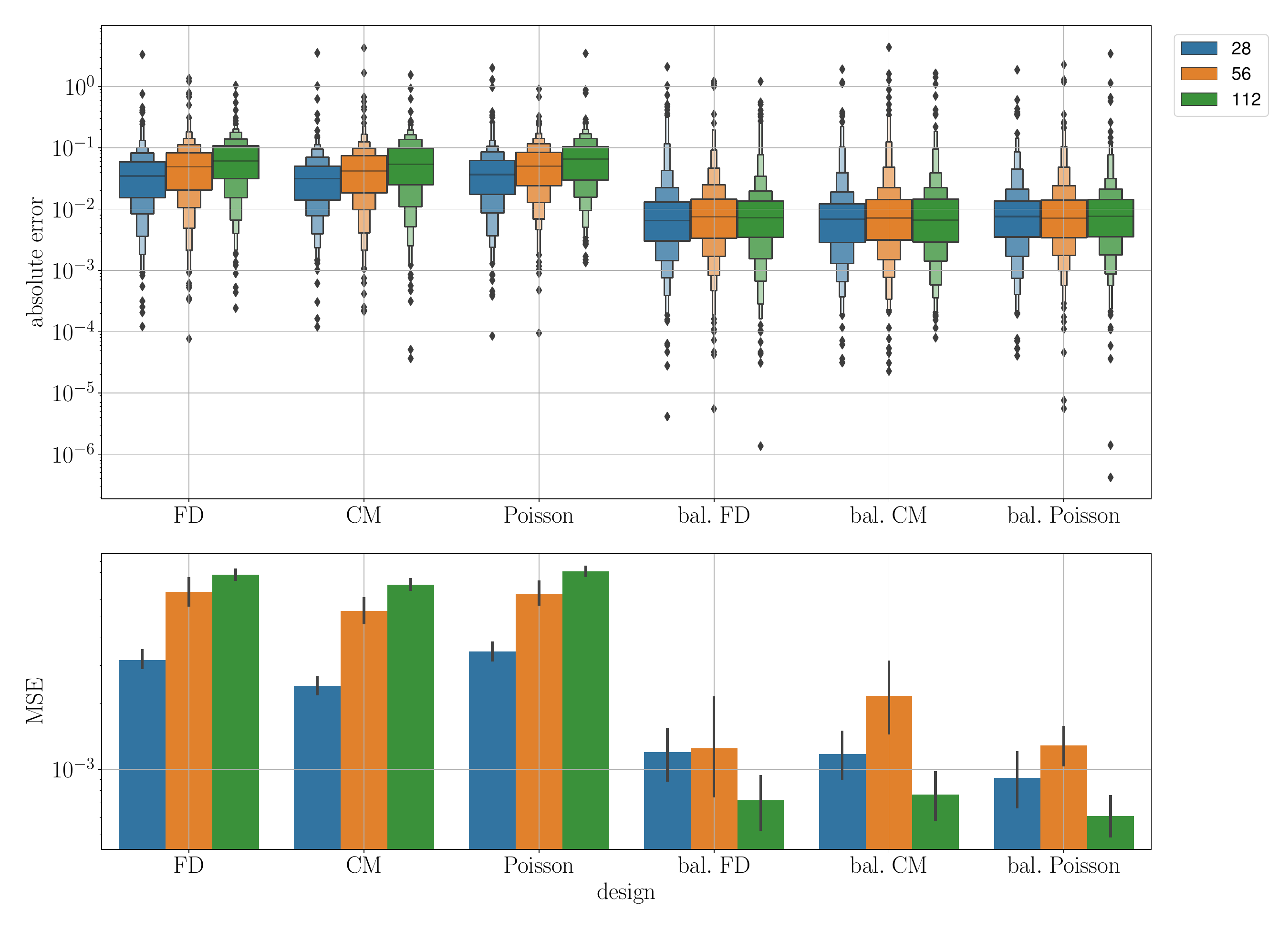}
    \caption{Boxenplot of the absolute estimation error and MSE from two-week synthetic experiments using various designs. The boxenplot is obtained based on $500$ draws of two-week historical data and a \textit{$112$-minute CEC} from the empirical distribution of smoothed CECs. In this simulation, one experiment starts and ends simultaneously with the primary synthetic experiment. The MSE in the second row is the average of squared errors after removing the 1\% outliers of the estimation errors. This is to ensure that the MSE is not dominated by the values of the outliers. }
    \label{fig:112min-mse}
\end{figure}

\begin{figure}[h]
	\centering
 \includegraphics[width=1\textwidth]{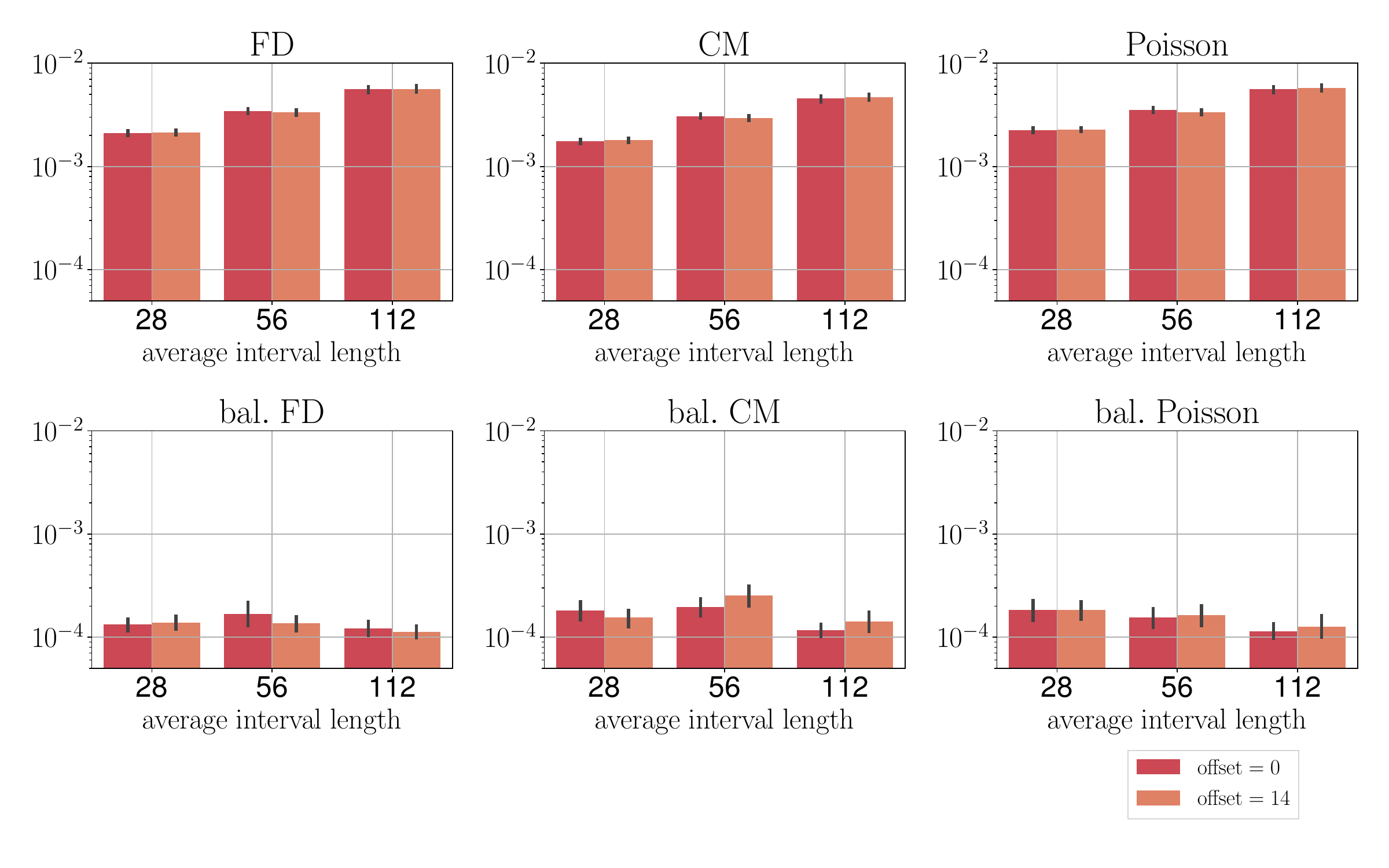}
    \caption{Effect of offset of switching points on estimation error when one experiment is run simultaneously. Offsetting is the most helpful when both the primary and simultaneous experiments use the $56$-minute fixed duration switchbacks. Otherwise, offsetting may not help to reduce the estimation error. }
    \label{fig:aggregate-comparison}
\end{figure}

\begin{figure}[h]
	\centering
 \includegraphics[width=0.9\textwidth]{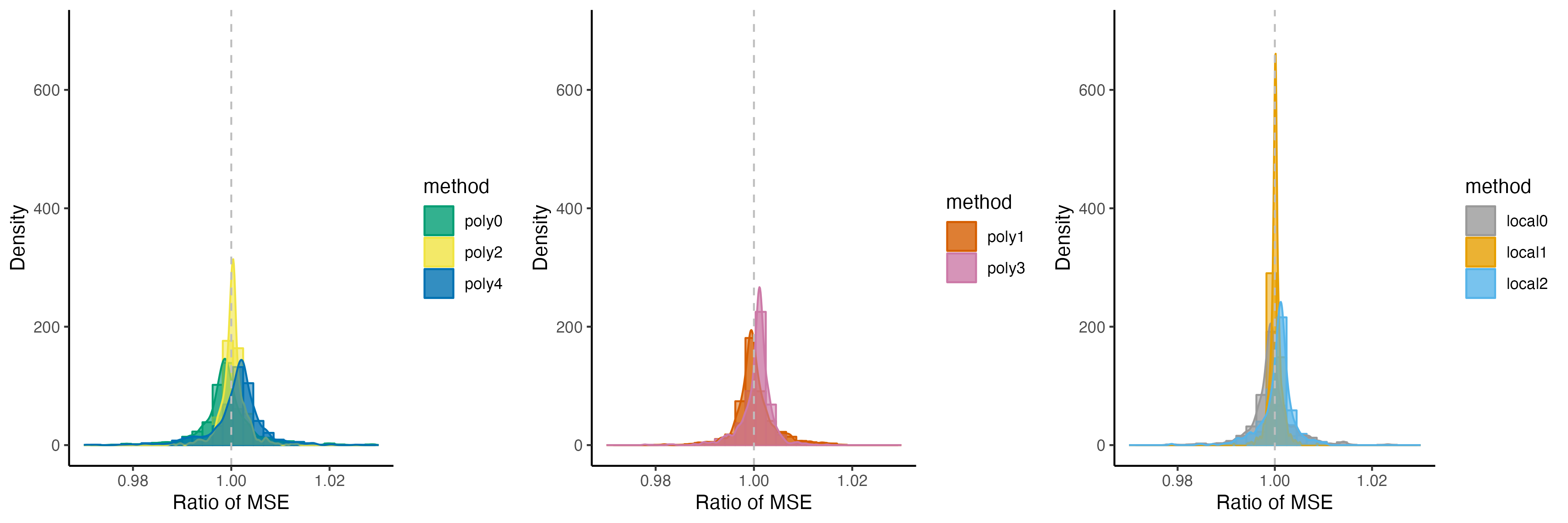}
    \caption{Leave-two-out cross-validation mean-squared error of various smoothing methods. The histogram shows the ratio of cross-validation MSE between an alternative method and natural cubic splines. The alternative method includes polynomial regression of degree $d$ for $d\in\{0,1,2,3,4\}$ and local regression of degree $d$ for $d\in\{0,1,2\}$.
    }
    \label{fig:model-comparison}
\end{figure}

%% file: appendix-theoretical.tex
\clearpage

    \section{Supplementary Theoretical Results}\label{sec:appendix-switchback}

    \subsection{Additional Treatment Effect Estimands}

    We additionally define the average instantaneous and carryover effects, which are building blocks of GATE. The average instantaneous effect $\delta^\inst$ is defined as
\begin{align*}
    \delta^\inst  =& \int \delta^\inst_{t} f(t) dt \, ,
\end{align*}
where $\delta^\inst_{t}$ is the instantaneous treatment effect at time $t$ that is defined as
\begin{align*}
    \delta^\inst_{t} =  Y_{t}(e_{t}, \bm{0}_t, \cdots, \bm{0}_t) - Y_{t}(\bm{0}_t, \bm{0}_t, \cdots, \bm{0}_t) 
\end{align*}
and $e_{t} = (0, \cdots, 0, \underbrace{1}_{\text{time~}t}, 0, \cdots, 0)$ is a one-hot-encoded vector with the entry of time  $t$ to be $1$ and all the remaining entries to be $0$.

The average carryover effect  $\delta^\co_\ell(\bm{w}) $, given treatment assignments $\bm{w}$, is defined as 
\begin{align*}
    \delta^\co(\bm{w})  =& \int \delta^\co_{t}(\bm{w}_t) f(t) dt \, ,
\end{align*}
where $\delta^\co_{t}(\bm{w}_t)$ is the carryover effect at time $t$ that is defined as
\begin{align*}
    \delta^\co_{t}(\bm{w}_t) =  Y_{t}(\bm{w}_t, \bm{0}_t,  \cdots, \bm{0}_t) - Y_{t}(\bm{w}_t \circ e_{t}, \bm{0}_t, \cdots, \bm{0}_t) 
\end{align*}
and ``$\circ$'' denotes the entry-wise product. 
Let $\delta^\co \coloneqq \delta^\co(\bm{1})$ be the average carryover effect under global treatment.
Then we can decompose the GATE as
\[\delta^\gate = \delta^\inst + \delta^\co. \]

Given the instantaneous and carryover effects, we first introduce a few more notations to measure the heterogeneity in treatment effects at the interval level (analogous to $\Xi^{(m)}$).
Let
\begin{align*}
        \Xi^{\inst,(m)} =& \int_{t \in \mathcal{I}_{m}} \delta^\inst_{t}   f(t) d t \\
        \Xi^{\inst,(m)}_\dem =& \Xi^{\inst,(m)} - \delta^\inst \mu^{(m)} \, ,
    \end{align*}
    where $\Xi^{\inst,(m)} $ measures the integrated instantaneous effect $\delta^\inst_{t}$ over time $t$ in the interval $\mathcal{I}_{m}$, and $\Xi^{\inst,(m)}_{\dem}$ measures the deviation of $\Xi^{\inst,(m)} $  from the average instantaneous effect $\delta^\inst$ scaled by the event occurrence probability in the interval $\mathcal{I}_{m}$ with $\sum_{m = 1}^M \Xi^{\inst,(m)}_\dem = 0$ (``dem'' is the abbreviation of ``demeaned''). 
    Analogously, let 
    \begin{align*}
        \Xi^{\co,(m)} =& \int_{t \in \mathcal{I}_{m}} \delta^\co_{t} f(t) d t \\
        \Xi^{\co,(m)}_\dem =& \Xi^{\co,(m)} - \delta^\co \mu^{(m)} \, ,
    \end{align*}
    where $\Xi^{\co,(m)} $ measures the integrated carryover effect $\delta^\co_{t}$ over time $t$ in the interval $\mathcal{I}_{m}$, and $\Xi^{\co,(m)}_{\dem}$ measures the deviation of $\Xi^{\co,(m)} $  from the average carryover effect $\delta^\co$ scaled by the event occurrence probability in the interval $\mathcal{I}_{m}$ with $\sum_{m = 1}^M \Xi^{\co,(m)}_\dem = 0$. Then we can decompose $\Xi^{(m)}$ as
    \[\Xi^{(m)} = \Xi^{\inst,(m)} +  \Xi^{\co,(m)} \, .\]

\subsection{Generalization to More than One Simultaneous Experiment}\label{subsec:multi-simul-exp}

 From this subsection onward, we consider a more general setting with $K$ simultaneous experiments, where $K$ can be any integer. Then we use the subscript $1, \cdots, K$ in $\bm{W}^\s_1, \cdots, \bm{W}^\s_K$ to denote the index of the simultaneous experiment. 
 We further generalize the definition of potential outcomes $Y_t(\bm{W}_t, \bm{W}^\s)$ in Section \ref{sec:setup} to 
 \[Y_t(\bm{W}_t, \bm{W}^\s_1, \cdots, \bm{W}^\s_K) \]
to accommodate more than one simultaneous experiment. In addition, we let $W_\ell^{\s(m)}$ be the treatment assignment of the $m$-th interval of the $\ell$-th simultaneous intervention. We also assume that the assumptions in Theorem \ref{theorem:switchback-bias} and \ref{theorem:bias-variance-switchback} hold for more than one simultaneous experiment. In other words, 
$W_\ell^{\s(m)}$ is independent in $m$ with $ \P(W^{\s(m)} = 1) = 1/2$, and $\bm{W}$ is independent of $\bm{W}^\s_\ell$ for all $\ell$. 

\paragraph{Generalization of Theorem \ref{theorem:switchback-bias}} Theorem \ref{theorem:switchback-bias} continues to hold for more than one simultaneous experiment if we use a more general definition of $S^{(m)}$ in the bias of the HT estimator. Specifically, the term $\Phi^\simul_t$ used to define $S^{(m)}$ needs to allow for $K$ simultaneous experiments, that is, 
\[\Phi^\simul_t = \+E_{\bm{W}^{(-m)} }\big[ \delta^\simul_t(\bm{W}^{(-m)},  W^{(m)} = 1 ) \big]  - \+E_{\bm{W}^{(-m)} }\big[ \delta^\simul_t(\bm{W}^{(-m)},  W^{(m)} = 0 ) 
        \big] \, \]
    where 
\[\delta^\simul_t(\bm{W}_t) = \+E_{\bm{W}^\s_t}\left[Y_t(\bm{W}_t, \bm{W}^\s_{1t}, \cdots, \bm{W}^\s_{Kt}) - Y_t(\bm{W}_t, \bm{0}_t, \cdots, \bm{0}_t) \mid \bm{W}_t, t \right] \]

\paragraph{Generalization of Theorem \ref{theorem:bias-variance-switchback}} In Section \ref{subsec:notations}, we provide the definition of two terms related to simultaneous experiments, $\+E[\mathcal{E}_\mathrm{simul}^2]$ and $\+E[(\mathcal{E}_\mathrm{inst}+\mathcal{E}_\mathrm{carryover}) \cdot \mathcal{E}_\mathrm{simul}] $ in terms of $K$ simultaneous experiments. The definition holds for both $K = 1$ and $K > 1$. With this definition, Theorem \ref{theorem:bias-variance-switchback} holds for more than one simultaneous experiment.

\paragraph{Generalization of Proposition \ref{prop:additive-one-simul-effect}}

 We state a more general form of Proposition \ref{prop:additive-one-simul-effect}, which allows for more than one simultaneous experiment. 

          \begin{proposition}\label{prop:additive-simul-effect}
        Under the assumptions in Theorem \ref{theorem:bias-variance-switchback}, if Condition \ref{cond:additive} holds and the effect of each simultaneous intervention satisfies Assumption \ref{ass:interference-structure}, then the bias from simultaneous interventions $\bias(\mathcal{E}_\mathrm{simul}) $ is zero
         and the variance from simultaneous interventions is 
             \begin{align*}
         &\+E[\mathcal{E}_\mathrm{simul}^2] 
         =  \sum_{m = 1}^M  \bigg( 
\int_{t\in \mathcal{I}_{m}} \bigg[\sum_{\ell=1}^K 
 \delta^{\s.\gate}_{\ell,t} \bigg]  f(t) dt \bigg)^2   \\ & + \sum_{m = 1}^M  \sum_{m^\prime =1}^M \sum_{\ell=1}^K \bigg( 
\int_{t\in \mathcal{I}_{m} \cap \mathcal{I}^\s_{\ell m^\prime} }  \delta^{\s.\inst}_{\ell,t}  f(t) dt + \int_{t \in \mathcal{I}_{m}, t^\prime \in \mathcal{I}^\s_{\ell m^\prime} }  \delta^{\s.\co}_{\ell, t} d^{\s.\co}_{\ell,t}(t^\prime)  f(t) f(t^\prime)  dt d t^\prime\bigg)^2 \, ,
    \end{align*}
    and
    \begin{align*}
     \+E[(\mathcal{E}_\mathrm{inst}+\mathcal{E}_\mathrm{carryover}) \cdot \mathcal{E}_\mathrm{simul}] =& \sum_{m=1}^M \left(\Xi^{(m)} + 2 \mu^{(m)}_{Y^\mathrm{ctrl}}\right)\bigg( 
\int_{t\in \mathcal{I}_{m}} \bigg[\sum_{\ell =1}^K 
 \delta^{\s.\gate}_{\ell,t} \bigg]  f(t) dt \bigg) \, ,
 \end{align*}
 where $\delta^{\s.\gate}_{\ell,t}$, $\delta^{\s.\inst}_{\ell,t}$, $\delta^{\s.\co}_{\ell,t}$, $d^{\s.\co}_{\ell,t}$ are the total treatment effect, instantaneous effect, carryover effect, and carryover kernel of simultaneous intervention $\ell$ at time $t$, respectively, and $\mathcal{I}^\s_{\ell m}$ is the $m$-th interval of simultaneous intervention $\ell$. 
    \end{proposition}

 \subsection{Definition of Notations in Theorem \ref{theorem:bias-variance-switchback}}\label{subsec:notations}

 \paragraph{The expression of $S^{(m,m^\prime)}_{\mathrm{var}}$ in $\+E[\mathcal{E}_\mathrm{simul}^2]$} \texttt{} 

$S^{(m,m^\prime)}_{\mathrm{var}}$ is defined as 
 \begin{equation}\label{eqn:S-m-mprime-var}
 \begin{aligned}
     S^{(m,m^\prime)}_{\mathrm{var}} =& 4  \int_{t \in \mathcal{I}_{m}, t^\prime \in \mathcal{I}_{m^\prime}}  \left(\bm{1}\left(m = m^\prime \right) \+E_{W}[\delta^{\simul,2}_{t, t^\prime}(\bm{W}) \mid t, t^\prime] + \bm{1}\left(m \neq m^\prime \right) \Phi^{2\dagger}_{t,t^\prime} \right) f(t) f(t^\prime) dt dt^\prime \, ,
 \end{aligned}
 \end{equation}
 where $\delta^{\simul,2}_{t, t^\prime}(\bm{W})$ is equal to 
 \begin{equation}\label{eqn:delta-simul-2}
     \begin{aligned}
        \delta^{\simul,2}_{ t, t^\prime}(\bm{W}) = \+E_{\bm{W}^\s_1, \cdots, \bm{W}^\s_K } &\left[ \left(Y_{t}(\bm{W}, \bm{W}^\s_1, \cdots, \bm{W}^\s_K) - Y_{t}(\bm{W}, \bm{0}, \cdots, \bm{0})\right) \times \right.  \\ & \left.  \left(Y_{t^\prime}(\bm{W}, \bm{W}^\s_1, \cdots, \bm{W}^\s_K) - Y_{t^\prime}(\bm{W}, \bm{0}, \cdots, \bm{0})\right)  \mid \bm{W}, t, t^\prime \right] \, ,
    \end{aligned}
 \end{equation}
    and for $t \in \mathcal{I}_{m}$ and $t^\prime \in \mathcal{I}_{m^\prime}$ with $m \neq m^\prime$,  $\Phi^{2\dagger}_{t,t^\prime}$ is equal to
    \begin{equation}\label{eqn:phi-2}
        \begin{aligned}
        \Phi^{2\dagger}_{t,t^\prime} =& \frac{1}{4}\bigg(\+E_{\bm{W}^{(-m,-m^\prime)} }\bigg[ \delta^{\simul,2}_{t, t^\prime}(\bm{W}^{(-m,-m^\prime)},  \bm{W}^{(m,m^\prime)} = (1,1) ) \bigg]  \\
        & - \+E_{\bm{W}^{(-m,-m^\prime)} }\bigg[ \delta^{\simul,2}_{t, t^\prime}(\bm{W}^{(-m,-m^\prime)},  \bm{W}^{(m,m^\prime)} = (1,0) ) \bigg] \\
        & - \+E_{\bm{W}^{(-m,-m^\prime)} }\bigg[ \delta^{\simul,2}_{ t, t^\prime}(\bm{W}^{(-m,-m^\prime)},  \bm{W}^{(m,m^\prime)} = (0,1) ) \bigg] \\
        & + \+E_{\bm{W}^{(-m,-m^\prime)} }\bigg[ \delta^{\simul,2}_{ t, t^\prime}(\bm{W}^{(-m,-m^\prime)},  \bm{W}^{(m,m^\prime)} = (0,0) ) \bigg] \bigg) \, .
    \end{aligned}
    \end{equation}

    The term $\delta^{\simul,2}_{t, t^\prime}(\bm{W})$ is the expected value of the product between the effects of simultaneous interventions at time $t$ and time $t^\prime$, conditional on $\bm{W}$.  The term $\Phi^{2\dagger}_{t,t^\prime} $ then measures the discrepancy in $\delta^{\simul,2}_{t, t^\prime}(\bm{W})$ by varying the values of $W^{(m)}$ and $W^{(m^\prime)}$ and marginalizing over $\bm{W}^{(-m,-m^\prime)}$.

    \paragraph{The expression of $S^{(m,m^\prime)}_{\mathrm{cov}}$ in $\+E[(\mathcal{E}_\mathrm{inst}+\mathcal{E}_\mathrm{carryover}) \cdot \mathcal{E}_\mathrm{simul}] $} \texttt{} 

    $S^{(m,m^\prime)}_{\mathrm{cov}}$ is defined as 
    \begin{align}\label{eqn:S-m-mprime-cov}
        S^{(m,m^\prime)}_{\mathrm{cov}} = (\delta^\gate \mu^{(m^\prime)} + 2 \mu^{(m^\prime)}_{Y^\mathrm{ctrl}}) \cdot S^{(m,m^\prime)}_{1}  + \left( \Xi^{\inst,(m)}_\dem - \delta^\co \mu^{(m)} \right) S^{(m,m^\prime)}_{2} + S^{(m,m^\prime)}_{3} \, .
    \end{align}

    The term $S^{(m,m^\prime)}_{1}$ is defined as 
    \begin{align}\label{eqn:S-m-mprime-1}
        S^{(m,m^\prime)}_{1} = 2 \int_{t \in \mathcal{I}_{m}}  \left( \bm{1}\left(m = m^\prime \right)  \+E_{W}\Big[ \delta^\simul_{ t}(\bm{W}) 
        \Big] + \bm{1}\left(m \neq m^\prime \right)  \Phi_{t}^{\simul,(-m^\prime)} \right) f(t) dt \, ,
    \end{align}
where $\+E_{W}\big[ \delta^\simul_{ t}(\bm{W}) \big]$ is defined in Section \ref{subsec:block-stats} and, for time $t \in \mathcal{I}_{m}$ and time $t^\prime \in \mathcal{I}_{m^\prime}$ with $m \neq m^\prime$, $\Phi_{t}^{\simul,(-m^\prime)}$ is defined as 
\begin{equation}\label{eqn:phi-simul}
    \begin{aligned}
        \Phi_{t}^{\simul,(-m^\prime)} =&  \frac{1}{4}\bigg(\+E_{\bm{W}^{(-m,-m^\prime)} }\bigg[ \delta^\simul_{ t}(\bm{W}^{(-m,-m^\prime)},  \bm{W}^{(m,m^\prime)} = (1,1) ) \bigg]  \\
        & - \+E_{\bm{W}^{(-m,-m^\prime)} }\bigg[ \delta^\simul_{ t}(\bm{W}^{(-m,-m^\prime)},  \bm{W}^{(m,m^\prime)} = (1,0) ) \bigg] \\
        & - \+E_{\bm{W}^{(-m,-m^\prime)} }\bigg[ \delta^\simul_{ t}(\bm{W}^{(-m,-m^\prime)},  \bm{W}^{(m,m^\prime)} = (0,1) ) \bigg] \\
        & + \+E_{\bm{W}^{(-m,-m^\prime)} }\bigg[ \delta^\simul_{ t}(\bm{W}^{(-m,-m^\prime)},  \bm{W}^{(m,m^\prime)} = (0,0) ) \bigg] \bigg) \, .
    \end{aligned}
\end{equation}
Here 
\begin{align*}
        \bm{W}^{(m,m^\prime)} =&~ (W^{(m)}, W^{(m^\prime)}), \qquad
        \bm{W}^{(-m,-m^\prime)} =~ \bm{W} \backslash \bm{W}^{(m,m^\prime)}  \, ,
    \end{align*}
    where $\bm{W}^{(-m,-m^\prime)}$ is an $M-2$ dimensional vector denoting the treatment status of the primary intervention for all intervals excluding the $m$-th and $m^\prime$-th intervals.
 
    Recall that $\delta^\simul_{ t}(\bm{W}) $ is the expected value of the effects of simultaneous interventions at time $t$, conditional on $\bm{W}$. Then $\+E_{W }\big[ \delta^\simul_{ t}(\bm{W}) \big]$ is the expected value of $\delta^\simul_{ t}(\bm{W}) $ marginalizing over $\bm{W}$. 
    The term $\+E_{\bm{W}^{(-m,-m^\prime)} }\big[ \delta^\simul_{ t}(\bm{W}^{(-m,-m^\prime)},  \bm{W}^{(m,m^\prime)} ) \big]$ is the expected value of $\delta^\simul_{ t}(\bm{W}) $ conditional on $\bm{W}^{(m,m^\prime)}$ and marginalizing over $\bm{W}^{(-m,-m^\prime)}$.

    $\Phi_{t}^{\simul,(-m^\prime)}$ is a measure of the discrepancy in the effects of simultaneous interventions by varying the values of $W^{(m)}$ and $W^{(m^\prime)}$   ($m$ is the interval in which time $t$ lies). $\Phi_{t}^{\simul,(-m^\prime)}$ is closely connected to $\Phi^\simul_{t}$ defined in Section \ref{subsec:block-stats}, but the difference is that $\Phi_{t}^{\simul,(-m^\prime)}$ varies the value of both  $W^{(m)}$ and $W^{(m^\prime)}$, while $\Phi^\simul_{t}$ only varies the value of $W^{(m)}$.

   In addition, the term $S^{(m,m^\prime)}_{2}$ is defined as 
 \begin{equation}\label{eqn:S-m-mprime-2}
 \begin{aligned}
     S^{(m,m^\prime)}_{2} =& 2  \int_{t^\prime \in \mathcal{I}_{m^\prime}}\left(\bm{1}\left(m = m^\prime \right) \+E_{\bm{W}^{(-m)} }\big[ \delta^\simul_{t^\prime}(\bm{W}^{(-m)},  W^{(m)} = 1 ) \big]  + \bm{1}\left(m \neq m^\prime \right) \Phi_{t^\prime}^{\simul,(-m^\prime)\dagger} \right) f(t^\prime) dt^\prime \, ,
 \end{aligned}
 \end{equation}
  where for time $t^\prime \in \mathcal{I}_{m^\prime}$, $\Phi_{t^\prime}^{\simul,(-m^\prime)\dagger}$ is equal to 
    \begin{equation}\label{eqn:phi-simul-2}
        \begin{aligned}
        \Phi_{t^\prime}^{\simul,(-m^\prime)\dagger} =& \frac{1}{2}\bigg(\+E_{\bm{W}^{(-m,-m^\prime)} }\bigg[ \delta^{\dagger}_{  t^\prime}(\bm{W}^{(-m,-m^\prime)},  \bm{W}^{(m,m^\prime)} = (1,1) ) \bigg] \\
        & - \+E_{\bm{W}^{(-m,-m^\prime)} }\bigg[ \delta^{\dagger}_{ t^\prime}(\bm{W}^{(-m,-m^\prime)},  \bm{W}^{(m,m^\prime)} = (1,0) ) \bigg]\bigg) \, .
    \end{aligned}
    \end{equation}
    $S^{(m,m^\prime)}_{2}$ is conceptually very similar to $S^{(m,m^\prime)}_{1}$, but the difference is that $S^{(m,m^\prime)}_{2}$ conditions on $W^{(m)} = 1$, while $S^{(m,m^\prime)}_{1}$ does not.

    Lastly, $S^{(m,m^\prime)}_{3}$ is defined as 
 \begin{equation}\label{eqn:S-m-mprime-3}
 \begin{aligned}
     S^{(m,m^\prime)}_{3} =& 4  \int_{t \in \mathcal{I}_{m}, t^\prime \in \mathcal{I}_{m^\prime}}\left(\bm{1}\left(m = m^\prime \right) \+E_W\left[\delta^\co_{t}(\bm{W})  \delta^\simul_{t^\prime}(\bm{W}) \right]  + \bm{1}\left(m \neq m^\prime \right)\Phi_{t, t^\prime}^{\co,\simul} \right) f(t) f(t^\prime) dt dt^\prime \, ,
 \end{aligned}
 \end{equation}
    where for time $t \in \mathcal{I}_{m}$ and time $t^\prime \in \mathcal{I}_{m^\prime}$, $\Phi_{t, t^\prime}^{\co,\simul}$ is equal to
        \begin{equation}\label{eqn:phi-co-simul}
            \begin{aligned}
        \Phi_{t, t^\prime}^{\co,\simul} =& \frac{1}{4} \bigg(\+E_{\bm{W}^{(-m,-m^\prime)} }\bigg[ \delta^\co_{t}(\bm{W}^{(-m,-m^\prime)},  \bm{W}^{(m,m^\prime)} = (1,1) ) \cdot  \delta^\simul_{t^\prime}(\bm{W}^{(-m,-m^\prime)},  \bm{W}^{(m,m^\prime)} = (1,1) ) \bigg]  \\
        & - \+E_{\bm{W}^{(-m,-m^\prime)} }\bigg[  \delta^\co_{t}(\bm{W}^{(-m,-m^\prime)},  \bm{W}^{(m,m^\prime)} = (1,0) ) \cdot \delta^\simul_{t^\prime}(\bm{W}^{(-m,-m^\prime)},  \bm{W}^{(m,m^\prime)} = (1,0) ) \bigg] \\
        & - \+E_{\bm{W}^{(-m,-m^\prime)} }\bigg[  \delta^\co_{t}(\bm{W}^{(-m,-m^\prime)},  \bm{W}^{(m,m^\prime)} = (0,1) ) \cdot \delta^\simul_{t^\prime}(\bm{W}^{(-m,-m^\prime)},  \bm{W}^{(m,m^\prime)} = (0,1) ) \bigg] \\
        & + \+E_{\bm{W}^{(-m,-m^\prime)} }\bigg[  \delta^\co_{t}(\bm{W}^{(-m,-m^\prime)},  \bm{W}^{(m,m^\prime)} = (0,0) ) \cdot \delta^\simul_{t^\prime}(\bm{W}^{(-m,-m^\prime)},  \bm{W}^{(m,m^\prime)} = (0,0) ) \bigg] \bigg) \, .
    \end{aligned}
        \end{equation}
    
    $S^{(m,m^\prime)}_{3}$ measures the expected value of the product of carryover effect at time $t$ and effects of simultaneous interventions at time $t^\prime$.

 \subsection{Additional Examples}\label{subsec:additional-examples}

 Below, we show an example of the value of $C^{(m)}$ when the covariance decays linearly in the distance between time $t_i$ and time $t_j$. 

    \begin{example}
    \label{example:covariance-value}
        Suppose the covariance $\+E_{\varepsilon}\left[\varepsilon^{(i)} \varepsilon^{(j)} \mid  t_i, t_j \right] $ decays linearly in $|t_i - t_j|$ for all time $t_j \in [t_i - h, t_i + h]$, and is zero outside this interval (i.e., $\+E_{\varepsilon}\left[\varepsilon^{(i)} \varepsilon^{(j)} \mid  t_i, t_j \right] = \sigma^2 (h - |t_i - t_j|)/h$). Suppose the event density $f(t)$ is uniform in time $t$. If $h < |\mathcal{I}_{m}|$, then $C^{(m)} = \sigma^2 \big(|\mathcal{I}_{m}|^2 - |\mathcal{I}_{m}| h + 2h^2/3 \big)/T^2 $; otherwise, $C^{(m)} =  \sigma^2 \big(|\mathcal{I}_{m}|^2 - |\mathcal{I}_{m}|^3/(3h) \big) / T^2$. 
    \end{example}

        \begin{example}
    \label{example:carryover-value}
        Suppose event density $f(t)$ is uniform in time $t$ and the fixed-duration design is used. Furthermore, suppose the carryover effect $ \delta^\co_{t}$ is constant in time $t$ and carryover intensity is constant for time $t^\prime \in [t-h, t]$ for any time $t$ and for $h < T/M$. Then $I^{(m)}=\delta^\co(1/M - h/(2T))$.
    \end{example}

Next we show an example of reducing the bias from simultaneous interventions by staggering the switching times. 
    
\begin{example}[Staggering switching times for simultaneous interventions]\label{example:simul-misalign}
Suppose there is one primary and one simultaneous intervention. Consider a discrete-time setting with $T$ periods. Let $\bm{W} =(W_1, \cdots, W_T)$ be the treatment design of the primary intervention and $\bm{W}^\s = (W^\s_1, \cdots, W^\s_T)$ be the treatment design of the simultaneous intervention. Suppose a simple structure of the treatment effects, i.e., for any $t$,
\[Y_t (\bm{W}_t, \bm{W}_t^\s) - Y_t(\bm{0}_t, \bm{0}_t) = \delta W_t + \delta^\s W^\s_t + \delta^{\mathrm{compd}} W_t W^\s_t \]
implying that there is instantaneous effect only and the effect is constant over time. Moreover, 
\begin{itemize}
    \item When $W_t = 1$ and $W_t^\s = 1$, the treatment effect is $\delta+\delta^\s+\delta^{\mathrm{compd}}$
    \item When $W_t = 1$ and $W_t^\s = 0$, the treatment effect is $\delta$
    \item When $W_t = 0$ and $W_t^\s = 1$, the treatment effect is $\delta^\s$
\end{itemize}

Below we consider two fixed-duration switchback designs with an interval length of two periods for both the primary and simultaneous interventions
\begin{itemize}
    \item Design $1$ (same switching times): designs of both main and simultaneous interventions randomize at time $t = 1, 3, 5, \cdots$.
    \item Design $2$ (staggering switching times): design of the primary intervention randomizes at time $t = 1, 3, 5, \cdots$; design of the simultaneous intervention randomizes at time $t = 2, 4, 6, \cdots$.
\end{itemize}

Next we calculate the bias of the two designs. For the first design, the bias of the HT estimator $\hat{\delta}$ is 
\begin{align*}
    \+E_{W,\varepsilon,t}[\hat{\delta} - \delta]  =& \delta+\delta^\s+\delta^{\mathrm{compd}} - \delta =  \delta^\s+\delta^{\mathrm{compd}}
\end{align*}

For the second design, the bias of the HT estimator $\hat{\delta}$ is 
\begin{align*}
    \+E_{W,\varepsilon,t}[\hat{\delta} - \delta]  =& \left[\left(\frac{1}{2} \cdot \left(\delta+\delta^\s+\delta^{\mathrm{compd}} \right) + \frac{1}{2} \cdot \delta \right)- \left(\frac{1}{2} \cdot \delta^\s + \frac{1}{2} \cdot 0 \right)\right] - \delta = \frac{1}{2} \delta^{\mathrm{compd}}
\end{align*}

Then the sufficient conditions for the bias of the first design to be larger than the bias of the second design is either of following two conditions holds
\begin{itemize}
    \item $\delta^\s \delta^{\mathrm{compd}} \geq 0$, that is, $\delta^\s$ and $\delta^{\mathrm{compd}}$ have the same signs or one of $\delta^\s$ and $\delta^{\mathrm{compd}}$ is zero
    \item $|\delta^\s| \geq 1.5|\delta^{\mathrm{compd}}|$, that is, the scale of $\delta^\s$ is at least $1.5$ times of the scale of $\delta^{\mathrm{compd}}$.
\end{itemize}
\end{example}

%% file: appendix-proof.tex
    \clearpage

    \section{Proof of Main Results}\label{appendix:proof}

    In this section, we first prove Theorem \ref{theorem:switchback-bias} in Section \ref{subsec:proof-first-theorem}, then prove Theorem \ref{theorem:bias-variance-switchback} in Section \ref{subsec:proof-second-theorem}, and finally prove Proposition \ref{prop:additive-simul-effect} in Section \ref{subsec:proof-proposition}. All the proofs allow for more than one simultaneous experiment, and use the notations defined in Appendices \ref{subsec:multi-simul-exp} and \ref{subsec:notations}.

\subsection{Proof of Proposition \ref{prop:additive-simul-effect}}\label{subsec:proof-proposition}

\begin{proof}{Proof of Proposition \ref{prop:additive-simul-effect}}

The bias $\bias(\mathcal{E}_\mathrm{simul})$ is zero following Example \ref{example:additive-simul} in Section \ref{subsec:block-stats}. In what follows, we first show how the expression of $\mathrm{E}(\mathcal{E}_\mathrm{simul}^2)$ is simplified under the additive condition, and then show how the expression of $\+E[(\mathcal{E}_\mathrm{inst}+\mathcal{E}_\mathrm{carryover})\cdot \mathcal{E}_\mathrm{simul}] $ is simplified under the additive condition.

\subsubsection{Simplification of $\mathbb{E}(\mathcal{E}_\mathrm{simul}^2)$ under Condition \ref{cond:additive}}

    Recall that 
    \begin{align*}
        \+E[\mathcal{E}_\mathrm{simul}^2] = \sum_{m=1}^M \sum_{m^\prime = 1}^M S^{(m,m^\prime)}_{\mathrm{var}}  \, ,
    \end{align*}
    in the lemma below, we show that under the additive Condition \ref{cond:additive}, $S^{(m,m^\prime)}_{\mathrm{var}} = 0$ for $m^\prime \neq m$ and we provide the expression of $S^{(m,m^\prime)}_{\mathrm{var}}$ for $m^\prime = m$. 

    \begin{lemma}\label{lemma:S-var}
        Under the assumptions in Theorem \ref{theorem:bias-variance-switchback} and Condition \ref{cond:additive}, 
        $$\mathbb{E}[\mathcal{E}_\mathrm{simul}^2] = \sum_{m = 1}^M S^{(m,m)}_{\mathrm{var}} \, ,$$
        where 
            \begin{align*}
        S^{(m,m)}_{\mathrm{var}}
         =&  \bigg( 
\int_{t\in \mathcal{I}_{m}} \bigg[\sum_{\ell=1}^K 
 \delta^{\s.\gate}_{\ell,t} \bigg]  f(t) dt \bigg)^2   \\ & +  \sum_{m^\prime =1}^M \sum_{\ell=1}^K \bigg( 
\int_{t\in \mathcal{I}_{m} \cap \mathcal{I}^\s_{\ell m^\prime} }  \delta^{\s.\inst}_{\ell,t}  f(t) dt + \int_{t^\prime\in \mathcal{I}_{m}, u \in \mathcal{I}^\s_{\ell m^\prime} }  \delta^{\s.\co}_{\ell, t^\prime} d^{\s.\co}_{\ell,t^\prime}(u)  f(u)  f(t^\prime) dt^\prime d u\bigg)^2 \, ,
    \end{align*}
    and $S^{(m,m^\prime)}_{\mathrm{var}} = 0$ for $m^\prime \neq m$.

    \end{lemma}

    \begin{proof}{Proof of Lemma \ref{lemma:S-var}}
        When treatment effects of simultaneous interventions are additive, using $ \P(W_\ell^{\s(m)} = 1) = 1/2$ for all $\ell$ and $m$, the term $\delta^{\simul,2}_{t, t^\prime}(\bm{W})$ in Equation \eqref{eqn:S-m-mprime-var} is 
        \begin{align*}
         \delta^{\simul,2}_{t, t^\prime}(\bm{W}) =& \+E_{\bm{W}^\s_1, \cdots, \bm{W}^\s_K } \left[ \left(\sum_{\ell = 1}^K \left[W^\s_{\ell, t}  \delta^{\s.\inst}_{\ell,t}  + \delta^{\s.\co}_{\ell,t}  \cdot \sum_{k = 1}^M W_{\ell}^{\s(k)} \int_{u \in \mathcal{I}^\s_{\ell k}} d^{\s.\co}_{\ell,t}(u)  f(u) d u \right]  \right) \right.\\
        & \left. \times \left(\sum_{\ell =1}^K \left[W^\s_{\ell, t^\prime}  \delta^{\s.\inst}_{\ell,t^\prime}  + \delta^{\s.\co}_{\ell,t^\prime}  \cdot \sum_{k = 1}^M W_{\ell}^{\s(k)} \int_{u \in \mathcal{I}^\s_{\ell k}} d^{\s.\co}_{\ell,t^\prime}(u)  f(u) d u \right]  \right)  \mid \bm{W}, t, t^\prime \right] \\
        =& \frac{1}{4} \left(\sum_{\ell = 1}^K \delta^{\s.\gate}_{\ell,t}   \right) \left(\sum_{\ell = 1}^K \delta^{\s.\gate}_{\ell,t^\prime}  \right) \\
        & + \frac{1}{4} \sum_{\ell = 1}^K \delta^{\s.\co}_{\ell,t} \delta^{\s.\co}_{\ell,t^\prime} \sum_{k = 1}^M \left(\int_{u \in \mathcal{I}^\s_{\ell k}} d^{\s.\co}_{\ell,t}(u)  f(u) d u \right) \left(\int_{u \in \mathcal{I}^\s_{\ell k}} d^{\s.\co}_{\ell,t^\prime}(u)  f(u) d u \right) \\
        & + \frac{1}{4} \sum_{\ell = 1}^K \delta^{\s.\inst}_{\ell,t} \delta^{\s.\co}_{\ell,t^\prime} \int_{u \in \mathcal{I}^\s_{\ell m(t)}} d^{\s.\co}_{\ell,t^\prime}(u)  f(u) d u \tag{$\mathcal{I}^\s_{\ell m(t)}$ denotes the interval of simul. intervention $\ell$ to which  $t$ belongs} \\
        & + \frac{1}{4} \sum_{\ell = 1}^K \delta^{\s.\inst}_{\ell,t^\prime} \delta^{\s.\co}_{\ell,t} \int_{u \in \mathcal{I}^\s_{\ell m(t^\prime)}} d^{\s.\co}_{\ell,t}(u)  f(u) d u \\
        & + \frac{1}{4} \sum_{\ell = 1}^K \bm{1}\left(\text{$t$ and $t^\prime$ in the same interval of simul. intervention $\ell$} \right)\delta^{\s.\inst}_{\ell,t} \delta^{\s.\inst}_{\ell,t^\prime}  \, .
    \end{align*}
    We can see that $\delta^{\simul,2}_{t, t^\prime}(\bm{W}) $ does not depend on $\bm{W}$ for any time $t$ and time $t^\prime$. 

    When time $t$ and time $t^\prime$ are in different intervals of the primary intervention ($m \neq m^\prime$), the term $\Phi^{2\dagger}_{t,t^\prime}$ in Equation \eqref{eqn:S-m-mprime-var} is 0 as $ \delta^{\simul,2}_{t, t^\prime}(\bm{W}^{(-m,-m^\prime)},  \bm{W}^{(m,m^\prime)} = (w^m, w^{m^\prime}) )$ is the same for $(w^m, w^{m^\prime}) = \{(1,1), (1,0), (0,1), (0,0)\}$. Therefore, when $m \neq m^\prime$, 
    \[S^{(m,m^\prime)}_{\mathrm{var}} = 0 \,   \]
    and 
    \begin{align*}
        \mathrm{E}(\mathcal{E}_\mathrm{simul}^2) =&  \sum_{m = 1}^M \sum_{m^\prime = 1}^M S^{(m,m^\prime)}_{\mathrm{var}} = \sum_{m = 1}^M S^{(m,m)}_{\mathrm{var}} \, . 
    \end{align*}

    Next is to provide the expression of $S^{(m,m)}_{\mathrm{var}}$. We have 
    \begin{align*}
         S^{(m,m)}_{\mathrm{var}} =& 4  \int_{t, t^\prime \in \mathcal{I}_{m}}  \+E_{W}[\delta^{\simul,2}_{ t, t^\prime}(\bm{W}) \mid t, t^\prime] f(t) f(t^\prime) dt dt^\prime \\
         =& \int_{t, t^\prime \in \mathcal{I}_{m}} \left(\sum_{\ell = 1}^K \delta^{\s.\gate}_{\ell,t}   \right) \left(\sum_{\ell = 1}^K \delta^{\s.\gate}_{\ell,t^\prime}  \right)  f(t) f(t^\prime) dt dt^\prime  \\
         & + \sum_{\ell = 1}^K \sum_{k=1}^M \int_{t, t^\prime \in \mathcal{I}_{m}} \delta^{\s.\co}_{\ell,t} \delta^{\s.\co}_{\ell,t^\prime}  \left(\int_{u \in \mathcal{I}^\s_{\ell k}} d^{\s.\co}_{\ell,t}(u)  f(u) d u \right) \left(\int_{u \in \mathcal{I}^\s_{\ell k}} d^{\s.\co}_{\ell,t^\prime}(u)  f(u) d u \right) f(t) f(t^\prime) dt dt^\prime  \\
         & + 2 \sum_{\ell = 1}^K \sum_{m^\prime = 1}^M \int_{t \in \mathcal{I}_{m} \cap \mathcal{I}^\s_{\ell m^\prime}, t^\prime \in \mathcal{I}_{m}} \delta^{\s.\inst}_{\ell,t} \delta^{\s.\co}_{\ell,t^\prime} \left(\int_{u \in \mathcal{I}^\s_{\ell m^\prime}} d^{\s.\co}_{\ell,t^\prime}(u)  f(u) d u \right) f(t) f(t^\prime) dt dt^\prime \tag{combining the third and fourth terms in the decomposition of $\delta^{\simul,2}_{t, t^\prime}(\bm{W})$}  \\
         & + \sum_{\ell = 1}^K \sum_{m^\prime =1}^M  \int_{t, t^\prime \in \mathcal{I}_{m}  \cap \mathcal{I}^\s_{\ell m^\prime} 
 } \delta^{\s.\inst}_{\ell,t}  \delta^{\s.\inst}_{\ell,t^\prime} f(t) f(t^\prime) dt dt^\prime \, .
    \end{align*}
    We can further simplify $S^{(m,m)}_{\mathrm{var}} $ to 
    \begin{align*}
        S^{(m,m)}_{\mathrm{var}} 
         =&   \bigg(\int_{t\in \mathcal{I}_{m}} \bigg[\sum_{\ell = 1}^K 
 \delta^{\s.\gate}_{\ell,t} \bigg]  f(t) dt  \bigg)^2 
 + \sum_{\ell = 1}^K \sum_{m^\prime =1}^M \bigg(\int_{t\in \mathcal{I}_{m}, u \in \mathcal{I}^\s_{\ell m^\prime} }  \delta^{\s.\co}_{\ell, t} d^{\s.\co}_{\ell,t}(u)  f(u)  f(t) dt d u   \bigg)^2 \\
 &
          + 2 \sum_{\ell = 1}^K \sum_{m^\prime =1}^M \bigg(\int_{t\in \mathcal{I}_{m} \cap \mathcal{I}^\s_{\ell m^\prime} }  \delta^{\s.\inst}_{\ell, t}  f(t) dt  \bigg)  \bigg(\int_{t^\prime\in \mathcal{I}_{m}, u \in \mathcal{I}^\s_{\ell m^\prime} }  \delta^{\s.\co}_{\ell, t^\prime} d^{\s.\co}_{\ell, t^\prime}(u)  f(u)  f(t^\prime) dt^\prime d u   \bigg) \\
          & +  \sum_{\ell = 1}^K \sum_{m^\prime =1}^M \bigg(\int_{t\in \mathcal{I}_{m} \cap \mathcal{I}^\s_{\ell m^\prime} }  \delta^{\s.\inst}_{\ell,t}  f(t) dt \bigg)^2 \\
          =&  \bigg(\int_{t\in \mathcal{I}_{m}} \bigg[\sum_{\ell = 1}^K 
 \delta^{\s.\gate}_{\ell,t} \bigg]  f(t) dt  \bigg)^2 \\
 & + \sum_{\ell = 1}^K \sum_{m^\prime =1}^M \bigg(\int_{t\in \mathcal{I}_{m} \cap \mathcal{I}^\s_{\ell m^\prime} }  \delta^{\s.\inst}_{\ell,t}  f(t) dt + \int_{t^\prime\in \mathcal{I}_{m}, u \in \mathcal{I}^\s_{\ell m^\prime} }  \delta^{\s.\co}_{\ell, t^\prime} d^{\s.\co}_{\ell,t^\prime}(u)  f(u)  f(t^\prime) dt^\prime d u \bigg)^2 \, ,
    \end{align*}
    which is the same as the expression of $S^{(m,m)}_{\mathrm{var}}$ provided in the statement of Lemma \ref{lemma:S-var} and concludes the proof of Lemma \ref{lemma:S-var}. \halmos
    \end{proof}

    \subsubsection{Simplification of $\+E[(\mathcal{E}_\mathrm{inst}+\mathcal{E}_\mathrm{carryover})\cdot \mathcal{E}_\mathrm{simul}] $ under Condition \ref{cond:additive}} Recall that 
    \[\+E[(\mathcal{E}_\mathrm{inst}+\mathcal{E}_\mathrm{carryover}) \cdot \mathcal{E}_\mathrm{simul}]  = \sum_{m = 1}^M \sum_{m^\prime = 1}^M S^{(m,m^\prime)}_{\mathrm{cov}} \, ,\]
    where 
    \[S^{(m,m^\prime)}_{\mathrm{cov}} = (\delta^\gate \mu^{(m^\prime)} + 2 \mu^{(m^\prime)}_{Y^\mathrm{ctrl}}) \cdot S^{(m,m^\prime)}_{1}  + \left( \Xi^{\inst,(m)}_\dem - \delta^\co \mu^{(m)} \right) S^{(m,m^\prime)}_{2} + S^{(m,m^\prime)}_{3} \, .\]
    
    In Lemmas \ref{lemma:S-cov-1}, \ref{lemma:S-cov-2}, and \ref{lemma:S-cov-3} below, we show the expression of $S^{(m,m^\prime)}_{1}$, $S^{(m,m^\prime)}_{2}$, and $S^{(m,m^\prime)}_{3}$ under Condition \ref{cond:additive}. 

    \begin{lemma}\label{lemma:S-cov-1}
        Under the assumptions in Theorem \ref{theorem:bias-variance-switchback} and Condition \ref{cond:additive}, 
        \begin{align*}
        S^{(m,m)}_{1} =& \sum_{\ell = 1}^K \left(\int_{t \in \mathcal{I}_{m}} \delta^{\s.\gate}_{\ell,t} f(t) dt \right) \, 
    \end{align*}
    and $S^{(m,m^\prime)}_{1} = 0$ for $m \neq m^\prime$. 
    \end{lemma}

    \begin{proof}{Proof of Lemma \ref{lemma:S-cov-1}}
        When the effects of primary and simultaneous interventions are additive, the term $\delta^\simul_{t}(\bm{W})$ in Equation \eqref{eqn:S-m-mprime-1} is equal to
\begin{align*}
    & \delta^\simul_{t}(\bm{W}) = \+E_{\bm{W}^\s_1, \cdots, \bm{W}^\s_K }\left[Y_{t}(\bm{W}, \bm{W}^\s_1, \cdots, \bm{W}^\s_K) - Y_{t}(\bm{W}, \bm{0}, \cdots, \bm{0}) \mid \bm{W}, t \right]  \\ =& 
 \+E_{\bm{W}^\s_1, \cdots, \bm{W}^\s_K } \left[ \sum_{\ell = 1}^K \left(W^\s_{\ell, t}  \delta^{\s.\inst}_{\ell,t}  + \delta^{\s.\co}_{\ell,t}  \cdot \sum_{k = 1}^M W_{\ell}^{\s(k)} \int_{u \in \mathcal{I}^\s_{\ell k}} d^{\s.\co}_{\ell,t}(u)  f(u) d u \right)   \mid \bm{W}, t, t^\prime \right] \\
 =& \frac{1}{2} \sum_{\ell = 1}^K \left[  \delta^{\s.\inst}_{\ell,t}  + \delta^{\s.\co}_{\ell,t}  \cdot \sum_{k = 1}^M  \int_{u \in \mathcal{I}^\s_{\ell k}} d^{\s.\co}_{\ell,t}(u)  f(u) d u \right] = \frac{1}{2} \sum_{\ell = 1}^K \delta^{\s.\gate}_{\ell,t} \, .
\end{align*}
We can see that $\delta^\simul_{t}(\bm{W}) $ does not depend on the value of $\bm{W}$. Then for $m \neq m^\prime$, the term $\Phi_{t}^{\simul,(-m^\prime)}$ in Equation \eqref{eqn:S-m-mprime-1} is $0$, because $\delta^\simul_{ t}(\bm{W}^{(-m,-m^\prime)},  \bm{W}^{(m,m^\prime)} = (w^m,w^{m^\prime} ) ) $ is the same for $(w^m, w^{m^\prime}) = \{(1,1), (1,0), (0,1), (0,0)\}$. Therefore, for $m \neq m^\prime$, 
    \[ S^{(m,m^\prime)}_{1} = 0 \qquad m \neq m^\prime \, . \]
    For $m = m^\prime$, we have
    \begin{align*}
        S^{(m,m)}_{1} =&2  \int_{t \in \mathcal{I}_{m}} \left(\frac{1}{2} \sum_{\ell = 1}^K \delta^{\s.\gate}_{\ell,t} \right)  f(t) dt = \sum_{\ell = 1}^K \left(\int_{t \in \mathcal{I}_{m}} \delta^{\s.\gate}_{\ell,t} f(t) dt \right) \, ,
    \end{align*} which concludes the proof of Lemma \ref{lemma:S-cov-1}. \halmos
    \end{proof}

    \begin{lemma}\label{lemma:S-cov-2}
        Under the assumptions in Theorem \ref{theorem:bias-variance-switchback} and Condition \ref{cond:additive}, 
        \begin{align*}
        S^{(m,m)}_{2} =& \sum_{\ell = 1}^K  \left(\int_{t^\prime \in \mathcal{I}_{m}} \delta^{\s.\gate}_{\ell,t^\prime} f(t^\prime) dt^\prime \right) \, ,
    \end{align*}
    and $S^{(m,m^\prime)}_{2} = 0 $ for $m \neq m^\prime$. 
    \end{lemma}

    \begin{proof}{Proof of Lemma \ref{lemma:S-cov-2}}
        If the effects of main and simultaneous interventions are additive, the term $\+E_{\bm{W}^{(-m)} }\big[ \delta^\simul_{ t^\prime}(\bm{W}^{(-m)},  W^{(m)} = 1 ) \big]$ in Equation \eqref{eqn:S-m-mprime-2} is equal to 
        \begin{align*}
        \+E_{\bm{W}^{(-m)} }\big[ \delta^\simul_{ t^\prime}(\bm{W}^{(-m)},  W^{(m)} = 1 ) \big] = \frac{1}{2} \sum_{\ell = 1}^K \delta^{\s.\gate}_{\ell,t^\prime}
    \end{align*}
    following the same argument as the proof of Lemma \ref{lemma:S-cov-1}. 
    Then the term $\Phi_{t^\prime}^{\simul,(-m^\prime)\dagger}$ in Equation \eqref{eqn:S-m-mprime-2} is $0$ because $\delta^\simul_{ t}(\bm{W}^{(-m,-m^\prime)},  \bm{W}^{(m,m^\prime)} = (w^m,w^{m^\prime} ) ) $ is the same for $(w^m, w^{m^\prime}) = \{(1,1), (1,0)\}$. Therefore,  for $m \neq m^\prime$
    \begin{align*}
        S^{(m,m^\prime)}_{2} = 0 
    \end{align*}
    and 
    \begin{align*}
        S^{(m,m)}_{2} =& 2  \int_{t^\prime \in \mathcal{I}_{m}}\left(\frac{1}{2} \sum_{\ell = 1}^K \delta^{\s.\gate}_{\ell,t^\prime} \right) f(t^\prime) dt^\prime  =\sum_{\ell = 1}^K  \left(\int_{t^\prime \in \mathcal{I}_{m}} \delta^{\s.\gate}_{\ell,t^\prime} f(t^\prime) dt^\prime \right)  \, ,
    \end{align*}
    which is equal to $S^{(m,m)}_{1}$ and concludes the proof of Lemma \ref{lemma:S-cov-2}. \halmos
    \end{proof}

    \begin{lemma}\label{lemma:S-cov-3}
        Under the assumptions in Theorem \ref{theorem:bias-variance-switchback} and Condition \ref{cond:additive}, 
        \begin{align*}
     S^{(m,m)}_{3} 
     =& \left(\int_{t \in \mathcal{I}_{m}} \delta^\co_{t} f(t) dt \right) \left(\sum_{\ell = 1}^K  
 \int_{t^\prime \in \mathcal{I}_{m}} \delta^{\s.\gate}_{\ell,t^\prime} f(t^\prime) dt^\prime \right) \, 
 \end{align*}
 and $S^{(m,m^\prime)}_{3} = 0$ for $m \neq m^\prime$. 
    \end{lemma}

    \begin{proof}{Proof of Lemma \ref{lemma:S-cov-3}}
        If the effects of main and simultaneous interventions are additive, from Lemma \ref{lemma:S-cov-1}, we have 
        \[\delta^\simul_{ t^\prime}(\bm{W}) = \frac{1}{2} \sum_{\ell = 1}^K \delta^{\s.\gate}_{\ell,t^\prime} \, , \]
        which does not depend on $\bm{W}$. 
    We then have 
    \begin{align*}
        \+E_W\left[\delta^\co_{t}(\bm{W})  \delta^\simul_{ t^\prime}(\bm{W}) \right] = \+E_W\left[\delta^\co_{t}(\bm{W})   \right] \cdot \frac{1}{2}  \sum_{\ell = 1}^K  
 \delta^{\s.\gate}_{\ell,t^\prime} = \frac{1}{4} \delta^\co_{t}  \sum_{\ell = 1}^K  
 \delta^{\s.\gate}_{\ell,t^\prime}
    \end{align*}
    following that 
    \begin{align*}
        \+E_W\left[\delta^\co_{t}(\bm{W}) \right] =& \+E_W\left[\delta^\co_{t}  \cdot \sum_{k = 1}^M W^{(k)} \int_{u \in \mathcal{I}_{k}} d^\co_{t}(u)  f(u) d u \right] = \frac{1}{2} \delta^\co_{t}  \, .
    \end{align*}

    Therefore, if time $t \in \mathcal{I}_{m}$ and time $t^\prime \in \mathcal{I}_{m^\prime}$ with $m \neq m^\prime$, we have
    $\Phi_{t, t^\prime}^{\co,\simul} = 0$ because $\delta^\co_{t}(\bm{W}^{(-m,-m^\prime)},  \bm{W}^{(m,m^\prime)} = (w^m,w^{m^\prime}) ) \cdot \delta^\simul_{t^\prime}(\bm{W}^{(-m,-m^\prime)},  \bm{W}^{(m,m^\prime)} = (w^m,w^{m^\prime}) )$ is the same for $(w^m, w^{m^\prime}) = \{(1,1), (1,0), (0,1), (0,0)\}$. Therefore, for $m \neq m^\prime$, 
    \begin{align*}
        S^{(m,m^\prime)}_{3} = 0  \, .
    \end{align*}
    For $m = m^\prime$, we have 
    \begin{align*}
     S^{(m,m)}_{3} =&  \int_{t, t^\prime \in \mathcal{I}_{m}} \delta^\co_{t} \left(\sum_{\ell = 1}^K \delta^{\s.\gate}_{\ell,t^\prime} \right)   f(t) f(t^\prime) dt dt^\prime \\
     =& \left(\int_{t \in \mathcal{I}_{m}} \delta^\co_{t} f(t) dt \right) \left(\sum_{\ell = 1}^K  
 \int_{t^\prime \in \mathcal{I}_{m}} \delta^{\s.\gate}_{\ell,t^\prime} f(t^\prime) dt^\prime \right) \, ,
 \end{align*} 
 which concludes the proof of Lemma \ref{lemma:S-cov-3}. \halmos
    \end{proof}

     From Lemmas \ref{lemma:S-cov-1}, \ref{lemma:S-cov-2}, and \ref{lemma:S-cov-3}, we have for $m \neq m^\prime$, 
     \begin{align*}
        S^{(m,m^\prime)}_{\mathrm{cov}} = (\delta^\gate \mu^{(m^\prime)} + 2 \mu^{(m^\prime)}_{Y^\mathrm{ctrl}}) \cdot S^{(m,m^\prime)}_{1}  + \left( \Xi^{\inst,(m)}_\dem - \delta^\co\mu^{(m)} \right) S^{(m,m^\prime)}_{2} + S^{(m,m^\prime)}_{3} = 0 \, .
    \end{align*}
    For $m = m^\prime$, we have 
    \begin{align*}
        S^{(m,m)}_{\mathrm{cov}} =&  (\delta^\gate \mu^{(m)} + 2 \mu^{(m)}_{Y^\mathrm{ctrl}}) \cdot S^{(m,m)}_{1}  + \left( \Xi^{\inst,(m)}_\dem - \delta^\co\mu^{(m)} \right) S^{(m,m)}_{2} + S^{(m,m)}_{3} \\
        =& \left(\delta^\gate \mu^{(m)} + 2 \mu^{(m)}_{Y^\mathrm{ctrl}}\right)\sum_{\ell = 1}^K \left(\int_{t \in \mathcal{I}_{m}} \delta^{\s.\gate}_{\ell,t} f(t) dt \right)  \\
     & + \left( \Xi^{\inst,(m)}_\dem - \delta^\co\mu^{(m)} \right) \sum_{\ell = 1}^K  \left(\int_{t^\prime \in \mathcal{I}_{m}} \delta^{\s.\gate}_{\ell,t^\prime} f(t^\prime) dt^\prime \right) \\
     & + \left(\int_{t \in \mathcal{I}_{m}} \delta^\co_{t} f(t) dt \right) \left(\sum_{\ell = 1}^K  
 \int_{t^\prime \in \mathcal{I}_{m}} \delta^{\s.\gate}_{\ell,t^\prime} f(t^\prime) dt^\prime \right) \\
 =& \left(\Xi^{(m)} + 2 \mu^{(m)}_{Y^\mathrm{ctrl}}\right)\bigg( 
\int_{t\in \mathcal{I}_{m}} \bigg[\sum_{\ell = 1}^K 
 \delta^{\s.\gate}_{\ell,t} \bigg]  f(t) dt \bigg)
    \end{align*}
    following that
    \begin{align*}
        & \delta^\gate \mu^{(m)} + \Xi^{\inst,(m)}_\dem -  \delta^\co\mu^{(m)} + \Xi^{\co,(m)}_\dem \\
        =& \delta^\inst \mu^{(m)} + \Xi^{\inst,(m)}_\dem  + \Xi^{\co,(m)}_\dem = \Xi^{\inst,(m)} + \Xi^{\co,(m)}_\dem =  \Xi^{(m)}  \, .
    \end{align*}

    Then we have
     \begin{align*}
     \+E_{\ell}[(\mathcal{E}_\mathrm{inst}+\mathcal{E}_\mathrm{carryover})\cdot \mathcal{E}_\mathrm{simul}] =& \sum_{m = 1}^M \sum_{m^\prime = 1}^M S^{(m,m^\prime)}_{\mathrm{cov}}  \\
     =& \sum_{m=1}^M \left(\Xi^{(m)} + 2 \mu^{(m)}_{Y^\mathrm{ctrl}}\right)\bigg( 
\int_{t\in \mathcal{I}_{m}} \bigg[\sum_{\ell = 1}^K 
 \delta^{\s.\gate}_{\ell,t} \bigg]  f(t) dt \bigg) \, ,
 \end{align*} which concludes the proof of Proposition \ref{prop:additive-simul-effect}. \halmos
 \end{proof}

\subsection{Proof of Theorem \ref{theorem:switchback-bias}}\label{subsec:proof-first-theorem}

The estimation error of $\hat{\delta}^\gate$ can be decomposed as 
 \begin{equation}
     \begin{aligned}
     \hat{\delta}^\gate  - \delta^\gate 
     =& \underbrace{\delta^\gate \left(\frac{1}{n} \sum_{i = 1}^n \frac{W_{t_i}}{\pi} - 1 \right)}_{\substack{\text{instantaneous effects} \\ \text{denoted by $\mathcal{E}_{\mathrm{inst}}$}}} + \underbrace{\frac{1}{n} \sum_{i = 1}^n \alpha_{t_i} \left(Y_{t_i}(\bm{W}, \bm{0}, \cdots, \bm{0})  - W_{t_i} \delta^\gate\right)}_{\substack{\text{carryover effects} \\ \text{denoted by $\mathcal{E}_{\mathrm{carryover}}$}}}  \\
     & + \underbrace{\frac{1}{n} \sum_{i = 1}^n \alpha_{t_i} \left(Y_{t_i}(\bm{W}, \bm{W}^\s_1, \cdots, \bm{W}^\s_K) - Y_{t_i}(\bm{W}, \bm{0}, \cdots, \bm{0})  \right)}_{\text{effects from other interventions, denoted by $\mathcal{E}_{\mathrm{simul}}$}} \\
     & + \underbrace{\frac{1}{n} \sum_{i = 1}^n \alpha_{t_i} \varepsilon^{(i)}}_{\substack{ \text{measurement errors} \\ \text{denoted by $\mathcal{E}_{\mathrm{meas}}$}}} \, ,
 \end{aligned}\label{eqn:decomposition} 
 \end{equation}
 where 
 \[ \alpha_{t_i} = \frac{W_{t_i} - \pi}{\pi(1 - \pi)} \, . \]

 We show the expected value of $\mathcal{E}_{\mathrm{meas}}$ in Lemma \ref{lemma:mean-measurement-error}, the expected value of $\mathcal{E}_{\mathrm{inst}}$ and $\mathcal{E}_{\mathrm{carryover}}$ in Lemma \ref{lemma:carryover-bias}, and the expected value of $\mathcal{E}_{\mathrm{simul}}$ in Lemma \ref{lemma:simul-bias}. If the expected value of a term is nonzero, then this term results in an estimation bias of $\hat{\delta}^\gate$. 

 \begin{lemma}[Expected value of $\mathcal{E}_{\mathrm{meas}}$]\label{lemma:mean-measurement-error}
     Under the assumptions in Theorem \ref{theorem:switchback-bias}, the mean of measurement errors is 
     \begin{align*}
           \+E_{W,\varepsilon,t} \left[ \mathcal{E}_{\mathrm{meas}}\right]=&~ \+E_{W,\varepsilon,t} \left[ \frac{1}{n} \sum_{i = 1}^n \alpha_{t_i} \varepsilon^{(i)}\right] = 0 \,  .
     \end{align*}
 \end{lemma}

 \begin{proof}{Proof of Lemma \ref{lemma:mean-measurement-error}}
     The mean of the measurement errors is
    \begin{align*}
        \+E_{W,\varepsilon,t} \left[ \frac{1}{n} \sum_{i = 1}^n \alpha_{t_i} \varepsilon^{(i)}\right] = \frac{1}{n} \sum_{i = 1}^n 
 \+E_{W,t} \left[ \alpha_{t_i}  \+E_{\varepsilon}\left[\varepsilon^{(i)} \mid \bm{W}, \bm{W}^\s_1, \cdots, \bm{W}^\s_K, t_i\right]\right] = 0
    \end{align*}
    following that $\varepsilon^{(i)}$ has mean zero and is independent of $\bm{W}, \bm{W}^\s_1, \cdots$, and $\bm{W}^\s_K$. 
    \halmos
    
 \end{proof}

 \begin{lemma}[Expected value of $\mathcal{E}_{\mathrm{inst}}$ and $\mathcal{E}_{\mathrm{carryover}}$]\label{lemma:carryover-bias}
    Under the assumptions in Theorem \ref{theorem:switchback-bias}, the expected value of $\mathcal{E}_{\mathrm{inst}}$ and $\mathcal{E}_{\mathrm{carryover}}$ are
    \begin{align*}
    \+E_{W,\varepsilon,t} \left[ \mathcal{E}_{\mathrm{inst}}\right] =&~ \+E_{W,\varepsilon,t} \left[  \delta^\gate \left(\frac{1}{n} \sum_{i = 1}^n \frac{W_{t_i}}{\pi} - 1 \right)\right] = 0 \, \\
        \+E_{W,\varepsilon,t} \left[ \mathcal{E}_{\mathrm{carryover}}\right] =& \+E_{W,\varepsilon,t}\left[\frac{1}{n} \sum_{i = 1}^n \alpha_{t_i} \left(Y_{t_i}(\bm{W}, \bm{0}, \cdots, \bm{0})  - W_{t_i} \delta^\gate\right) \right] = \sum_{m = 1}^M I^{(m)} - \delta^\co  \, .
    \end{align*}
\end{lemma}

\begin{proof}{Proof of Lemma \ref{lemma:carryover-bias}}
   The expected value of $\mathcal{E}_{\mathrm{inst}}$ is equal to 
    \begin{align*}
        \+E_{W,\varepsilon,t} \left[ \mathcal{E}_{\mathrm{inst}}\right] 
 = \+E_{W,\varepsilon,t} \left[ \frac{1}{n} \sum_{i = 1}^n \frac{W_{t_i}}{\pi} - 1\right] = \frac{1}{n} \sum_{i = 1}^n \frac{\+E_{W,\varepsilon,t}[W_{t_i}] }{\pi} - 1 = 0 \, .
    \end{align*}
As events are sampled i.i.d. from distribution $f(t)$, the expected value of $\mathcal{E}_{\mathrm{carryover}}$ is equal to 
    \begin{align*}
         &\+E_{W,\varepsilon,t} \left[ \mathcal{E}_{\mathrm{carryover}}\right] =\+E_{W,\varepsilon,t}\left[\frac{1}{n} \sum_{i = 1}^n \alpha_{t_i} \left(Y_{t_i}(\bm{W}, \bm{0}, \cdots, \bm{0})  - W_{t_i} \delta^\gate\right) \right] \\
         =&\+E_{W,\varepsilon,t}\left[\frac{1}{n} \sum_{i = 1}^n \alpha_{t_i} \left(Y_{t_i}(\bm{W}, \bm{0}, \cdots, \bm{0})  - Y_{t_i}(\bm{0}, \bm{0}, \cdots, \bm{0}) - W_{t_i} \delta^\gate\right) \right] \\
         & + \+E_{W,\varepsilon,t}\left[\frac{1}{n} \sum_{i = 1}^n \alpha_{t_i} Y_{t_i}(\bm{0}, \bm{0}, \cdots, \bm{0})  \right] \\ 
        =& \+E_{W,t} \left[\alpha_{t_i} (\delta^\inst_{t_i} - \delta^\inst) W_{t_i} + \alpha_{t_i}( \delta^\co_{t_i}(\bm{W}) - \delta^\co W_{t_i}) \right] + \+E_{W,t}\left[\alpha_{t_i} Y_{t_i}(\bm{0}, \bm{0}, \cdots, \bm{0})  \right]  \tag{the expected over $\varepsilon$ is dropped as marketplace outcomes do not depend on $\varepsilon^{(i)}$}  \\
        =&  \+E_t\bigg[(\delta^\inst_{t_i} - \delta^\inst) \cdot \underbrace{\+E_W\left[\frac{W_{t_i}}{\pi} \mid t_i \right]}_{=1} \bigg] +  \+E_{W,t} \left[ \alpha_{t_i}( \delta^\co_{t_i}(\bm{W}) - \delta^\co W_{t_i}) \right] \tag{the first term is zero because $\+E_t[\delta^\inst_{t_i}] = \delta^\inst$}  \\
        & + \+E_{W,t}\big[ \underbrace{\+E_W[\alpha_{t_i} \mid t_i] }_{=0}\cdot Y_{t_i}(\bm{0}, \bm{0}, \cdots, \bm{0})  \big]  \\
        =& \+E_{W,t} \left[ \alpha_{t}( \delta^\co_{t}(\bm{W}) - \delta^\co W_{t}) \right] \, . 
    \end{align*}
    By Assumption \ref{ass:interference-structure}, the last line can be further simplified to 
    \begin{align*}
         &\+E_{W,t} \left[ \alpha_{t}( \delta^\co_{t}(\bm{W}) - \delta^\co W_{t}) \right]  = \+E_{W,t}\Bigg[\frac{W_{t} - \pi}{\pi (1 - \pi)} \sum_{k = 1}^M \delta^\co_{t} W^{(k)} \int_{u \in \mathcal{I}_{k}} d^\co_{t}(u) f(u) d u  \Bigg] - \delta^\co \\
         =& \+E_{W,\varepsilon,t}\Bigg[\frac{W_{t} - \pi}{\pi (1 - \pi)} \delta^\co_{t} \sum_{k = 1}^M  W^{(k)} \mathbbm{1}(t \in \mathcal{I}_{k}) \int_{u \in \mathcal{I}_{k}} d^\co_{t}(u) f(u) d u  \Bigg] - \delta^\co \tag{the treatment assignments of any two intervals are independent and $\+E_{W,\varepsilon,t}\left[\frac{W_{t} - \pi}{\pi (1 - \pi)}\right] = 0$}\\
         =& \sum_{m = 1}^M \int_{t \in \mathcal{I}_{m}} \delta^\co_{t} \left[\int_{u \in \mathcal{I}_{m}} d^\co_{t}(u) f(u) d u \right] f(t) dt - \delta^\co \tag{$\+E_{W,\varepsilon,t}\left[\frac{(W_{t} - \pi) W_{t}}{\pi (1 - \pi)}\right] = 1$} \\ 
         =& \sum_{m = 1}^M I^{(m)} - \delta^\co \, . \tag{by definition of $I^{(m)}$}
    \end{align*} 
    We then finish the proof of Lemma \ref{lemma:carryover-bias}. \halmos
\end{proof}

\begin{lemma}[Expected value of $\mathcal{E}_{\mathrm{simul}}$]\label{lemma:simul-bias}
    Under the assumptions in Theorem \ref{theorem:switchback-bias}, the bias from simultaneous interventions is 
    \begin{align*}
         \+E_{W,\varepsilon,t} \left[ \mathcal{E}_{\mathrm{simul}}\right]
        =& \+E_{W,\varepsilon,t} \left[\frac{1}{n} \sum_{i = 1}^n\alpha_{t} \left(Y_{t_i}(\bm{W}, \bm{W}^\s_1, \cdots, \bm{W}^\s_K) - Y_{t_i}(\bm{W}, \bm{0}, \cdots, \bm{0})  \right) \right] \\ =& \sum_{m = 1}^M \int_{t \in \mathcal{I}_{m}} \Phi^\simul_{t} f(t) dt \, ,
    \end{align*}
    where 
    \begin{align*}
        \Phi^\simul_{t} = \+E_{\bm{W}^{(-m)} }\bigg[ \delta^\dagger_{ t}(\bm{W}^{(-m)},  W^{(m)} = 1 ) \bigg]  - \+E_{\bm{W}^{(-m)} }\bigg[ \delta^\dagger_{ t}(\bm{W}^{(-m)},  W^{(m)} = 0 ) 
        \bigg] \, 
    \end{align*}
    is defined in Section \ref{subsec:block-stats}.
\end{lemma}

\begin{proof}{Proof of Lemma \ref{lemma:simul-bias}}

As events are sampled i.i.d. from $f(t)$, we have
    \begin{align*}
        &\+E_{W,\varepsilon,t} \left[\frac{1}{n} \sum_{i = 1}^n\alpha_{t_i} \left(Y_{t_i}(\bm{W}, \bm{W}^\s_1, \cdots, \bm{W}^\s_K) - Y_{t_i}(\bm{W}, \bm{0}, \cdots, \bm{0})  \right) \right] \\
        =& \+E_{W,t} \left[ \alpha_{t} \left(Y_{t}(\bm{W}, \bm{W}^\s_1, \cdots, \bm{W}^\s_K) - Y_{t}(\bm{W}, \bm{0}, \cdots, \bm{0})  \right) \right] \tag{marketplace potential outcomes do not depend on $\varepsilon^{(i)}$}\\
        =& \+E_{W,t}\left[\alpha_{t}  \+E_{\bm{W}^\s_1, \cdots, \bm{W}^\s_K }\left[Y_{t}(\bm{W}, \bm{W}^\s_1, \cdots, \bm{W}^\s_K) - Y_{t}(\bm{W}, \bm{0}, \cdots, \bm{0}) \mid \bm{W}, t \right]\right] \tag{by the law of total expectation} \\
        =& \+E_{W,t}\left[\alpha_{t}  \delta^\simul_{ t}(\bm{W}) \right] \tag{by definition of $ \delta^\simul_{ t}(\bm{W})$} \\ 
        =& \sum_{m = 1}^M \int_{t \in \mathcal{I}_{m}} \Phi^\simul_{t} f(t) dt \, .\tag{first take the expected value over $\bm{W}^{(-m)}$ and then take the expected value over $W^{(m)}$}
    \end{align*}
    We then finish the proof of Lemma \ref{lemma:simul-bias}. \halmos
\end{proof}

\begin{proof}{Proof of Theorem \ref{theorem:switchback-bias}}
Based on the decomposition of the estimation error of $\hat{\delta}^\gate$ in Equation \eqref{eqn:decomposition}, the bias of $\hat{\delta}^\gate$ equals to
    \begin{align*}
        \+E_{W,\varepsilon,t}\left[\hat{\delta}^\gate  - \delta^\gate\right] =& \+E_{W,\varepsilon,t} \left[ \mathcal{E}_{\mathrm{carryover}}\right] + \+E_{W,\varepsilon,t} \left[ \mathcal{E}_{\mathrm{simul}}\right]\\
        =& \underbrace{\delta^\co \left[\sum_{m = 1}^M I^{(m)} - 1 \right]}_{\bias(\mathcal{E}_\mathrm{carryover})} + \underbrace{\sum_{m = 1}^M \int_{t \in \mathcal{I}_{m}} \Phi^\simul_{t} f(t) dt}_{\bias(\mathcal{E}_\mathrm{simul})}
    \end{align*}
    following Lemmas \ref{lemma:mean-measurement-error}, \ref{lemma:carryover-bias}, and \ref{lemma:simul-bias}. \halmos
\end{proof}

 \subsection{Proof of Theorem \ref{theorem:bias-variance-switchback}}\label{subsec:proof-second-theorem}

 We can decompose the mean-squared error of $\hat{\delta}^\gate$ as follows.
 \begin{align*}
     &\+E_{W,\varepsilon,t}\left[\left(\hat{\delta}^\gate  - \delta^\gate \right)^2\right] = \+E_{W,\varepsilon,t} \left[ \left(\mathcal{E}_{\mathrm{meas}} + \mathcal{E}_{\mathrm{inst}}+ \mathcal{E}_{\mathrm{carryover}} + \mathcal{E}_{\mathrm{simul}} \right)^2 \right]  \\ =&  \+E_{W,\varepsilon,t} \left[ \left( \mathcal{E}_{\mathrm{meas}} \right)^2 \right]  + 2 \+E_{W,\varepsilon,t} \left[ \mathcal{E}_{\mathrm{meas}}\left( \mathcal{E}_{\mathrm{inst}} + \mathcal{E}_{\mathrm{carryover}} + \mathcal{E}_{\mathrm{simul}} \right) \right]   \\
     &  + \+E_{W,\varepsilon,t} \left[ \left(\mathcal{E}_{\mathrm{inst}}\right)^2 \right] + \+E_{W,\varepsilon,t} \left[ \left(\mathcal{E}_{\mathrm{carryover}}\right)^2 \right] +  2 \+E_{W,\varepsilon,t} \left[ \mathcal{E}_{\mathrm{carryover}} \mathcal{E}_{\mathrm{inst}} \right]  \\
     &+ \+E_{W,\varepsilon,t} \left[ \left(\mathcal{E}_{\mathrm{simul}}\right)^2 \right]  + 2\+E_{W,\varepsilon,t} \left[\mathcal{E}_{\mathrm{simul}} \mathcal{E}_{\mathrm{inst}}   \right]  + 2\+E_{W,\varepsilon,t} \left[ \mathcal{E}_{\mathrm{simul}} \mathcal{E}_{\mathrm{carryover}}  \right] \, .
 \end{align*}

Below we show the value of each term in the decomposition separately.

 \begin{lemma}[Second moment of $\mathcal{E}_{\mathrm{meas}}$]\label{lemma:measurement-error}
     Under the assumptions in Theorem \ref{theorem:bias-variance-switchback}, 
     the second moment of the measurement error is
     \begin{align*}
     & \+E_{W,\varepsilon,t} \left[ \left( \mathcal{E}_{\mathrm{meas}}\right)^2 \right] 
         = \+E_{W,\varepsilon,t} \left[\left( \frac{1}{n} \sum_{i = 1}^n \alpha_{t_i} \varepsilon^{(i)}\right)^2\right] =\frac{4}{n} \sum_{m=1}^M \left(V^{(m)}+ (n-1)C^{(m)}\right)  \, .
     \end{align*}
 \end{lemma}

\begin{proof}{Proof of Lemma \ref{lemma:measurement-error}}

The second moment of $\mathcal{E}_{\mathrm{meas}}$ equals
\begin{align*}
    & \+E_{W,\varepsilon,t} \left[ \left( \mathcal{E}_{\mathrm{meas}} \right)^2 \right] = \frac{1}{n^2} \sum_{i,j} \+E_{W,\varepsilon,t} \left[ \alpha_{t_i} \alpha_{t_j} \varepsilon^{(i)}  \varepsilon^{(j)}\right] \\
        =& \frac{1}{n}\+E_{W,t} \left[\alpha_{t_i}^2 \+E_{\varepsilon}\left[\left(\varepsilon^{(i)}\right)^2 \mid  \bm{W}, \bm{W}^\s_1, \cdots, \bm{W}^\s_K, t_i\right]\right] + \frac{n-1}{n} \+E_{W,t} \left[ \+E_{\varepsilon}\left[\varepsilon^{(i)} \varepsilon^{(j)} \mid  \bm{W}, \bm{W}^\s_1, \cdots, \bm{W}^\s_K, t_i, t_j \right]\right] \\
        =& \frac{1}{n\pi (1 - \pi)} \sum_{m=1}^M \underbrace{\int_{t_i \in \mathcal{I}_{m} } \+E_{\varepsilon}\left[(\varepsilon^{(i)})^2  \mid  t_i \right] f(t_i) d t_i}_{V^{(m)}}  \\
        & + \frac{n-1}{n\pi (1 - \pi)} \sum_{m=1}^M \underbrace{\int_{t_i, t_j \in \mathcal{I}_{m} } \+E_{\varepsilon}\left[\varepsilon^{(i)} \varepsilon^{(j)} \mid  t_i, t_j \right] f(t_i) f(t_j) d t_i d t_j }_{C^{(m)}} \, ,
\end{align*}
where we use the following property to show the expression of the second term in the last equation
\begin{align*}
    \+E_W \left[\alpha_{t_i} \alpha_{t_j} \mid t_i, t_j \right] =& \+E_W \left[\frac{(W_{t_i} - \pi)(W_{t_j} - \pi)}{\pi^2 (1-\pi)^2} \mid t_i, t_j\right] \\ =& \begin{cases}
    \frac{1}{\pi (1- \pi) } & \text{$t_i$ and $t_j$ in the same interval} \\
    0 &  \text{otherwise.}
\end{cases} 
\end{align*}

    Setting $\pi$ as $1/2$, we have 
    \begin{align*}
        \+E_{W,\varepsilon,t} \left[ \left( \mathcal{E}_{\mathrm{meas}} \right)^2 \right] 
         = \frac{4}{n} \sum_{m=1}^M \left(V^{(m)}+ (n-1)C^{(m)}\right) .
    \end{align*} 
    We then finish the proof of Lemma \ref{lemma:measurement-error}. \halmos
\end{proof}

\begin{lemma}[Expected value of the product of $\mathcal{E}_{\mathrm{meas}}$ and $\mathcal{E}_{\mathrm{inst}}+\mathcal{E}_{\mathrm{carryover}}+\mathcal{E}_{\mathrm{simul}}$]\label{lemma:expected-product-meas-other-terms}

    Under the assumptions in Theorem \ref{theorem:bias-variance-switchback}, the expected product of $\mathcal{E}_{\mathrm{meas}}$ and $\mathcal{E}_{\mathrm{inst}}+\mathcal{E}_{\mathrm{carryover}}+\mathcal{E}_{\mathrm{simul}}$ is equal to 
    \begin{align*}
        \+E_{W,\varepsilon,t} \left[  \mathcal{E}_{\mathrm{meas}} \left( \mathcal{E}_{\mathrm{inst}}+\mathcal{E}_{\mathrm{carryover}}+\mathcal{E}_{\mathrm{simul}}\right) \right] =&~ 0 \, .
    \end{align*}
\end{lemma}

\begin{proof}{Proof of Lemma \ref{lemma:expected-product-meas-other-terms}}
    
First, for the expected value of the product of $\mathcal{E}_{\mathrm{meas}}$ and $ $$\mathcal{E}_{\mathrm{inst}}$, we have
\begin{align*}
    \+E_{W,\varepsilon,t} \left[  \mathcal{E}_{\mathrm{meas}} \mathcal{E}_{\mathrm{inst}} \right] =& \+E_{W,\varepsilon,t} \left[ \left( \frac{1}{n} \sum_{i = 1}^n \alpha_{t_i} \varepsilon^{(i)}\right) \left( \delta^\gate \left[\frac{1}{n} \sum_{i = 1}^n \frac{W_{t_i}}{\pi} - 1\right] \right)\right] \\=&\+E_{W,\varepsilon,t} \left[ \left( \frac{1}{n} \sum_{i = 1}^n \alpha_{t_i} \+E_\varepsilon\left[\varepsilon^{(i)} \mid \bm{W}, \bm{W}^\s_1, \cdots, \bm{W}^\s_K ,  t \right] \right) \left( \delta^\gate \left[\frac{1}{n} \sum_{i = 1}^n \frac{W_{t_i}}{\pi} - 1\right] \right)\right] \\
    =& 0  \, . \tag{$\+E_\varepsilon\left[\varepsilon^{(i)} \mid \bm{W}, \bm{W}^\s_1, \cdots, \bm{W}^\s_K ,  t_i \right] = 0$ for all $i$}
\end{align*}

Second, the expected value of the product of $\mathcal{E}_{\mathrm{meas}}$ and $ $$\mathcal{E}_{\mathrm{carryover}}$ is equal to  
    \begin{align*}
        &\+E_{W,\varepsilon,t} \left[ \mathcal{E}_{\mathrm{meas}} \mathcal{E}_{\mathrm{carryover}} \right] \\=&\+E_{W,\varepsilon,t}\left[\left(\frac{1}{n} \sum_{i = 1}^n \alpha_{t_i} \+E_\varepsilon\left[\varepsilon^{(i)} \mid \bm{W}, \bm{W}^\s_1, \cdots, \bm{W}^\s_K,  t_i\right] \right) \times\left(\frac{1}{n} \sum_{i = 1}^n \alpha_{t_i} \left(Y_{t_i}(\bm{W}, \bm{0}, \cdots, \bm{0})  - W_{t} \delta^\gate\right)\right) \right] \\ =& 0 \, .
    \end{align*}
    
    Third, the expected value of the product of $\mathcal{E}_{\mathrm{meas}}$ and $ $$\mathcal{E}_{\mathrm{simul}}$ is equal to  
    \begin{align*}
        &\+E_{W,\varepsilon,t} \left[ \mathcal{E}_{\mathrm{meas}} \mathcal{E}_{\mathrm{simul}} \right] \\=& \+E_{W,\varepsilon,t}\left[\left(\frac{1}{n} \sum_{i = 1}^n \alpha_{t_i} \+E_\varepsilon\left[\varepsilon^{(i)} \mid \bm{W}, \bm{W}^\s_1, \cdots, \bm{W}^\s_K,  t_i\right] \right) \cdot \left(\frac{1}{n} \sum_{i = 1}^n\alpha_{t_i} \left(Y_{t_i}(\bm{W}, \bm{W}^\s_1, \cdots, \bm{W}^\s_K) - Y_{t_i}(\bm{W}, \bm{0}, \cdots, \bm{0})  \right)\right) \right] \\ =& 0 \, .
    \end{align*}
    
    Then we have 
    \begin{align*}
        & \+E_{W,\varepsilon,t} \left[  \mathcal{E}_{\mathrm{meas}} \left( \mathcal{E}_{\mathrm{inst}}+\mathcal{E}_{\mathrm{carryover}}+\mathcal{E}_{\mathrm{simul}}\right) \right] \\ =&~ \+E_{W,\varepsilon,t} \left[  \mathcal{E}_{\mathrm{meas}} \mathcal{E}_{\mathrm{inst}} \right] + \+E_{W,\varepsilon,t} \left[  \mathcal{E}_{\mathrm{meas}} \mathcal{E}_{\mathrm{carryover}} \right] + \+E_{W,\varepsilon,t} \left[  \mathcal{E}_{\mathrm{meas}} \mathcal{E}_{\mathrm{simul}}\right] = 0 \, .
    \end{align*} 
    We then finish the proof of Lemma \ref{lemma:expected-product-meas-other-terms}.\halmos
\end{proof}

\begin{lemma}[Second moment of $\mathcal{E}_{\mathrm{inst}}$] \label{lemma:second-moment-inst}

Under the assumptions in Theorem \ref{theorem:switchback-bias}, the second moment of $\mathcal{E}_{\mathrm{inst}}$ is equal to 
    \begin{align*}
        \+E_{W,\varepsilon,t} \left[ \left(  \mathcal{E}_{\mathrm{inst}}\right)^2 \right] =& (\delta^\gate)^2 \sum_{m=1}^M \left[\mu^{(m)}\right]^2 \, .
    \end{align*}
\end{lemma}

\begin{proof}{Proof of Lemma \ref{lemma:second-moment-inst}}
    
    The second moment of $\mathcal{E}_{\mathrm{inst}}$ equals 
    \begin{align*}
        \+E_{W,\varepsilon,t} \left[ \left( \mathcal{E}_{\mathrm{inst}} \right)^2 \right] 
        =& \frac{(\delta^\gate)^2}{n^2} \sum_{i,j} \+E_{W,t} \left[ \left( \frac{W_{t_i}}{\pi} - 1\right) \left( \frac{W_{t_j}}{\pi} - 1\right)\right] \\
        =& (\delta^\gate)^2 \+E_{W,t} \left[ \left( \frac{W_{t_i}}{\pi} - 1\right) \left( \frac{W_{t_j}}{\pi} - 1\right)\right] 
 \\ =& (\delta^\gate)^2 \left(\frac{1}{\pi} -1 \right)\sum_{m=1}^M \underbrace{\int_{t_i, t_j \in \mathcal{I}_{m} } f(t_i) f(t_j) d t_i d t_j }_{\left[\mu^{(m)}\right]^2} \, ,
    \end{align*}
    where we use the following property to show the last equation
    \begin{align*}
        \+E_W \left[\left( \frac{W_{t_i}}{\pi} - 1\right) \left( \frac{W_{t_j}}{\pi} - 1\right) \mid t_i, t_j \right] = \begin{cases}
    \frac{\pi - 1}{\pi} & \text{$t_i$ and $t_j$ in the same interval} \\
    0 &  \text{otherwise.}
\end{cases} 
    \end{align*}
    Setting $\pi = 1/2$, we have 
        \begin{align*}
        \+E_{W,\varepsilon,t} \left[ \left(  \mathcal{E}_{\mathrm{inst}}\right)^2 \right] =& (\delta^\gate)^2 \sum_{m=1}^M \left[\mu^{(m)}\right]^2 \, .
    \end{align*}
      We then finish the proof of Lemma \ref{lemma:second-moment-inst}. \halmos
\end{proof}

\begin{lemma}[Second moment of $\mathcal{E}_{\mathrm{carryover}}$]\label{lemma:second-moment-carryover}
    Under the assumptions in Theorem \ref{theorem:bias-variance-switchback}, the second moment of $\mathcal{E}_{\mathrm{carryover}}$ equals to 
    \begin{align*}
        & \+E_{W,\varepsilon,t} \left[ \left(\mathcal{E}_{\mathrm{carryover}}\right)^2\right] \\ =&  \sum_{m=1}^M \left(\Xi^{\inst,(m)}_\dem  + \Xi^{\co,(m)}_\dem  + 2 \mu^{(m)}_{Y^\mathrm{ctrl}}\right)^2 + \left(\sum_{m = 1}^M I^{(m)} - \delta^\co\right)^2   + \sum_{m=1}^M \sum_{m^\prime \neq m} \left(\left[I^{(m,m^\prime)}\right]^2 + I^{(m,m^\prime)} I^{(m^\prime, m)}\right)  \, .
    \end{align*}
\end{lemma}

\begin{proof}{Proof of Lemma \ref{lemma:second-moment-carryover}}
    As events are sampled i.i.d. from distribution $f(t)$, we have
    \begin{align*}
        &\+E_{W,\varepsilon,t} \left[ \left(\mathcal{E}_{\mathrm{carryover}}\right)^2\right] \\
        =& \+E_{W,t}\left[\left(\frac{1}{n} \sum_{i = 1}^n \left[\alpha_{t_i} (\delta^\inst_{t_i} - \delta^\inst) W_{t_i} + \alpha_{t_i}( \delta^\co_{t_i}(\bm{W}) - \delta^\co W_{t_i}) \right] + \frac{1}{n} \sum_{i = 1}^n \alpha_{t_i} Y_{t_i}(\bm{0}, \cdots, \bm{0}) \right)^2\right] \tag{carryover effects do not depend on $\varepsilon$} \\
        =& \underbrace{\+E_{W,t}\left[\left(\frac{1}{n} \sum_{i = 1}^n \left[\alpha_{t_i} (\delta^\inst_{t_i} - \delta^\inst) W_{t_i} + \alpha_{t_i}( \delta^\co_{t_i}(\bm{W}) - \delta^\co W_{t_i}) \right] \right)^2\right] }_{A^{\mathrm{treat}}}\\
        & + 2 \underbrace{\+E_{W,t}\left[\left(\frac{1}{n} \sum_{i = 1}^n \left[\alpha_{t_i} (\delta^\inst_{t_i} - \delta^\inst) W_{t_i} + \alpha_{t_i}( \delta^\co_{t_i}(\bm{W}) - \delta^\co W_{t_i}) \right] \right) \left( \frac{1}{n} \sum_{i = 1}^n \alpha_{t_i} Y_{t_i}(\bm{0}, \cdots, \bm{0}) \right)\right]}_{A^{\mathrm{cross}}} \\
        & + \underbrace{\+E_{W,t}\left[\left( \frac{1}{n} \sum_{i = 1}^n \alpha_{t_i} Y_{t_i}(\bm{0}, \cdots, \bm{0}) \right)^2\right]}_{A^{\mathrm{control}}} \, .
    \end{align*}
    Let us first consider the simplest term $A^{\mathrm{control}}$. For this term, we have 
    \begin{align*}
        A^{\mathrm{control}} =& \frac{1}{n^2} \sum_{i,j}\+E_{W,t}\left[\alpha_{t_i} Y_{t_i}(\bm{0}, \cdots, \bm{0}) \cdot \alpha_{t_j} Y_{t_j}(\bm{0}, \cdots, \bm{0}) \right] \\
        =& \+E_{W,t}\left[\alpha_{t_i} Y_{t_i}(\bm{0}, \cdots, \bm{0}) \cdot \alpha_{t_j} Y_{t_j}(\bm{0}, \cdots, \bm{0}) \right] \\
        =& \sum_{m = 1}^M \int_{t_i, t_j \in \mathcal{I}_{m}}  \+E_{W}\left[\alpha_{t_i} \alpha_{t_j} \mid t_i, t_j \right] Y_{t_i}(\bm{0}, \cdots, \bm{0})  Y_{t_j}(\bm{0}, \cdots, \bm{0}) f(t_i) f(t_j) dt_i dt_j \\
        & + \sum_{m = 1}^M \sum_{m^\prime: m^\prime \neq m} \int_{t_i \in \mathcal{I}_{m}, t_j \in \mathcal{I}_{m^\prime}}\+E_{W}\left[\alpha_{t_i} \alpha_{t_j} \mid t_i, t_j \right] Y_{t}(\bm{0}, \cdots, \bm{0})  Y_{t_j}(\bm{0}, \cdots, \bm{0}) f(t_i) f(t_j) dt_i dt_j \\
        =& \frac{1}{\pi (1 - \pi)} \sum_{m = 1}^M \int_{t, t^\prime \in \mathcal{I}_{m}}  Y_{t}(\bm{0}, \cdots, \bm{0})  Y_{t^\prime}(\bm{0}, \cdots, \bm{0}) f(t) f(t^\prime) dt dt^\prime \\
        =& \frac{1}{\pi (1 - \pi)} \sum_{m = 1}^M \left(\int_{t \in \mathcal{I}_{m}}  Y_{t}(\bm{0}, \cdots, \bm{0}) f(t) dt \right)^2 = \frac{1}{\pi (1 - \pi)} \sum_{m = 1}^M \left[\mu^{(m)}_{Y^\mathrm{ctrl}} \right]^2 
    \end{align*}
    following the definition that $\mu^{(m)}_{Y^\mathrm{ctrl}}  = \int_{t \in \mathcal{I}_{m}}  Y_{t}(\bm{0}, \cdots, \bm{0}) f(t) dt$. 

    Next we consider the cross term $A^{\mathrm{cross}}$. For this term, we have 
    \begin{align*}
        A^{\mathrm{cross}} =& \frac{1}{n^2} \sum_{i,j} \+E_{W,t}\left[\left(\alpha_{t_i} (\delta^\inst_{t_i} - \delta^\inst) W_{t_i} + \alpha_{t_i}( \delta^\co_{t_i}(\bm{W}) - \delta^\co W_{t_i}) \right)  \alpha_{t_j} Y_{t_j}(\bm{0}, \cdots, \bm{0}) \right] \\
        =& \underbrace{\+E_{W,t}\left[\alpha_{t_i} (\delta^\inst_{t_i} - \delta^\inst) W_{t_i}  \cdot  \alpha_{t_j} Y_{t_j}(\bm{0}, \cdots, \bm{0}) \right]}_{B_{1,ij}} \\
        & + \underbrace{\+E_{W,t}\left[\alpha_{t_i}( \delta^\co_{t_i}(\bm{W}) - \delta^\co W_{t_i}) \cdot \alpha_{t_j} Y_{t_j}(\bm{0}, \cdots, \bm{0}) \right]}_{B_{2,ij}} \, .
    \end{align*}
    For $B_{1,ij}$ (approximation error of instantaneous effect), we have 
    \begin{align*}
        B_{1,ij} =& \+E_{W,t}\left[\frac{W_{t_i}}{\pi} \frac{W_{t_j} - \pi}{\pi (1 - \pi)}  (\delta^\inst_{t_i} - \delta^\inst)  \cdot Y_{t_j}(\bm{0}, \cdots, \bm{0}) \right] \\
        =& \sum_{m = 1}^M \int_{t_i, t_j \in \mathcal{I}_{m}}  \+E_{W}\left[\frac{W_{t_i}}{\pi} \frac{W_{t_j} - \pi}{\pi (1 - \pi)}  \mid t_i, t_j \right] (\delta^\inst_{t_i} - \delta^\inst)  \cdot Y_{t_j}(\bm{0}, \cdots, \bm{0}) f(t_i) f(t_j) dt dt_j \\
        & + \sum_{m = 1}^M \sum_{m^\prime: m^\prime \neq m} \int_{t_i \in \mathcal{I}_{m},  t_j \in \mathcal{I}_{m^\prime}}  \underbrace{\+E_{W}\left[\frac{W_{t_i}}{\pi} \frac{W_{t_j} - \pi}{\pi (1 - \pi)}  \mid t_i, t_j \right]}_{= 0 } \times (\delta^\inst_{t_i} - \delta^\inst)  \cdot Y_{t_j}(\bm{0}, \cdots, \bm{0}) f(t_i) f(t_j) dt_i dt_j \\
        =& \frac{1}{\pi} \sum_{m = 1}^M \int_{t, t^\prime \in \mathcal{I}_{m}}   (\delta^\inst_{t} - \delta^\inst)  \cdot Y_{t^\prime}(\bm{0}, \cdots, \bm{0}) f(t) f(t^\prime) dt dt^\prime \\
        =& \frac{1}{\pi} \sum_{m = 1}^M \Xi^{\inst,(m)}_\dem  \mu^{(m)}_{Y^\mathrm{ctrl}} \, .
    \end{align*}
    For $B_{2,ij}$ (carryover effect approximation error), we have 
    \begin{align*}
        & B_{2,ij} = 
        \underbrace{\+E_{W,t}\left[\alpha_{t_i}\delta^\co_{t_i}(\bm{W}) \cdot \alpha_{t_j} Y_{t_j}(\bm{0}, \cdots, \bm{0}) \right]}_{C_{1,ij}} - \underbrace{\delta^\co \+E_{W,t}\left[\frac{W_{t_i}}{\pi} \alpha_{t_j} Y_{t_j}(\bm{0}, \cdots, \bm{0}) \right]}_{C_{2,ij}} \, .
    \end{align*}
    $C_{1,ij}$ equals to 
    \begin{align*}
        C_{1,ij} =&   \sum_{m, m^\prime = 1}^M  \int_{t_i \in \mathcal{I}_{m},  t_j \in \mathcal{I}_{m^\prime}}  \+E_{W}\left[\frac{W_{t_i} - \pi}{\pi (1 - \pi)} \frac{W_{t_j} - \pi}{\pi (1 - \pi)} \left(\sum_{k = 1}^M   W^{(k)} \int_{u \in \mathcal{I}_{k}} d^\co_{t_i}(u) f(u) d u \right)  \mid t_i, t_j \right]   \\
        & \qquad \qquad \qquad \times \delta^\co_{t_i} Y_{t_j}(\bm{0}, \cdots, \bm{0}) f(t_i) f(t_j) dt_i dt_j \\
        =& \frac{1}{\pi}  \sum_{m = 1}^M  \int_{t,  t^\prime  \in \mathcal{I}_{m}}  \left[\int_{u \in \mathcal{I}_{m}} d^\co_{t}(u) f(u) d u \right]  \delta^\co_{t}  Y_{t^\prime}(\bm{0}, \cdots, \bm{0}) f(t) f(t^\prime) dt dt^\prime \tag{$m = m^\prime = k$} \\
        & + \frac{1}{1 - \pi}  \sum_{m = 1}^M  \int_{t,  t^\prime  \in \mathcal{I}_{m}}  \left[\sum_{k: k \neq m} \int_{u \in \mathcal{I}_{k}} d^\co_{t}(u) f(u) d u \right]  \delta^\co_{t}   Y_{t^\prime}(\bm{0}, \cdots, \bm{0}) f(t) f(t^\prime) dt dt^\prime \tag{$m = m^\prime \neq k$} \\
        =& 2  \sum_{m = 1}^M  \int_{t,  t^\prime  \in \mathcal{I}_{m}}  \delta^\co_{t}   Y_{t^\prime}(\bm{0}, \cdots, \bm{0}) f(t) f(t^\prime) dt dt^\prime \\
        =& 2 \sum_{m=1}^M \Xi^{\co,(m)} \mu^{(m)}_{Y^\mathrm{ctrl}} \, .
    \end{align*}
    $C_{2,ij}$ equals to 
    \begin{align*}
        C_{2,ij} =& \delta^\co \+E_{W,t}\left[\frac{W_{t_i}}{\pi} \alpha_{t_j} Y_{t_j}(\bm{0}, \cdots, \bm{0}) \right] \\
        =& \delta^\co \sum_{m=1}^M \int_{t_i, t_j \in \mathcal{I}_{m}} \underbrace{\+E_W\left[ \frac{W_{t_i}}{\pi} \alpha_{t_j} \mid t_i, t_j\right]}_{= 1/\pi} Y_{t_j}(\bm{0}, \cdots, \bm{0}) f(t_i) f(t_j) dt dt_j \tag{$t_i$ and $t_j$ in the same interval} \\
        & +\delta^\co \sum_{m=1}^M \sum_{m^\prime: m^\prime \neq m} \int_{t_i \in \mathcal{I}_{m},  t_j \in \mathcal{I}_{m^\prime}} \underbrace{\+E_W\left[ \frac{W_{t_i}}{\pi} \alpha_{t_j} \mid t_i, t_j\right]}_{=0} Y_{t_j}(\bm{0}, \cdots, \bm{0}) f(t_i) f(t_j) dt_i dt_j \tag{$t_i$ and $t_j$ in different intervals} \\
        =& \frac{\delta^\co}{\pi} \sum_{m=1}^M \int_{t, t^\prime \in \mathcal{I}_{m}}  Y_{t^\prime}(\bm{0}, \cdots, \bm{0}) f(t) f(t^\prime) dt dt^\prime \\
        =& 2 \delta^\co \sum_{m=1}^M \mu^{(m)}  \mu^{(m)}_{Y^\mathrm{ctrl}} \, .
    \end{align*}
    Combining $C_{1,ij}$ and $C_{2,ij}$, $B_{2,ij}$ is equal to
    \begin{align*}
        B_{2,ij} =& 2 \sum_{m = 1}^M \left( \Xi^{\co,(m)} - \delta^\co \mu^{(m)}   \right)\mu^{(m)}_{Y^\mathrm{ctrl}} = 2 \sum_{m = 1}^M  \Xi^{\co,(m)}_\dem \mu^{(m)}_{Y^\mathrm{ctrl}} \, .
    \end{align*}
    Combining $B_{1,ij}$ and $B_{2,ij}$, $A^{\mathrm{cross}}$ is then equal to 
    \begin{align*}
        A^{\mathrm{cross}} =& 2 \sum_{m = 1}^M \Xi^{\inst,(m)}_\dem  \mu^{(m)}_{Y^\mathrm{ctrl}} + 2 \sum_{m = 1}^M  \Xi^{\co,(m)}_\dem \mu^{(m)}_{Y^\mathrm{ctrl}} = 2  \sum_{m = 1}^M  (\Xi^{(m)} - \delta^\gate \mu^{(m)} ) \cdot \mu^{(m)}_{Y^\mathrm{ctrl}} \, .
    \end{align*}

    Last, we consider the term $A^{\mathrm{treat}}$ (second moment of the treatment effect approximation error). For this term, we have 
    \begin{align*}
        A^{\mathrm{treat}} =& \frac{1}{n^2} \sum_{i,j} \underbrace{\+E_{W,t}\left[ \alpha_{t_i} (\delta^\inst_{t_i} - \delta^\inst) W_{t_i} \cdot \alpha_{t_j} (\delta^\inst_{t_j} - \delta^\inst) W_{t_j} \right] }_{\coloneqq A_{1,ij}}\\
        & + \frac{2}{n^2} \sum_{i,j} \underbrace{\+E_{W,t}\left[ \alpha_{t_i} (\delta^\inst_{t_i} - \delta^\inst) W_{t} \cdot \alpha_{t_j} ( \delta^\co_{t_j}(\bm{W}) - \delta^\co W_{t_j}) \right]}_{\coloneqq A_{2,ij}} \\
        & + \frac{1}{n^2} \sum_{i,j} \underbrace{\+E_{W,t}\left[ \alpha_{t_i} ( \delta^\co_{t_i}(\bm{W}) - \delta^\co W_{t_i}) \alpha_{t_j} ( \delta^\co_{t_j}(\bm{W}) - \delta^\co W_{t_j}) \right] }_{\coloneqq A_{3,ij}} \, .
    \end{align*}
    Below we compute each of $A_{1,ij}$, $A_{2,ij}$, and $A_{3,ij}$. We first compute $A_{1,ij}$. 
    \begin{align*}
        A_{1,ij} =&\+E_{W,t}\left[ \frac{W_{t_i}}{\pi}  \frac{W_{t_j}}{\pi} (\delta^\inst_{t_i} - \delta^\inst)  (\delta^\inst_{t_j} - \delta^\inst) \right] \\
        =& \int_{t, t^\prime \in [0, T]} (\delta^\inst_{t} - \delta^\inst)  (\delta^\inst_{t^\prime} - \delta^\inst) f(t) f(t^\prime) dt dt^\prime \\
        & + \left(\frac{1}{\pi} - 1 \right) \sum_{m = 1}^M \int_{t, t^\prime \in \mathcal{I}_m} (\delta^\inst_{t} - \delta^\inst)  (\delta^\inst_{t^\prime} - \delta^\inst) f(t) f(t^\prime) dt dt^\prime \tag{additional term for $t_i$ and $t_j$ in the same interval} \\
 =&  \sum_{m = 1}^M \int_{t, t^\prime \in \mathcal{I}_m} (\delta^\inst_{t} - \delta^\inst)  (\delta^\inst_{t^\prime} - \delta^\inst) 
 f(t) f(t^\prime) dt dt^\prime = \sum_{m = 1}^M \left[\Xi^{\inst,(m)}_\dem\right]^2 \, . \tag{the first term is zero following that $\+E_{t}[\delta^\inst_{t^\prime}] = \delta^\inst$ and $\pi = 1/2$}
    \end{align*}

    Next we compute $A_{2,ij}$. 
    \begin{align*}
        A_{2,ij} =& \underbrace{\+E_{W,t}\left[\frac{W_{t_i}}{\pi} \frac{W_{t_j} - \pi}{\pi (1 - \pi)} (\delta^\inst_{t_i} - \delta^\inst)  \delta^\co_{t_j}(\bm{W})\right]}_{\coloneqq B_{1,ij}} - \underbrace{\delta^\co \+E_{W,t}\left[\frac{W_{t_i}}{\pi} \frac{W_{t_j}}{\pi} (\delta^\inst_{t_j} - \delta^\inst)  \right] }_{\coloneqq B_{2,ij}} \, .
    \end{align*}
    For $B_{1,ij}$, we have
    \begin{align*}
        B_{1,ij} =& \sum_{m=1}^M \int_{t_i,t_j \in \mathcal{I}_{m}} (\delta^\inst_{t_i} - \delta^\inst) \delta^\co_{t_j} \times \\ &  \+E_{W}\left[\frac{W_{t_i}}{\pi}  \frac{W_{t_j} - \pi}{\pi (1- \pi)}  \left(\sum_{k = 1}^M   W^{(k)} \int_{u \in \mathcal{I}_{k}} d^\co_{t_j}(u) f(u) d u \right) \mid t_i, t_j\right] f(t_i) f(t_j) dt_i dt_j \tag{$t_i$ and $t_j$ in the same interval}  \\
        & + \sum_{m=1}^M \sum_{m^\prime: m^\prime \neq m}\int_{t_i \in \mathcal{I}_{m}, t_j \in \mathcal{I}_{m^\prime}} (\delta^\inst_{t_i} - \delta^\inst)\delta^\co_{t_j} \times \\ & \left. \+E_{W}\left[\frac{W_{t_i}}{\pi}  \frac{W_{t_j} - \pi}{\pi (1- \pi)}  \left( \sum_{k = 1}^M   W^{(k)} \int_{u \in \mathcal{I}_{k}} d^\co_{t_j}(u) f(u) d u \right) \right.\mid t_i, t_j\right] f(t_i) f(t_j) dt_i dt_j  \tag{$t_i$ and $t_j$ in different intervals}  \\
        =& \sum_{m=1}^M  \underbrace{\left(\int_{t \in \mathcal{I}_{m}} (\delta^\inst_{t} - \delta^\inst)   f(t) d t \right)}_{\Xi^{\inst,(m)}_\dem} \underbrace{\left(\int_{t^\prime \in \mathcal{I}_{m}} \delta^\co_{t^\prime}  \left(\sum_{k = 1}^M   \int_{u \in \mathcal{I}_{k}} d^\co_{t^\prime}(u) f(u) d u \right) f(t^\prime) d t^\prime \right) }_{\Xi^{\co,(m)}}\\
        & + \left(\frac{1}{\pi} - 1 \right) \sum_{m=1}^M \left(\int_{t \in \mathcal{I}_{m}} (\delta^\inst_{t} - \delta^\inst)   f(t) d t \right) \underbrace{\left(\int_{t^\prime \in \mathcal{I}_{m}} \delta^\co_{t^\prime}  \left(\int_{u \in \mathcal{I}_{m}} d^\co_{t^\prime}(u) f(u) d u \right) f(t^\prime) d t^\prime \right)}_{I^{(m)}} \tag{additional term for $m = k$}  \\
        & +  \sum_{m=1}^M \sum_{m^\prime: m^\prime \neq m} \left(\int_{t \in \mathcal{I}_{m}} (\delta^\inst_{t} - \delta^\inst)  f(t) d t\right) \underbrace{\left(\int_{t^\prime \in \mathcal{I}_{m^\prime}} \delta^\co_{t^\prime}  \left(\int_{u \in \mathcal{I}_{m^\prime}} d^\co_{t^\prime}(u) f(u) d u \right) f(t^\prime) d t^\prime \right)}_{I^{(m^\prime)}}  \tag{$m^\prime = k$} \\
        =& \sum_{m = 1}^M \Xi^{\inst,(m)}_\dem \Xi^{\co,(m)} + \delta^\co \left(\sum_{m^\prime = 1}^M I^{(m^\prime)} \right) \underbrace{\left( \sum_{m = 1}^M \int_{t \in \mathcal{I}_{m}} (\delta^\inst_{t} - \delta^\inst)  f(t) d t\right)}_{= 0} \tag{$\pi = 1/2$ and combine the second and third terms} \\
        =& \sum_{m = 1}^M \Xi^{\inst,(m)}_\dem \Xi^{\co,(m)} \, .
    \end{align*}
    For $B_{2,ij}$, we have
    \begin{align*}
        B_{2,ij} =&\delta^\co \sum_{m=1}^M \int_{t_i, t_j \in \mathcal{I}_{m}} (\delta^\inst_{t_j} - \delta^\inst)\+E_{W}\left[\frac{W_{t_i}}{\pi} \frac{W_{t_j}}{\pi}  \mid t_i, t_j \right] f(t_i) f(t_j) d t_i d t_j \tag{$t_i$ and $t_j$ in the same interval} \\
        &+ \delta^\co \sum_{m=1}^M \sum_{m^\prime: m^\prime \neq m} \int_{t_i \in \mathcal{I}_{m} , t_j \in \mathcal{I}_{m^\prime}} (\delta^\inst_{t_j} - \delta^\inst)\+E_{W}\left[\frac{W_{t_i}}{\pi} \frac{W_{t_j}}{\pi}  \mid t_i, t_j \right] f(t_i) f(t_j) d t_i d t_j \tag{$t_i$ and $t_j$ in different intervals}  \\
        =&\frac{\delta^\co }{\pi} \sum_{m=1}^M \mu^{(m)} \int_{t^\prime \in \mathcal{I}_{m}} (\delta^\inst_{t^\prime} - \delta^\inst) f(t^\prime)  d t^\prime +\delta^\co  \sum_{m=1}^M \sum_{m^\prime: m^\prime \neq m}  \mu^{(m)} \int_{ t^\prime \in \mathcal{I}_{m^\prime}} (\delta^\inst_{t^\prime} - \delta^\inst)  f(t^\prime)  d t^\prime \\
        =& \delta^\co \sum_{m=1}^M \mu^{(m)} \int_{t^\prime \in \mathcal{I}_{m}} (\delta^\inst_{t^\prime} - \delta^\inst) f(t^\prime)  d t^\prime = \delta^\co \sum_{m=1}^M \mu^{(m)} \Xi^{\inst,(m)}_\dem \, ,
    \end{align*}
    by the definition of $\delta^\inst_{\ell,t}$ and $\pi = 1/2$. Combining $B_{1,ij}$ and $B_{2,ij}$,  $A_{2,ij}$ is equal to
    \begin{align*}
        A_{2,ij} =&~ B_{1,ij} - B_{2,ij} = \sum_{m = 1}^M \Xi^{\inst,(m)}_\dem \left(\Xi^{\co,(m)} - \delta^\co  \mu^{(m)} \right) = \sum_{m = 1}^M \Xi^{\inst,(m)}_\dem \Xi^{\co,(m)}_\dem  \, .
    \end{align*}

    Finally we compute $A_{3,ij}$. 
    \begin{align*}
        A_{3,ij} =& \underbrace{\+E_{W,t}\left[ \frac{W_{t_i} - \pi}{\pi (1- \pi)} \frac{W_{t_j} - \pi}{\pi (1- \pi)}  \delta^\co_{t_i}(\bm{W}) \delta^\co_{t_j}(\bm{W}) \right]}_{\coloneqq B_{1,ij}} - \underbrace{\delta^\co \+E_{W,t}\left[\frac{W_{t_i} - \pi}{\pi (1- \pi)} \frac{W_{t_j}}{\pi}   \delta^\co_{t_i}(\bm{W}) \right]}_{\coloneqq  B_{2,ij}} \\
        & - \underbrace{\delta^\co \+E_{W,t}\left[\frac{W_{t_i}}{\pi}  \frac{W_{t_j} - \pi}{\pi (1- \pi)}   \delta^\co_{t_j}(\bm{W})\right]}_{\coloneqq B_{3,ij}} + \underbrace{(\delta^\co)^2 \+E_{W,t}\left[ \frac{W_{t_i}}{\pi} \frac{W_{t_j}}{\pi} \right]}_{\coloneqq B_{4,ij}} \, .
    \end{align*}
    For $B_{1,ij}$, we have
    \begin{align*}
        B_{1,ij} =& \sum_{m,m^\prime=1}^M \int_{t_i \in \mathcal{I}_{m}, t_j \in \mathcal{I}_{m^\prime}} \delta^\co_{t_i} \delta^\co_{t_j} \+E_{W}\left[\frac{W_{t_i} - \pi}{\pi (1- \pi)} \frac{W_{t_j} - \pi}{\pi (1- \pi)}  \left(\sum_{k = 1}^M   W^{(k)} \int_{u \in \mathcal{I}_{k}} d^\co_{t_i}(u) f(u) d u \right) \right. \\ & \qquad \left.\left(\sum_{k^\prime = 1}^M   W^{(k^\prime)} \int_{u \in \mathcal{I}_{k^\prime}} d^\co_{t_j}(u) f(u) d u \right) \mid t_i, t_j\right] f(t_i) f(t_j) dt_i dt_j \, .
    \end{align*}
    When $m = m^\prime$, the expectation $\+E_W[\cdot \mid t_i, t_j]$ in the last equation is equal to 
    \begin{align*}
        & \+E_{W}\left[\frac{(W_{t_i} - \pi)^2}{\pi^2 (1- \pi)^2}  \left(\sum_{k = 1}^M   W^{(k)} \int_{u \in \mathcal{I}_{k}} d^\co_{t_i}(u) f(u) d u \right)\left(\sum_{k^\prime = 1}^M   W^{(k^\prime)} \int_{u \in \mathcal{I}_{k^\prime}} d^\co_{t_j}(u) f(u) d u \right) \mid t_i, t_j\right] \\
        =& \frac{1}{\pi} \left(\int_{u \in \mathcal{I}_{m}} d^\co_{t}(u) f(u) d u \right) \left(\int_{u \in \mathcal{I}_{m}} d^\co_{t^\prime}(u) f(u) d u \right) \tag{$k = k^\prime = m$} \\ & + \frac{1}{1 - \pi} \sum_{k: k \neq m}  \left(\int_{u \in \mathcal{I}_{k}} d^\co_{t}(u) f(u) d u \right) \left(\int_{u \in \mathcal{I}_{k}} d^\co_{t^\prime}(u) f(u) d u \right) \tag{$k = k^\prime$ and both do not equal to $m$} \\ & + \sum_{k^\prime: k^\prime \neq m}  \left(\int_{u \in \mathcal{I}_{m}} d^\co_{t}(u) f(u) d u \right) \left(\int_{u \in \mathcal{I}_{k^\prime}} d^\co_{t^\prime}(u) f(u) d u \right) \tag{$k = m$, but $k^\prime \neq m$} \\ 
        & + \sum_{k: k \neq m}  \left(\int_{u \in \mathcal{I}_{k}} d^\co_{t}(u) f(u) d u \right) \left(\int_{u \in \mathcal{I}_{m}} d^\co_{t^\prime}(u) f(u) d u \right) \tag{$k \neq m$, but $k^\prime = m$} \\
        & + \frac{\pi}{ 1-\pi} \sum_{k,k^\prime: k \neq m, k^\prime \neq m, k\neq k^\prime} \left(\int_{u \in \mathcal{I}_{k}} d^\co_{t}(u) f(u) d u \right) \left(\int_{u \in \mathcal{I}_{k^\prime}} d^\co_{t^\prime}(u) f(u) d u \right) \tag{$k \neq m$, $k^\prime \neq m$, $k \neq k^\prime$}  \\
        =&1 + \sum_{k = 1}^M  \left(\int_{u \in \mathcal{I}_{k}} d^\co_{t}(u) f(u) d u \right) \left(\int_{u \in \mathcal{I}_{k}} d^\co_{t^\prime}(u) f(u) d u \right)  \tag{$\pi = 1/2$}
    \end{align*}
    where we use $\sum_{k = 1}^K \int_{u \in \mathcal{I}_{k}} d^\co_{t}(u) f(u) d u = 1$ for any $t$. 

    When $m \neq m^\prime$, the term $\+E_W[\cdot \mid t_i, t_j]$ is equal to 
    \begin{align*}
        \+E_W[\cdot \mid t_i, t_j] =& \left(\int_{u \in \mathcal{I}_{m}} d^\co_{t}(u) f(u) d u \right) \left(\int_{u \in \mathcal{I}_{m^\prime}} d^\co_{t^\prime}(u) f(u) d u \right) \\
        & + \left(\int_{u \in \mathcal{I}_{m^\prime}} d^\co_{t}(u) f(u) d u \right) \left(\int_{u \in \mathcal{I}_{m}} d^\co_{t^\prime}(u) f(u) d u \right) \, .
    \end{align*}
    
    Combining both cases ($m = m^\prime$ and $m \neq m^\prime$), $B_{1,ij}$ is equal to 
    \begin{align*}
        B_{1,ij} =& \sum_{m=1}^M \int_{t,t^\prime \in \mathcal{I}_{m}} \delta^\co_{t} \delta^\co_{t^\prime} \left[1 + \sum_{k = 1}^M  \left(\int_{u \in \mathcal{I}_{k}} d^\co_{t}(u) f(u) d u \right) \left(\int_{u \in \mathcal{I}_{k}} d^\co_{t^\prime}(u) f(u) d u \right) \right] f(t) f(t^\prime) dt dt^\prime \\
        & + \sum_{m,m^\prime: m\neq m^\prime} \int_{t \in \mathcal{I}_{m}, t^\prime \in \mathcal{I}_{m^\prime}} \delta^\co_{t} \delta^\co_{t^\prime} \left(\int_{u \in \mathcal{I}_{m}} d^\co_{t}(u) f(u) d u \right) \left(\int_{u \in \mathcal{I}_{m^\prime}} d^\co_{t^\prime}(u) f(u) d u \right) \\
        & + \delta^\co_{t} \delta^\co_{t^\prime}\left(\int_{u \in \mathcal{I}_{m^\prime}} d^\co_{t}(u) f(u) d u \right) \left(\int_{u \in \mathcal{I}_{m}} d^\co_{t^\prime}(u) f(u) d u \right)  f(t) f(t^\prime) dt dt^\prime  \\
        =&  \sum_{m=1}^M \left[\Xi^{\co,(m)}\right]^2 + \sum_{m=1}^M \sum_{m^\prime = 1}^M \left[I^{(m,m^\prime)}\right]^2 + \sum_{m=1}^M \sum_{m^\prime: m^\prime \neq m} \left( I^{(m)}  I^{(m^\prime)}  + I^{(m,m^\prime)} I^{(m^\prime, m)}\right) \, ,
    \end{align*}
    where 
    \[I^{(m,m^\prime)} = \int_{t \in \mathcal{I}_{m},  u \in \mathcal{I}_{m^\prime}} \delta^\co_{t}  d^\co_{t}(u) f(t) f(u)  dt d u \, .  \]
    
    $B_{2,ij}$ is equal to
    \begin{align*}
        & B_{2,ij} = \delta^\co \+E_{W,t}\left[\frac{W_{t_i} - \pi}{\pi (1- \pi)} \frac{W_{t_j}}{\pi}   \delta^\co_{t_i}(\bm{W}) \right] \\ 
        =& \delta^\co \+E_{W,t}\left[\frac{W_{t_i} - \pi}{\pi (1- \pi)} \frac{W_{t_j}}{\pi} \delta^\co_{t_i}  \sum_{k = 1}^M   W^{(k)} \int_{u \in \mathcal{I}_{k}} d^\co_{t_i}(u) f(u) d u \right] \\
        =& \delta^\co \sum_{m=1}^M \int_{t_i, t_j \in \mathcal{I}_{m}} \delta^\co_{t} \+E_{W}\left[\frac{W_{t}}{\pi^2} \left(\sum_{k = 1}^M   W^{(k)} \int_{u \in \mathcal{I}_{k}} d^\co_{t}(u) f(u) d u \right) \mid t \right]f(t_i)  f(t_j) dt_i dt_j \tag{$t_i$ and $t_j$ in the same interval} \\ 
        & +\delta^\co \sum_{m=1}^M \sum_{m^\prime: m^\prime \neq m} \int_{t_i \in \mathcal{I}_{m}, t_j \in \mathcal{I}_{m^\prime}} \delta^\co_{t_i} \cdot \+E_{W}\left[\frac{W_{t_i} - \pi}{\pi (1- \pi)} \frac{W_{t_j}}{\pi}  \left(\sum_{k = 1}^M   W^{(k)} \int_{u \in \mathcal{I}_{k}} d^\co_{t_i}(u) f(u) d u \right) \mid t \right]f(t_i)  f(t_j) dt_i dt_j \tag{$t_i$ and $t_j$ in different intervals} \\
        =& \delta^\co \sum_{m=1}^M \mu^{(m)} \int_{t \in \mathcal{I}_{m}} \delta^\co_{t} \left(1 +  \int_{u \in \mathcal{I}_{m}} d^\co_{t}(u) f(u) d u \right) f(t)  dt \\ 
        & + \delta^\co \sum_{m=1}^M \sum_{m^\prime: m^\prime \neq m} \mu^{(m^\prime)} \int_{t \in \mathcal{I}_{m}} \delta^\co_{t} \left[\int_{u \in \mathcal{I}_{m}} d^\co_{t}(u) f(u) d u \right] f(t) dt \tag{$m = k$} \\
        =& \delta^\co \sum_{m=1}^M \mu^{(m)} \left(\Xi^{\co,(m)} + I^{(m)} \right) + \delta^\co \sum_{m=1}^M \sum_{m^\prime: m^\prime \neq m} \mu^{(m^\prime)} I^{(m)} \\
        =& \delta^\co \left(\sum_{m = 1}^M I^{(m)} +  \sum_{m=1}^M  \mu^{(m)} \Xi^{\co,(m)}   \right) \, .
    \end{align*}
    Similarly we can show that $B_{3,ij} = B_{2,ij}$. 
    
    For $B_{4,ij}$, we have
    \begin{align*}
        B_{4,ij} =& (\delta^\co)^2  \+E_{W,t}\left[ \frac{W_{t_i}}{\pi} \frac{W_{t_j}}{\pi} \right] \\
        =& (\delta^\co)^2   \int  
 f(t_i) f(t_j) dt_i dt_j   +  (\delta^\co)^2\left(\frac{1}{\pi} - 1 \right)  \sum_{m = 1}^M \int_{t_i, t_j \in \mathcal{I}_m}  
 f(t_i) f(t_j) dt_i dt_j   \\
 =& (\delta^\co)^2 +  (\delta^\co)^2 \sum_{m = 1}^M \left[\mu^{(m)}\right]^2  \, .
    \end{align*}
    
    Combining $B_{1,ij}$, $B_{2,ij}$, $B_{3,ij}$ and $B_{4,ij}$, $A_{3,ij}$ is equal to 
    \begin{align*}
        A_{3,ij} =& \sum_{m=1}^M \left[\Xi^{\co,(m)}\right]^2 + \sum_{m=1}^M \sum_{m^\prime = 1}^M \left[I^{(m,m^\prime)}\right]^2 + \sum_{m=1}^M \sum_{m^\prime: m^\prime \neq m} \left(I^{(m)}  I^{(m^\prime)} + I^{(m,m^\prime)} I^{(m^\prime, m)}\right)\\
        & - 2 \delta^\co \left(\sum_{m = 1}^M  I^{(m)} +  \sum_{m=1}^M  \mu^{(m)} \Xi^{\co,(m)}   \right) + (\delta^\co)^2 +  (\delta^\co)^2 \sum_{m = 1}^M \left[\mu^{(m)}\right]^2  \\
        =& \sum_{m=1}^M \left[\Xi^{\co,(m)}_\dem\right]^2 +  \left(\delta^\co - \sum_{m=1}^M  I^{(m)}  \right)^2 + \sum_{m=1}^M \sum_{m^\prime: m^\prime \neq m} \left(\left[I^{(m,m^\prime)}\right]^2 + I^{(m,m^\prime)} I^{(m^\prime, m)}\right) \, ,
    \end{align*}
       where $\Xi^{\co,(m)}_\dem =  \Xi^{\co,(m)} -\delta^\co \mu^{(m)}$. 

       Combining $A_{1,ij}$, $A_{2,ij}$ and $A_{3,ij}$, $A^{\mathrm{treat}}$ is equal to
    \begin{align*}
        & A^{\mathrm{treat}} = A_{1,ij} + 2 A_{2,ij} + A_{3,ij} \\
        =&  \sum_{m = 1}^M \left[\Xi^{\inst,(m)}_\dem\right]^2 + 2 \sum_{m = 1}^M \Xi^{\inst,(m)}_\dem \Xi^{\co,(m)}_\dem + \sum_{m=1}^M \left[\Xi^{\co,(m)}_\dem\right]^2  \\ & +  \left(\delta^\co - \sum_{m=1}^M  I^{(m)}  \right)^2 +  \sum_{m=1}^M \sum_{m^\prime: m^\prime \neq m} \left(\left[I^{(m,m^\prime)}\right]^2 + I^{(m,m^\prime)} I^{(m^\prime, m)}\right) \\
        =& \sum_{m=1}^M \left(\Xi^{\inst,(m)}_\dem  + \Xi^{\co,(m)}_\dem \right)^2 + \left(\sum_{m = 1}^M I^{(m)} - \delta^\co\right)^2   + \sum_{m=1}^M \sum_{m^\prime \neq m} \left(\left[I^{(m,m^\prime)}\right]^2 + I^{(m,m^\prime)} I^{(m^\prime, m)}\right)  \, .
    \end{align*}

    Combining $A^{\mathrm{treat}}$, $A^{\mathrm{cross}}$, and $A^{\mathrm{control}}$, the second moment of $\mathcal{E}_{\mathrm{carryover}}$ is equal to 
    \begin{align*}
        &\+E_{W,\varepsilon,t} \left[ \left(\mathcal{E}_{\mathrm{carryover}}\right)^2\right] = A^{\mathrm{treat}} + 2A^{\mathrm{cross}} + A^{\mathrm{control}} \\
        =& \sum_{m=1}^M \left(\Xi^{\inst,(m)}_\dem  + \Xi^{\co,(m)}_\dem  \right)^2 + \left(\sum_{m = 1}^M I^{(m)} - \delta^\co\right)^2   + \sum_{m=1}^M \sum_{m^\prime \neq m} \left(\left[I^{(m,m^\prime)}\right]^2 + I^{(m,m^\prime)} I^{(m^\prime, m)}\right) \\
        & +4  \sum_{m = 1}^M  \Xi^{(m)} \mu^{(m)}_{Y^\mathrm{ctrl}} + 4\sum_{m = 1}^M \left[\mu^{(m)}_{Y^\mathrm{ctrl}} \right]^2 \\
        =&  \sum_{m=1}^M \left(\Xi^{\inst,(m)}_\dem  + \Xi^{\co,(m)}_\dem  + 2 \mu^{(m)}_{Y^\mathrm{ctrl}}\right)^2 + \left(\sum_{m = 1}^M I^{(m)} - \delta^\co\right)^2    + \sum_{m=1}^M \sum_{m^\prime \neq m} \left(\left[I^{(m,m^\prime)}\right]^2 + I^{(m,m^\prime)} I^{(m^\prime, m)}\right) \, .
    \end{align*} 
    We then finish the proof of Lemma \ref{lemma:second-moment-carryover}. \halmos
 
\end{proof}

\begin{lemma}[Expected value of the product of $\mathcal{E}_{\mathrm{inst}}$ and $ \mathcal{E}_{\mathrm{carryover}}$] \label{lemma:product-two-carryover-effects}

Under the assumptions in Theorem \ref{theorem:switchback-bias}, the expected value of the product of $\mathcal{E}_{\mathrm{inst}}$ and $\mathcal{E}_{\mathrm{carryover}}$ equals to 
    \begin{align*}
        \+E_{W,\varepsilon,t} \left[\mathcal{E}_{\mathrm{inst}} \mathcal{E}_{\mathrm{carryover}}  \right] =& \delta^\gate \sum_{m = 1}^M \mu^{(m)} \left(\Xi^{\inst,(m)}_\dem + \Xi^{\co,(m)}_\dem + 2  \mu^{(m)}_{Y^\mathrm{ctrl}}  \right) \, .
    \end{align*}
    
\end{lemma}

\begin{proof}{Proof of Lemma \ref{lemma:product-two-carryover-effects}}

    Next we compute the expected value of the product of $\mathcal{E}_{\mathrm{carryover}} $ and $\mathcal{E}_{\mathrm{inst}}$.
    \begin{align*}
    & \+E_{W,\varepsilon,t} \left[ \mathcal{E}_{\mathrm{inst}} \mathcal{E}_{\mathrm{carryover}}  \right] \\
        =& \+E_{W,\varepsilon,t}\left[\delta^\gate 
 \left(\frac{1}{n} \sum_{j = 1}^n \frac{W_{t_j}}{\pi} - 1 \right) \left(\frac{1}{n} \sum_{i = 1}^n \alpha_{t_i} \left(Y_{t_i}(\bm{W}, \bm{0}, \cdots, \bm{0}) - Y_{t_i}(\bm{0}, \cdots, \bm{0}) - W_{t_i} \delta^\gate\right)\right) \right] \\
 & + \+E_{W,\varepsilon,t}\left[\delta^\gate 
 \left(\frac{1}{n} \sum_{j = 1}^n \frac{W_{t_j}}{\pi} - 1 \right) \left(\frac{1}{n} \sum_{i = 1}^n \alpha_{t_i} Y_{t_i}(\bm{0}, \cdots, \bm{0})\right) \right] \\
        =& \frac{1}{n^2} \sum_{i,j} \underbrace{\delta^\gate \+E_t\bigg[(\delta^\inst_{t_i} - \delta^\inst) \cdot \frac{W_{t_i}}{\pi} \left(\frac{W_{t_j}}{\pi} - 1 \right) \bigg]}_{\coloneqq A_{1,ij}} \\
        & + \frac{1}{n^2} \sum_{i,j} \underbrace{\delta^\gate \+E_{W,t} \left[ \alpha_{t_i}( \delta^\co_{t_i}(\bm{W}) - \delta^\co W_{t_i}) \left(\frac{W_{t_j}}{\pi} - 1 \right)\right] }_{\coloneqq A_{2,ij}} \\
        & + \frac{1}{n^2} \sum_{i,j} \underbrace{\delta^\gate \+E_{W,t} \left[ \alpha_{t_i}  Y_{t_i}(\bm{0}, \cdots, \bm{0})\left(\frac{W_{t_j}}{\pi} - 1 \right)\right] }_{\coloneqq A_{3,ij}} \, .
    \end{align*}
    For $A_{1,ij}$, we have 
    \begin{align*}
        A_{1,ij} =& \delta^\gate \sum_{m = 1}^M \int_{t_i,t_j \in \mathcal{I}_{m}} (\delta^\inst_{t_i} - \delta^\inst) \cdot \+E_{W}\left[\frac{W_{t_i}}{\pi} \left(\frac{W_{t_j}}{\pi} - 1 \right) \mid t_i, t_j\right] f(t_i) f(t_j) dt_i dt_j \\
        & + \delta^\gate \sum_{m = 1}^M \sum_{m^\prime: m^\prime \neq m} \int_{t_i \in \mathcal{I}_{m},t_j \in \mathcal{I}_{m^\prime}} (\delta^\inst_{t_i} - \delta^\inst) \cdot \underbrace{\+E_{W}\left[\frac{W_{t_i}}{\pi} \left(\frac{W_{t_j}}{\pi} - 1 \right) \mid t_i, t_j\right]}_{= 0 ~\text{as $t_i$ and $t_j$ in different interval}} f(t_i) f(t_j) dt_i dt_j \\
        =& \delta^\gate  \frac{1-\pi}{\pi} \sum_{m = 1}^M \int_{t,t^\prime \in \mathcal{I}_{m}} (\delta^\inst_{t} - \delta^\inst)  f(t) f(t^\prime) dt dt^\prime \\
        =&  \delta^\gate \sum_{m = 1}^M \Xi^{\inst,(m)}_\dem \mu^{(m)} \, . \tag{$\pi = 1/2$}
    \end{align*}
    For $A_{2,ij}$, we have 
    \begin{align*}
        A_{2,ij} =& \underbrace{\delta^\gate \+E_{W,t} \left[ \alpha_{t_i} \delta^\co_{t_i}(\bm{W}) \left(\frac{W_{t_j}}{\pi} - 1 \right)\right] }_{\coloneqq B_{1,ij}}  - \underbrace{\delta^\gate \delta^\co\+E_{W,t} \left[ \alpha_{t_i}W_{t_i} \left(\frac{W_{t_j}}{\pi} - 1 \right)\right]}_{\coloneqq B_{2,ij}} \, .
    \end{align*}
    For $B_{2,ij}$, we have
    \begin{align*}
        B_{2,ij} =& \delta^\gate \delta^\co \+E_{W,t} \left[ \alpha_{t_i}W_{t_i} \left(\frac{W_{t_j}}{\pi} - 1 \right)\right] = \delta^\co \+E_{W,t} \left[ \frac{W_{t_i}}{\pi}  \left(\frac{W_{t_j}}{\pi} - 1 \right)\right] \\
        =& \delta^\gate \delta^\co \frac{1 -\pi}{\pi }\sum_{m=1}^M \int_{t, t^\prime \in \mathcal{I}_{m} } f(t) f(t^\prime) d t d t^\prime \tag{$t_i$ and $t_j$ in the same interval} \\
        =& \delta^\gate \delta^\co\sum_{m=1}^M \left[\mu^{(m)}\right]^2 \, . \tag{$\pi = 1/2$}
    \end{align*}
    For $B_{1,ij}$, we have
    \begin{align*}
        &B_{1,ij} = \delta^\gate \+E_{W,t} \left[ \alpha_{t_i} \delta^\co_{t_i}(\bm{W}) \left(\frac{W_{t_j}}{\pi} - 1 \right)\right] \\ 
        =& \delta^\gate \sum_{m = 1}^M \int_{t_i, t_j \in \mathcal{I}_{m}} \+E_W\left[\frac{W_{t_i} - \pi}{\pi (1 - \pi)} \frac{W_{t_j} - \pi}{\pi}  \delta^\co_{t_i}(\bm{W}) \mid t_i, t_j\right] f(t_i) f(t_j) dt_i d t_j  \\
        & + \delta^\gate \sum_{m = 1}^M \sum_{m^\prime: m^\prime \neq m} \int_{t_i \in \mathcal{I}_{m}, t_j \in \mathcal{I}_{m^\prime}} \+E_W\left[\frac{W_{t_i} - \pi}{\pi (1 - \pi)} \frac{W_{t_j} - \pi}{\pi}    \delta^\co_{t_i}(\bm{W}) \mid t_i, t_j\right] f(t_i) f(t_j) dt_i d t_j \tag{this term is zero as $t_i$ and $t_j$ in different intervals}  \\
        =& \delta^\gate \sum_{m = 1}^M \int_{t, t^\prime \in \mathcal{I}_{m}} \delta^\co_{t} \left(\frac{1 - \pi}{\pi} \left[\int_{u \in \mathcal{I}_{m}} d^\co_{t}(u) f(u) d u \right] + \sum_{k: k \neq m} \frac{\pi}{\pi }  \left[\int_{u \in \mathcal{I}_{k}} d^\co_{t}(u) f(u) d u \right] \right)   f(t) f(t^\prime) dt d t^\prime   \\
        =& \delta^\gate \sum_{m = 1}^M \int_{t, t^\prime \in \mathcal{I}_{m}} \delta^\co_{t} f(t) f(t^\prime) dt d t^\prime   =  \delta^\gate \sum_{m=1}^M \mu^{(m)} \Xi^{(m)} \, .
    \end{align*}
    Combining $B_{1,ij}$ and $B_{2,ij}$, $A_{2,ij}$ is equal to 
    \begin{align*}
        A_{2,ij} =\delta^\gate \sum_{m=1}^M \mu^{(m)} \Xi^{(m)} - \delta^\gate \delta^\co\sum_{m=1}^M \left[\mu^{(m)}\right]^2 = \delta^\gate \sum_{m=1}^M \mu^{(m)} \Xi^{(m)}_\dem   \, .
    \end{align*}
    For $A_{3,ij}$, we have 
    \begin{align*}
        A_{3,ij} =& \delta^\gate \+E_{W,t} \left[ \alpha_{t_i}  Y_{t_i}(\bm{0}, \cdots, \bm{0})\left(\frac{W_{t_j}}{\pi} - 1 \right)\right] \\
        =& 2 \delta^\gate \sum_{m = 1}^M \int_{t, t^\prime \in \mathcal{I}_{m}} Y_{t}(\bm{0}, \cdots, \bm{0}) f(t) f(t^\prime) dt d t^\prime   \\
        =& 2 \delta^\gate \sum_{m = 1}^M \mu^{(m)}\mu^{(m)}_{Y^\mathrm{ctrl}} \, .
    \end{align*}
    Combining $A_{1,ij}$, $A_{2,ij}$ and $A_{3,ij}$, $\+E_{W,\varepsilon,t} \left[ \mathcal{E}_{\mathrm{inst}} \cdot \mathcal{E}_{\mathrm{carryover}}  \right]$ is equal to
    \begin{align*}
        \+E_{W,\varepsilon,t} \left[ \mathcal{E}_{\mathrm{inst}} \cdot \mathcal{E}_{\mathrm{carryover}}  \right] =& A_{1,ij} + A_{2,ij}  + A_{3,ij} \\
        =&  \delta^\gate \sum_{m = 1}^M \mu^{(m)} \left(\Xi^{\inst,(m)}_\dem + \Xi^{\co,(m)}_\dem + 2  \mu^{(m)}_{Y^\mathrm{ctrl}}  \right) \, .
    \end{align*}
    We then finish the proof of Lemma \ref{lemma:product-two-carryover-effects}. \halmos
\end{proof}

\begin{lemma}[Second moment of $\mathcal{E}_{\mathrm{simul}}$]\label{lemma:simul-second-moment}
    Under the assumptions in Theorem \ref{theorem:switchback-bias}, the second moment of $\mathcal{E}_{\mathrm{simul}}$ is 
    \begin{align*}
         \+E_{W,\varepsilon,t} \left[ \left(\mathcal{E}_{\mathrm{simul}}\right)^2 \right]
        = \sum_{m = 1}^M \sum_{m^\prime = 1}^M S^{(m,m^\prime)}_{\mathrm{var}} \, ,
    \end{align*}
    where $S^{(m,m^\prime)}_{\mathrm{var}}$ is defined in Equation \eqref{eqn:S-m-mprime-var}.
\end{lemma}

\begin{proof}{Proof of Lemma \ref{lemma:simul-second-moment}}
As events are sampled i.i.d. from $f(t)$, we have
    \begin{align*}
        &\+E_{W,\varepsilon,t} \left[ \left(\mathcal{E}_{\mathrm{simul}}\right)^2 \right]
        = \+E_{W,\varepsilon,t} \left[\left(\frac{1}{n} \sum_{i = 1}^n\alpha_{t_i} \left(Y_{t_i}(\bm{W}, \bm{W}^\s_1, \cdots, \bm{W}^\s_K) - Y_{t_i}(\bm{W}, \bm{0}, \cdots, \bm{0})  \right) \right)^2 \right] \\
        =& \frac{1}{n^2} \sum_{i,j} \+E_{W,t}\left[\frac{W_{t_i} - \pi}{\pi (1 - \pi)} \frac{W_{t_j} - \pi}{\pi (1 - \pi)} \+E_{\bm{W}^\s_1, \cdots, \bm{W}^\s_K } \left[ \left(Y_{t_i}(\bm{W}, \bm{W}^\s_1, \cdots, \bm{W}^\s_K) - Y_{t}(\bm{W}, \bm{0}, \cdots, \bm{0})\right) \times \right. \right. \\ &\quad  \left(Y_{t_j}(\bm{W}, \bm{W}^\s_1, \cdots, \bm{W}^\s_K) - Y_{t_j}(\bm{W}, \bm{0}, \cdots, \bm{0})\right)  \left.\left. \mid \bm{W}, t_i, t_j \right]\right] \tag{by the law of total expectation} \\
        =& \+E_{W,t}\left[\frac{W_{t_i} - \pi}{\pi (1 - \pi)} \frac{W_{t_j} - \pi}{\pi (1 - \pi)} \cdot  \delta^{\simul,2}_{ t_i, t_j}(\bm{W}) \right] \tag{by definition of $\delta^{\simul,2}_{ t_i, t_j}(\bm{W})$ in Equation \eqref{eqn:delta-simul-2}} \\ 
        =&4 \sum_{m = 1}^M \int_{t,t^\prime \in \mathcal{I}_{m}} \+E_{W}[\delta^{\simul,2}_{ t, t^\prime}(\bm{W}) \mid t, t^\prime] f(t) f(t^\prime) dt dt^\prime  \\
        & + 4\sum_{m = 1}^M \sum_{m^\prime: m^\prime \neq m}\int_{t \in \mathcal{I}_{m},t^\prime \in \mathcal{I}_{m^\prime}} \Phi^{2\dagger}_{t,t^\prime} f(t) f(t^\prime) dt dt^\prime \tag{by the definition of $\Phi^{2\dagger}_{t,t^\prime}$ in Equation \eqref{eqn:phi-2} } \\
        =& \sum_{m = 1}^M \sum_{m^\prime = 1}^M S^{(m,m^\prime)}_{\mathrm{var}} \,
    \end{align*}
    following the definition of $S^{(m,m^\prime)}_{\mathrm{var}}$ in Equation \eqref{eqn:S-m-mprime-var}.
    We then finish the proof of Lemma \ref{lemma:simul-second-moment}. \halmos
\end{proof}

\begin{lemma}[Expected value of the product between $\mathcal{E}_{\mathrm{simul}}$ and $\mathcal{E}_{\mathrm{inst}} $]\label{lemma:product-simul-inst}
    Under the assumptions in Theorem \ref{theorem:bias-variance-switchback}, we have
    \begin{align*}
         \+E_{W,\varepsilon,t} \left[ \mathcal{E}_{\mathrm{simul}} \cdot \mathcal{E}_{\mathrm{inst}} \right]  
         =&   \sum_{m = 1}^M \sum_{m^\prime = 1}^M \delta^\gate \mu^{(m^\prime)} S^{(m,m^\prime)}_{1}   \, ,
     \end{align*}
     where $S^{(m,m^\prime)}_{1}$ is defined in Equation \eqref{eqn:S-m-mprime-1}.
\end{lemma}

\begin{proof}{Proof of Lemma \ref{lemma:product-simul-inst}}
    
    The expected value of the product of $\mathcal{E}_{\mathrm{simul}}$ and $\mathcal{E}_{\mathrm{inst}} $ is equal to
    \begin{align*}
        \+E_{W,\varepsilon,t} \left[ \mathcal{E}_{\mathrm{simul}} \mathcal{E}_{\mathrm{inst}}  \right]
        = \frac{1}{n^2} \sum_{i,j} \underbrace{\+E_{W,t} \left[ \alpha_{t_i} \left(Y_{t_i}(\bm{W}, \bm{W}^\s_1, \cdots, \bm{W}^\s_K) - Y_{t_i}(\bm{W}, \bm{0}, \cdots, \bm{0})  \right) \delta^\gate \left(\frac{W_{t_j}}{\pi} - 1 \right) \right] }_{A_{ij}} \, .
    \end{align*}
    Using the definition of $\delta^\simul_{ t_i}(\bm{W})$ in Section \ref{subsec:block-stats} (note that $\delta^\simul_{ t_i}(\bm{W}) =\delta^\simul_{ t_i}(\bm{W}_t)$ using the assumption of non-anticipating outcomes), $A_{ij}$ is equal to
    \begin{align*}
        A_{ij}&
        = \delta^\gate \+E_{W,t}\left[\frac{W_{t_i} - \pi}{\pi (1-\pi)} \left(\frac{W_{t_j}}{\pi} - 1 \right) \delta^\simul_{ t_i}(\bm{W}) \right] \\ =& \delta^\gate \sum_{m = 1}^M \int_{t, t_j \in \mathcal{I}_{m}} \+E_{W}\left[\frac{W_{t_i} - \pi}{\pi (1-\pi)} \left(\frac{W_{t_j}}{\pi} - 1 \right) \delta^\simul_{ t_i}(\bm{W}) \mid t_i, t_j\right] f(t_i) f(t_j) dt_i dt_j \\ & +  \delta^\gate \sum_{m = 1}^M \sum_{m^\prime: m^\prime \neq m} \int_{t_i \in \mathcal{I}_{m}, t_j\in \mathcal{I}_{m^\prime}} \+E_{W}\left[\frac{W_{t_i} - \pi}{\pi (1-\pi)} \left(\frac{W_{t_j}}{\pi} - 1 \right) \delta^\simul_{ t_i}(\bm{W}) \mid t_i, t_j\right] f(t_i) f(t_j) dt_i dt_j  \\ =& 2\delta^\gate  \sum_{m = 1}^M \int_{t, t^\prime \in \mathcal{I}_{m}}  \+E_{W }\left[ \delta^\simul_{ t}(\bm{W}) 
        \right] f(t) f(t^\prime) dt dt^\prime \\
        & +2\delta^\gate  \sum_{m = 1}^M \sum_{m^\prime: m^\prime \neq m} \int_{t \in \mathcal{I}_{m}, t^\prime\in \mathcal{I}_{m^\prime}}  \Phi_{t}^{\simul,(-m^\prime)} f(t) f(t^\prime) dt dt^\prime \, ,
    \end{align*}
    following the definition of $\Phi_{t}^{\simul,(-m^\prime)}$ in Equation \eqref{eqn:phi-simul}.
    Therefore,
    \begin{align*}
        \+E_{W,\varepsilon,t} \left[ \mathcal{E}_{\mathrm{inst}} \cdot \mathcal{E}_{\mathrm{simul}}  \right] =&2\delta^\gate  \sum_{m = 1}^M \mu^{(m)} \int_{t \in \mathcal{I}_{m}}  \+E_{W }\left[ \delta^\simul_{ t}(\bm{W}) 
        \right] f(t)  dt \\
        & + 2\delta^\gate  \sum_{m = 1}^M \sum_{m^\prime: m^\prime \neq m} \mu^{(m^\prime)} \int_{t \in \mathcal{I}_{m}}  \Phi_{t}^{\simul,(-m^\prime)} f(t)  dt \\
        =& \delta^\gate  \sum_{m = 1}^M \sum_{m^\prime = 1}^M \mu^{(m^\prime)} S^{(m,m^\prime)}_{1}  \, 
    \end{align*}
    following the definition of $S^{(m,m^\prime)}_{1}$ in Equation \eqref{eqn:S-m-mprime-1}.
    We then finish the proof of Lemma \ref{lemma:product-simul-inst}. \halmos
\end{proof}

\begin{lemma}[Expected value of the product between $\mathcal{E}_{\mathrm{simul}}$ and $\mathcal{E}_{\mathrm{carryover}} $]\label{lemma:product-simul-carryover}
Under the assumptions in Theorem \ref{theorem:bias-variance-switchback}, we have 
    \begin{align*}
        & \+E_{W,\varepsilon,t} \left[ \mathcal{E}_{\mathrm{simul}} \cdot \mathcal{E}_{\mathrm{carryover}} \right] = \sum_{m = 1}^M 
 \sum_{m^\prime = 1}^M \left[ 2 \mu^{(m^\prime)}_{Y^\mathrm{ctrl}}  S^{(m,m^\prime)}_{1}  + \left( \Xi^{\inst,(m)}_\dem - \delta^\co \mu^{(m)} \right) S^{(m,m^\prime)}_{2} + S^{(m,m^\prime)}_{3}\right]  \, ,
    \end{align*}
    where $S^{(m,m^\prime)}_{1}$, $S^{(m,m^\prime)}_{2}$, and $S^{(m,m^\prime)}_{3}$ are defined in Equations \eqref{eqn:S-m-mprime-1}, \eqref{eqn:S-m-mprime-2}, and \eqref{eqn:S-m-mprime-3} respectively.
\end{lemma}

\begin{proof}{Proof of Lemma \ref{lemma:product-simul-carryover}}
The expected value of the product of $\mathcal{E}_{\mathrm{simul}} $ and $\mathcal{E}_{\mathrm{carryover}}$ is equal to 
    \begin{align*}
        & \+E_{W,\varepsilon,t} \left[ \mathcal{E}_{\mathrm{carryover}}\mathcal{E}_{\mathrm{simul}} \right] \\ =& \frac{1}{n^2} \sum_{i,j} \underbrace{\+E_{W,t}\left[ \alpha_{t_i} (\delta^\inst_{t_i} - \delta^\inst) W_{t_i}  \alpha_{t_j} \left(Y_{t_j}(\bm{W}, \bm{W}^\s_1, \cdots, \bm{W}^\s_K) - Y_{t_j}(\bm{W}, \bm{0}, \cdots, \bm{0})  \right)  \right]}_{\coloneqq A_{1,ij}} \\
        +&  \frac{1}{n^2} \sum_{i,j} \underbrace{\+E_{W,t}\left[\alpha_{t_i}( \delta^\co_{t_i}(\bm{W}) - \delta^\co W_{t_i})  \alpha_{t_j} \left(Y_{t_j}(\bm{W}, \bm{W}^\s_1, \cdots, \bm{W}^\s_K) - Y_{t_j}(\bm{W}, \bm{0}, \cdots, \bm{0})  \right)  \right] }_{\coloneqq A_{2,ij}} \\
        +&  \frac{1}{n^2} \sum_{i,j} \underbrace{\+E_{W,t}\left[\alpha_{t_i} Y_{t_i}(\bm{0}, \cdots, \bm{0})  \alpha_{t_j} \left(Y_{t_j}(\bm{W}, \bm{W}^\s_1, \cdots, \bm{W}^\s_K) - Y_{t_j}(\bm{W}, \bm{0}, \cdots, \bm{0})  \right)  \right] }_{\coloneqq A_{3,ij}} \, .
    \end{align*}

    For $A_{1,ij}$, we have 
    \begin{align*}
        A_{1,ij} =& \+E_{W,t} \left[\frac{W_{t_i}}{\pi} \frac{W_{t_j} - \pi}{\pi (1 - \pi)} (\delta^\inst_{t_i} - \delta^\inst) \left(Y_{t_j}(\bm{W}, \bm{W}^\s_1, \cdots, \bm{W}^\s_K) - Y_{t_j}(\bm{W}, \bm{0}, \cdots, \bm{0})  \right)\right] \\
        =& \+E_{W,t} \left[\frac{W_{t_i}}{\pi} \frac{W_{t_j} - \pi}{\pi (1 - \pi)} (\delta^\inst_{t_i} - \delta^\inst) \delta^\simul_{ t_j}(\bm{W}) \right] \\
        =& \frac{1}{\pi} \sum_{m = 1}^M \int_{t, t^\prime \in \mathcal{I}_{m}} (\delta^\inst_{t} - \delta^\inst) \+E_{\bm{W}^{(-m)} }\Big[ \delta^\simul_{ t^\prime}(\bm{W}^{(-m)},  W^{(m)} = 1 ) \Big] f(t) f(t^\prime) dt dt^\prime \\
        & + \sum_{m = 1}^M \sum_{m^\prime: m^\prime \neq m} \int_{t \in \mathcal{I}_{m}, t^\prime \in \mathcal{I}_{m^\prime}} (\delta^\inst_{t} - \delta^\inst)\Phi_{t^\prime}^{\simul,(-m^\prime)\dagger} f(t) f(t^\prime) dt dt^\prime \tag{by the definition of $\Phi_{t^\prime}^{\simul,(-m^\prime)\dagger}$ in Equation \eqref{eqn:phi-simul-2}} \\
        =& 2 \sum_{m = 1}^M \Xi^{\inst,(m)}_\dem \int_{t^\prime \in \mathcal{I}_{m}} \+E_{\bm{W}^{(-m)} }\Big[ \delta^\simul_{ t^\prime}(\bm{W}^{(-m)},  W^{(m)} = 1 ) \Big] f(t^\prime)  dt^\prime \\
        & + 2\sum_{m = 1}^M \Xi^{\inst,(m)}_\dem 
 \sum_{m^\prime: m^\prime \neq m} \int_{ t^\prime \in \mathcal{I}_{m^\prime}} \Phi_{t^\prime}^{\simul,(-m^\prime)\dagger} f(t^\prime) dt^\prime \\
 =&  \sum_{m = 1}^M 
 \sum_{m^\prime = 1}^M \Xi^{\inst,(m)}_\dem S^{(m,m^\prime)}_{2} \, ,
    \end{align*}
    following the definition of $S^{(m,m^\prime)}_{2}$ in Equation \eqref{eqn:S-m-mprime-2}.

    For $A_{2,ij}$, we have
    \begin{align*}
        A_{2,ij} =&  \underbrace{\+E_{W,t} \left[\frac{W_{t_i} - \pi}{\pi (1 - \pi)} \frac{W_{t_j} - \pi}{\pi (1 - \pi)}  \delta^\co_{t_i}(\bm{W})  \delta^\simul_{ t_j}(\bm{W}) \right]}_{\coloneqq B_{1,ij}}   - \underbrace{\delta^\co\+E_{W,t} \left[\frac{W_{t_i} }{\pi } \frac{W_{t_j} - \pi}{\pi (1 - \pi)}    \delta^\simul_{ t_j}(\bm{W}) \right] }_{\coloneqq B_{2,ij}} \, .
    \end{align*}
    For $B_{1,ij}$, we have 
    \begin{align*}
        B_{1,ij} =& \+E_{W,t} \left[\frac{W_{t_i} - \pi}{\pi (1 - \pi)} \frac{W_{t_j} - \pi}{\pi (1 - \pi)}  \delta^\co_{t_i}(\bm{W})  \delta^\simul_{ t_j}(\bm{W}) \right] \\
        =& 4\sum_{m = 1}^M \int_{t, t^\prime \in \mathcal{I}_{m}} \+E_W\left[\delta^\co_{t}(\bm{W})  \delta^\simul_{ t^\prime}(\bm{W}) \right] f(t) f(t^\prime) dt dt^\prime \tag{$\pi = 1/2$} \\
        & +4\sum_{m = 1} \sum_{m^\prime: m^\prime \neq m} \int_{t \in \mathcal{I}_{m}, t^\prime \in \mathcal{I}_{m^\prime}} \Phi_{t, t^\prime}^{\co,\simul} f(t) f(t^\prime) dt dt^\prime \tag{by the definition of $\Phi_{t, t^\prime}^{\co,\simul}$ in Equation \eqref{eqn:phi-co-simul}} \\
        =& \sum_{m = 1}^M \sum_{m^\prime = 1}^M S^{(m,m^\prime)}_{3} \, ,
    \end{align*}
    following the definition of $S^{(m,m^\prime)}_{3}$ in Equation \eqref{eqn:S-m-mprime-3}.

    For $B_{2,ij}$, we have 
    \begin{align*}
        B_{2,ij} =& \delta^\co\+E_{W,t} \left[\frac{W_{t_i} }{\pi } \frac{W_{t_j} - \pi}{\pi (1 - \pi)}    \delta^\simul_{ t_j}(\bm{W}) \right] \\
        =& 2 \delta^\co \sum_{m = 1}^M  \mu^{(m)}\int_{t^\prime \in \mathcal{I}_{m}} \+E_{\bm{W}^{(-m)} }\bigg[ \delta^\simul_{ t^\prime}(\bm{W}^{(-m)},  W^{(m)} = 1 ) \bigg] f(t^\prime)  dt^\prime \\
        & + 2\delta^\co \sum_{m = 1}^M \mu^{(m)} \sum_{m^\prime: m^\prime \neq m} \int_{ t^\prime \in \mathcal{I}_{m^\prime}} \Phi_{t^\prime}^{\simul,(-m^\prime)\dagger} f(t^\prime) dt^\prime \tag{by the definition of $\Phi_{t^\prime}^{\simul,(-m^\prime)\dagger}$ in Equation \eqref{eqn:phi-simul-2}} \\ =&  \sum_{m = 1}^M 
 \sum_{m^\prime = 1}^M \delta^\co  \mu^{(m)} S^{(m,m^\prime)}_{2}\, .
    \end{align*}
    
    For $A_{3,ij}$, we have 
    \begin{align*}
        A_{3,ij} =& \+E_{W,t}\left[\alpha_{t_i} Y_{t_i}(\bm{0}, \cdots, \bm{0})  \alpha_{t_j} \left(Y_{t_j}(\bm{W}, \bm{W}^\s_1, \cdots, \bm{W}^\s_K) - Y_{t_j}(\bm{W}, \bm{0}, \cdots, \bm{0})  \right)  \right] \\
        =& \+E_{W,t} \left[\frac{W_{t_i} - \pi}{\pi (1 - \pi)}  \frac{W_{t_j} - \pi}{\pi (1 - \pi)} Y_{t_i}(\bm{0}, \cdots, \bm{0}) \left(Y_{t_j}(\bm{W}, \bm{W}^\s_1, \cdots, \bm{W}^\s_K) - Y_{t_j}(\bm{W}, \bm{0}, \cdots, \bm{0})  \right)\right] \\
        =&  4 \sum_{m = 1}^M \int_{t, t^\prime \in \mathcal{I}_{m}} Y_{t}(\bm{0}, \cdots, \bm{0}) \+E_W\left[\delta_{t^\prime}^\simul(\bm{W}) \right] f(t) f(t^\prime) dt dt^\prime \\
        & + 4 \sum_{m = 1}^M \sum_{m^\prime: m^\prime \neq m} \int_{t \in \mathcal{I}_{m}, t^\prime \in \mathcal{I}_{m^\prime}} Y_{t}(\bm{0}, \cdots, \bm{0}) \Phi_{t^\prime}^{\simul,(-m)} f(t) f(t^\prime) dt dt^\prime \tag{by the definition of $\Phi_{t^\prime}^{\simul,(-m)}$ in Equation \eqref{eqn:phi-simul}}  \\
        =& 4 \sum_{m = 1}^M \mu^{(m)}_{Y^\mathrm{ctrl}} \int_{t^\prime \in \mathcal{I}_{m}} \+E_W\left[\delta_{t^\prime}^\simul(\bm{W}) \right] f(t^\prime)  dt^\prime  + 4\sum_{m = 1}^M \mu^{(m)}_{Y^\mathrm{ctrl}}
 \sum_{m^\prime: m^\prime \neq m} \int_{ t^\prime \in \mathcal{I}_{m^\prime}} \Phi_{t^\prime}^{\simul,(-m)} f(t^\prime) dt^\prime \\
 =& 2 \sum_{m = 1}^M \sum_{m^\prime = 1}^M \mu^{(m^\prime)}_{Y^\mathrm{ctrl}} S^{(m,m^\prime)}_{1}  \, ,
    \end{align*}
    following the definition of $S^{(m,m^\prime)}_{1}$ in Equation \eqref{eqn:S-m-mprime-1}. Then we have
    \begin{align*}
        & \+E_{W,\varepsilon,t} \left[ \mathcal{E}_{\mathrm{simul}} \mathcal{E}_{\mathrm{carryover}}  \right] \\ 
        =& A_{1,ij} + \underbrace{B_{1,ij} - B_{2,ij}}_{A_{2,ij}} + A_{3,ij} \\
 =&  \sum_{m = 1}^M 
 \sum_{m^\prime = 1}^M \Xi^{\inst,(m)}_\dem S^{(m,m^\prime)}_{2} + \sum_{m = 1}^M \sum_{m^\prime = 1}^M S^{(m,m^\prime)}_{3} - \sum_{m = 1}^M 
 \sum_{m^\prime = 1}^M \delta^\co  \mu^{(m)} S^{(m,m^\prime)}_{2} + 2 \sum_{m = 1}^M \sum_{m^\prime = 1}^M \mu^{(m^\prime)}_{Y^\mathrm{ctrl}} S^{(m,m^\prime)}_{1} \\
 =& \sum_{m = 1}^M 
 \sum_{m^\prime = 1}^M \left[ 2 \mu^{(m^\prime)}_{Y^\mathrm{ctrl}}  S^{(m,m^\prime)}_{1}  + \left( \Xi^{\inst,(m)}_\dem - \delta^\co \mu^{(m)} \right) S^{(m,m^\prime)}_{2} + S^{(m,m^\prime)}_{3}\right]  \, .
    \end{align*}
We then finish the proof of Lemma \ref{lemma:product-simul-carryover}. 
    \halmos
\end{proof}

\begin{proof}{Proof of Theorem \ref{theorem:bias-variance-switchback}}
Following Lemmas \ref{lemma:second-moment-inst}, \ref{lemma:second-moment-carryover}, and \ref{lemma:product-two-carryover-effects}, the second moment of $\mathcal{E}_{\mathrm{inst}}+\mathcal{E}_{\mathrm{carryover}}$ is 
\begin{align*}
    &\+E_{W,\varepsilon,t} \left[ \left(\mathcal{E}_{\mathrm{inst}}+\mathcal{E}_{\mathrm{carryover}}\right)^2 \right] 
    =\+E_{W,\varepsilon,t} \left[ \left(\mathcal{E}_{\mathrm{inst}}\right)^2 \right] +  \+E_{W,\varepsilon,t} \left[ \left(\mathcal{E}_{\mathrm{carryover}}\right)^2 \right] +   2 \+E_{W,\varepsilon,t} \left[ \mathcal{E}_{\mathrm{inst}} \mathcal{E}_{\mathrm{carryover}} \right] \\
    =& \sum_{m = 1}^M \left(\Xi^{(m)} + 2  \mu^{(m)}_{Y^\mathrm{ctrl}}  \right)^2 + \left(\sum_{m = 1}^M I^{(m)} - \delta^\co\right)^2  + \sum_{m=1}^M \sum_{m^\prime \neq m} \left(\left[I^{(m,m^\prime)}\right]^2 + I^{(m,m^\prime)} I^{(m^\prime, m)}\right) \, ,
\end{align*}
where $\Xi^{(m)} = \Xi^{\inst,(m)}_\dem  + \Xi^{\co,(m)}_\dem +  \delta^\gate \mu^{(m)}  $.

Following Lemmas \ref{lemma:mean-measurement-error}-\ref{lemma:product-simul-carryover}, the mean-squared error of $\hat{\delta}^\gate$ is equal to
    \begin{align*}
        \+E_{W,\varepsilon,t}&\left[\left(\hat{\delta}^\gate  - \delta^\gate \right)^2\right]= \frac{4}{n} \sum_{m=1}^M \left(V^{(m)}+ (n-1)C^{(m)}\right) \tag{equals to $\+E_{W,\varepsilon,t} \left[ \left( \mathcal{E}_{\mathrm{meas}}  \right)^2 \right] $} \\
        & + \sum_{m = 1}^M \left(\Xi^{(m)} + 2  \mu^{(m)}_{Y^\mathrm{ctrl}}  \right)^2 + \left(\sum_{m = 1}^M I^{(m)} - \delta^\co\right)^2   + \sum_{m=1}^M \sum_{m^\prime \neq m} \left(\left[I^{(m,m^\prime)}\right]^2 + I^{(m,m^\prime)} I^{(m^\prime, m)}\right)  \tag{equals to $\+E_{W,\varepsilon,t} \left[ \left(\mathcal{E}_{\mathrm{inst}}+\mathcal{E}_{\mathrm{carryover}}\right)^2 \right] $} \\
        &+ \sum_{m = 1}^M \sum_{m^\prime = 1}^M S^{(m,m^\prime)}_{\mathrm{var}} \tag{equals to $\+E_{W,\varepsilon,t} \left[ \left(\mathcal{E}_{\mathrm{simul}}\right)^2 \right]$} \\
        %
        & + 2 \sum_{m = 1}^M \sum_{m^\prime = 1}^M \delta^\gate \mu^{(m^\prime)} S^{(m,m^\prime)}_{1} \tag{equals to $2 \+E_{W,\varepsilon,t} \left[ \mathcal{E}_{\mathrm{simul}} \mathcal{E}_{\mathrm{inst}} \right]  $} \\
        & + 2 \sum_{m = 1}^M 
 \sum_{m^\prime = 1}^M \left[ 2 \mu^{(m^\prime)}_{Y^\mathrm{ctrl}}  S^{(m,m^\prime)}_{1}  + \left( \Xi^{\inst,(m)}_\dem - \delta^\co \mu^{(m)} \right) S^{(m,m^\prime)}_{2} + S^{(m,m^\prime)}_{3}\right]  \tag{equals to $2 \+E_{W,\varepsilon,t} \left[ \mathcal{E}_{\mathrm{simul}} \mathcal{E}_{\mathrm{carryover}} \right]$}  \\
 =& \var(\mathcal{E}_\mathrm{meas}) +  \bias(\mathcal{E}_\mathrm{carryover}) ^2  + \var(\mathcal{E}_\mathrm{inst}+\mathcal{E}_\mathrm{carryover})  + \+E[\mathcal{E}_\mathrm{simul}^2] + 2 \+E[(\mathcal{E}_\mathrm{inst}+\mathcal{E}_\mathrm{carryover})\cdot \mathcal{E}_\mathrm{simul}] \, .
    \end{align*}
    using the definition of $S^{(m,m^\prime)}_{\mathrm{var}}$ in Equation \eqref{eqn:S-m-mprime-var} and the definition of $S^{(m,m^\prime)}_{\mathrm{cov}}$ in Equation \eqref{eqn:S-m-mprime-cov}.
    We then finish the proof of Theorem \ref{theorem:bias-variance-switchback}. \halmos
\end{proof}

%% file: main.bbl
\begin{thebibliography}{59}
\expandafter\ifx\csname natexlab\endcsname\relax\def\natexlab#1{#1}\fi
\expandafter\ifx\csname url\endcsname\relax
  \def\url#1{{\tt #1}}\fi
\expandafter\ifx\csname urlprefix\endcsname\relax\def\urlprefix{URL }\fi
\expandafter\ifx\csname urlstyle\endcsname\relax
  \expandafter\ifx\csname doi\endcsname\relax
  \def\doi#1{doi:\discretionary{}{}{}#1}\fi \else
  \expandafter\ifx\csname doi\endcsname\relax
  \def\doi{doi:\discretionary{}{}{}\begingroup \urlstyle{rm}\Url}\fi \fi

\bibitem[{Abadie and Zhao(2021)}]{abadie2021synthetic}
Abadie, Alberto, Jinglong Zhao. 2021.
\newblock Synthetic controls for experimental design.
\newblock {\it arXiv preprint arXiv:2108.02196\/} .

\bibitem[{Adam et~al.(2023)Adam, He, and Zheng}]{adam2023machine}
Adam, Hammaad, Pu~He, Fanyin Zheng. 2023.
\newblock Machine learning for demand estimation in long tail markets.
\newblock {\it Management Science\/} .

\bibitem[{Baird et~al.(2018)Baird, Bohren, McIntosh, and {\"O}zler}]{baird2018optimal}
Baird, Sarah, J~Aislinn Bohren, Craig McIntosh, Berk {\"O}zler. 2018.
\newblock Optimal design of experiments in the presence of interference.
\newblock {\it Review of Economics and Statistics\/} {\bf 100}(5) 844--860.

\bibitem[{Bajari et~al.(2023)Bajari, Burdick, Imbens, Masoero, McQueen, Richardson, and Rosen}]{bajari2021multiple}
Bajari, Patrick, Brian Burdick, Guido~W Imbens, Lorenzo Masoero, James McQueen, Thomas~S Richardson, Ido~M Rosen. 2023.
\newblock Experimental design in marketplaces.
\newblock {\it Statistical Science\/} {\bf 1}(1) 1--19.

\bibitem[{Basse and Feller(2018)}]{basse2018analyzing}
Basse, Guillaume, Avi Feller. 2018.
\newblock Analyzing two-stage experiments in the presence of interference.
\newblock {\it Journal of the American Statistical Association\/} {\bf 113}(521) 41--55.

\bibitem[{Basse et~al.(2023)Basse, Ding, and Toulis}]{basse2019minimax}
Basse, Guillaume~W, Yi~Ding, Panos Toulis. 2023.
\newblock Minimax designs for causal effects in temporal experiments with treatment habituation.
\newblock {\it Biometrika\/} {\bf 110}(1) 155--168.

\bibitem[{Bojinov et~al.(2020)Bojinov, Saint-Jacques, and Tingley}]{bojinov2020avoid}
Bojinov, Iavor, Guillaume Saint-Jacques, Martin Tingley. 2020.
\newblock Avoid the pitfalls of a/b testing.
\newblock {\it Harvard Business Review\/} .

\bibitem[{Bojinov et~al.(2023)Bojinov, Simchi-Levi, and Zhao}]{bojinov2020design}
Bojinov, Iavor, David Simchi-Levi, Jinglong Zhao. 2023.
\newblock Design and analysis of switchback experiments.
\newblock {\it Management Science\/} {\bf 69}(7) 3759--3777.

\bibitem[{Boyarsky et~al.(2023)Boyarsky, Namkoong, and Pouget-Abadie}]{boyarsky2023modeling}
Boyarsky, Ariel, Hongseok Namkoong, Jean Pouget-Abadie. 2023.
\newblock Modeling interference using experiment roll-out.
\newblock {\it arXiv preprint arXiv:2305.10728\/} .

\bibitem[{Candogan et~al.(2021)Candogan, Chen, and Niazadeh}]{candogan2021near}
Candogan, Ozan, Chen Chen, Rad Niazadeh. 2021.
\newblock Near-optimal experimental design for networks: Independent block randomization.
\newblock {\it Available at SSRN\/} .

\bibitem[{Chamandy(2016)}]{chamandy2016experiment}
Chamandy, Nicholas. 2016.
\newblock Experimentation in a ridesharing marketplace.

\bibitem[{Chen and Simchi-Levi(2023)}]{chen2023switchback}
Chen, Hongyu, David Simchi-Levi. 2023.
\newblock Switchback experiments in a reactive environment.
\newblock {\it Available at SSRN 4436643\/} .

\bibitem[{Chin(2018)}]{chin2018central}
Chin, Alex. 2018.
\newblock Central limit theorems via stein's method for randomized experiments under interference.
\newblock {\it arXiv preprint arXiv:1804.03105\/} .

\bibitem[{Chin(2019)}]{chin2019regression}
Chin, Alex. 2019.
\newblock Regression adjustments for estimating the global treatment effect in experiments with interference.
\newblock {\it Journal of Causal Inference\/} {\bf 7}(2).

\bibitem[{Cochran et~al.(1941)Cochran, Autrey, and Cannon}]{cochran1941double}
Cochran, WG, KM~Autrey, CY~Cannon. 1941.
\newblock A double change-over design for dairy cattle feeding experiments.
\newblock {\it Journal of Dairy Science\/} {\bf 24}(11) 937--951.

\bibitem[{Cooprider and Nassiri(2023)}]{Cooprider2023amazon}
Cooprider, Joe, Shima Nassiri. 2023.
\newblock The science of price experiments in the amazon store.
\newblock {\it Amazon science\/} .

\bibitem[{Cortez et~al.(2022)Cortez, Eichhorn, and Yu}]{cortez2022staggered}
Cortez, Mayleen, Matthew Eichhorn, Christina Yu. 2022.
\newblock Staggered rollout designs enable causal inference under interference without network knowledge.
\newblock {\it Advances in Neural Information Processing Systems\/} {\bf 35} 7437--7449.

\bibitem[{Cr{\'e}pon et~al.(2013)Cr{\'e}pon, Duflo, Gurgand, Rathelot, and Zamora}]{crepon2013labor}
Cr{\'e}pon, Bruno, Esther Duflo, Marc Gurgand, Roland Rathelot, Philippe Zamora. 2013.
\newblock Do labor market policies have displacement effects? evidence from a clustered randomized experiment.
\newblock {\it The quarterly journal of economics\/} {\bf 128}(2) 531--580.

\bibitem[{Dasgupta et~al.(2015)Dasgupta, Pillai, and Rubin}]{dasgupta2015causal}
Dasgupta, Tirthankar, Natesh~S Pillai, Donald~B Rubin. 2015.
\newblock Causal inference from 2k factorial designs by using potential outcomes.
\newblock {\it Journal of the Royal Statistical Society Series B: Statistical Methodology\/} {\bf 77}(4) 727--753.

\bibitem[{Doudchenko et~al.(2019)Doudchenko, Gilinson, Taylor, and Wernerfelt}]{doudchenko2019designing}
Doudchenko, Nick, David Gilinson, Sean Taylor, Nils Wernerfelt. 2019.
\newblock Designing experiments with synthetic controls.
\newblock Tech. rep., Working paper.

\bibitem[{Doudchenko et~al.(2021)Doudchenko, Khosravi, Pouget-Abadie, Lahaie, Lubin, Mirrokni, Spiess et~al.}]{doudchenko2021synthetic}
Doudchenko, Nick, Khashayar Khosravi, Jean Pouget-Abadie, Sebastien Lahaie, Miles Lubin, Vahab Mirrokni, Jann Spiess, et~al. 2021.
\newblock Synthetic design: An optimization approach to experimental design with synthetic controls.
\newblock {\it Advances in Neural Information Processing Systems\/} {\bf 34}.

\bibitem[{Eckles et~al.(2017)Eckles, Karrer, and Ugander}]{eckles2017design}
Eckles, Dean, Brian Karrer, Johan Ugander. 2017.
\newblock Design and analysis of experiments in networks: Reducing bias from interference.
\newblock {\it Journal of Causal Inference\/} {\bf 5}(1).

\bibitem[{Farias et~al.(2022)Farias, Li, Peng, and Zheng}]{farias2022markovian}
Farias, Vivek, Andrew Li, Tianyi Peng, Andrew Zheng. 2022.
\newblock Markovian interference in experiments.
\newblock {\it Advances in Neural Information Processing Systems\/} {\bf 35} 535--549.

\bibitem[{Fisher(1936)}]{fisher1936design}
Fisher, Ronald~Aylmer. 1936.
\newblock Design of experiments.
\newblock {\it British Medical Journal\/} {\bf 1}(3923) 554.

\bibitem[{Forastiere et~al.(2021)Forastiere, Airoldi, and Mealli}]{forastiere2021identification}
Forastiere, Laura, Edoardo~M Airoldi, Fabrizia Mealli. 2021.
\newblock Identification and estimation of treatment and interference effects in observational studies on networks.
\newblock {\it Journal of the American Statistical Association\/} {\bf 116}(534) 901--918.

\bibitem[{Han et~al.(2022)Han, Li, Mao, and Wu}]{han2022detecting}
Han, Kevin, Shuangning Li, Jialiang Mao, Han Wu. 2022.
\newblock Detecting interference in a/b testing with increasing allocation.
\newblock {\it arXiv preprint arXiv:2211.03262\/} .

\bibitem[{Holtz et~al.(2023)Holtz, Lobel, Lobel, Liskovich, and Aral}]{holtzreducing}
Holtz, David, Felipe Lobel, Ruben Lobel, Inessa Liskovich, Sinan Aral. 2023.
\newblock Reducing interference bias in online marketplace experiments using cluster randomization: Evidence from a pricing meta-experiment on airbnb.
\newblock {\it Management Science\/} .

\bibitem[{Horvitz and Thompson(1952)}]{horvitz1952generalization}
Horvitz, Daniel~G, Donovan~J Thompson. 1952.
\newblock A generalization of sampling without replacement from a finite universe.
\newblock {\it Journal of the American statistical Association\/} {\bf 47}(260) 663--685.

\bibitem[{Hu and Wager(2022)}]{hu2022switchback}
Hu, Yuchen, Stefan Wager. 2022.
\newblock Switchback experiments under geometric mixing.
\newblock {\it arXiv preprint arXiv:2209.00197\/} .

\bibitem[{Hudgens and Halloran(2008)}]{hudgens2008toward}
Hudgens, Michael~G, M~Elizabeth Halloran. 2008.
\newblock Toward causal inference with interference.
\newblock {\it Journal of the American Statistical Association\/} {\bf 103}(482) 832--842.

\bibitem[{Johari et~al.(2022)Johari, Li, Liskovich, and Weintraub}]{johari2020experimental}
Johari, Ramesh, Hannah Li, Inessa Liskovich, Gabriel~Y Weintraub. 2022.
\newblock Experimental design in two-sided platforms: An analysis of bias.
\newblock {\it Management Science\/} {\bf 68}(10) 7069--7089.

\bibitem[{Jones and Kenward(2003)}]{jones2003design}
Jones, Byron, Michael~G Kenward. 2003.
\newblock {\it Design and analysis of cross-over trials\/}.
\newblock Chapman and Hall/CRC.

\bibitem[{Karkar et~al.(2016)Karkar, Zia, Vilardaga, Mishra, Fogarty, Munson, and Kientz}]{karkar2016framework}
Karkar, Ravi, Jasmine Zia, Roger Vilardaga, Sonali~R Mishra, James Fogarty, Sean~A Munson, Julie~A Kientz. 2016.
\newblock A framework for self-experimentation in personalized health.
\newblock {\it Journal of the American Medical Informatics Association\/} {\bf 23}(3) 440--448.

\bibitem[{Kastelman and Ramesh(2018)}]{Kastelman2018switchback}
Kastelman, David, Raghav Ramesh. 2018.
\newblock Switchback tests and randomized experimentation under network effects at doordash.
\newblock {\it Medium\/} .

\bibitem[{Krafft et~al.(2018)Krafft, Della~Penna, and Pentland}]{krafft2018experimental}
Krafft, Peter~M, Nicol{\'a}s Della~Penna, Alex~Sandy Pentland. 2018.
\newblock An experimental study of cryptocurrency market dynamics.
\newblock {\it Proceedings of the 2018 CHI conference on human factors in computing systems\/}. 1--13.

\bibitem[{Leung(2023)}]{leung2023network}
Leung, Michael~P. 2023.
\newblock Network cluster-robust inference.
\newblock {\it Econometrica\/} {\bf 91}(2) 641--667.

\bibitem[{Li et~al.(2021)Li, Zhao, Johari, and Weintraub}]{li2021interference}
Li, Hannah, Geng Zhao, Ramesh Johari, Gabriel~Y Weintraub. 2021.
\newblock Interference, bias, and variance in two-sided marketplace experimentation: Guidance for platforms.
\newblock {\it arXiv preprint arXiv:2104.12222\/} .

\bibitem[{Li et~al.(2023)Li, Johari, Wager, and Xu}]{li2023experimenting}
Li, Shuangning, Ramesh Johari, Stefan Wager, Kuang Xu. 2023.
\newblock Experimenting under stochastic congestion.
\newblock {\it arXiv preprint arXiv:2302.12093\/} .

\bibitem[{Liu and Hudgens(2014)}]{liu2014large}
Liu, Lan, Michael~G Hudgens. 2014.
\newblock Large sample randomization inference of causal effects in the presence of interference.
\newblock {\it Journal of the american statistical association\/} {\bf 109}(505) 288--301.

\bibitem[{Masoero et~al.(2023)Masoero, Imbens, Richardson, McQueen, Vijaykumar, and Rosen}]{masoero2023efficient}
Masoero, Lorenzo, Guido Imbens, Thomas Richardson, James McQueen, Suhas Vijaykumar, Ido Rosen. 2023.
\newblock Efficient switchback experiments via multiple randomization designs.
\newblock {\it Amazon Science\/} .

\bibitem[{Mirza et~al.(2017)Mirza, Punja, Vohra, and Guyatt}]{mirza2017history}
Mirza, RD, S~Punja, S~Vohra, G~Guyatt. 2017.
\newblock The history and development of n-of-1 trials.
\newblock {\it Journal of the Royal Society of Medicine\/} {\bf 110}(8) 330--340.

\bibitem[{Ni et~al.(2023)Ni, Bojinov, and Zhao}]{ni2023design}
Ni, Tu, Iavor Bojinov, Jinglong Zhao. 2023.
\newblock Design of panel experiments with spatial and temporal interference.
\newblock {\it Available at SSRN 4466598\/} .

\bibitem[{Qu et~al.(2021)Qu, Xiong, Liu, and Imbens}]{qu2021efficient}
Qu, Zhaonan, Ruoxuan Xiong, Jizhou Liu, Guido Imbens. 2021.
\newblock Efficient treatment effect estimation in observational studies under heterogeneous partial interference.
\newblock {\it arXiv preprint arXiv:2107.12420\/} .

\bibitem[{Silbert(2022)}]{silbert2022switchback}
Silbert, Noah. 2022.
\newblock Switchback experimentation at tubi.
\newblock {\it Medium\/} .

\bibitem[{Simchi-Levi et~al.(2024)Simchi-Levi, Wang, and Zheng}]{simchi2024non}
Simchi-Levi, David, Chonghuan Wang, Zeyu Zheng. 2024.
\newblock Non-stationary experimental design under linear trends.
\newblock {\it Advances in Neural Information Processing Systems\/} {\bf 36}.

\bibitem[{Sinclair et~al.(2012)Sinclair, McConnell, and Green}]{sinclair2012detecting}
Sinclair, Betsy, Margaret McConnell, Donald~P Green. 2012.
\newblock Detecting spillover effects: Design and analysis of multilevel experiments.
\newblock {\it American Journal of Political Science\/} {\bf 56}(4) 1055--1069.

\bibitem[{Tritchler(1984)}]{tritchler1984inverting}
Tritchler, David. 1984.
\newblock On inverting permutation tests.
\newblock {\it Journal of the American Statistical Association\/} {\bf 79}(385) 200--207.

\bibitem[{Ugander et~al.(2013)Ugander, Karrer, Backstrom, and Kleinberg}]{ugander2013graph}
Ugander, Johan, Brian Karrer, Lars Backstrom, Jon Kleinberg. 2013.
\newblock Graph cluster randomization: Network exposure to multiple universes.
\newblock {\it Proceedings of the 19th ACM SIGKDD international conference on Knowledge discovery and data mining\/}. 329--337.

\bibitem[{Wager and Xu(2021)}]{wager2021experimenting}
Wager, Stefan, Kuang Xu. 2021.
\newblock Experimenting in equilibrium.
\newblock {\it Management Science\/} {\bf 67}(11) 6694--6715.

\bibitem[{Wu and Hamada(2011)}]{wu2011experiments}
Wu, CF~Jeff, Michael~S Hamada. 2011.
\newblock {\it Experiments: planning, analysis, and optimization\/}.
\newblock John Wiley \& Sons.

\bibitem[{Wu(2022)}]{wu2022dynamic}
Wu, Linjia. 2022.
\newblock Dynamic stochastic models for experimentation and matching.
\newblock {\it Stanford University\/} .

\bibitem[{Wu et~al.(2022)Wu, Zheng, Zhang, Zhang, and Wang}]{wu2022non}
Wu, Yuhang, Zeyu Zheng, Guangyu Zhang, Zuohua Zhang, Chu Wang. 2022.
\newblock Non-stationary a/b tests.
\newblock {\it Proceedings of the 28th ACM SIGKDD Conference on Knowledge Discovery and Data Mining\/}. 2079--2089.

\bibitem[{Wu et~al.(2023)Wu, Zheng, Zhang, Zhang, and Wang}]{wu2023non}
Wu, Yuhang, Zeyu Zheng, Guangyu Zhang, Zuohua Zhang, Chu Wang. 2023.
\newblock Non-stationary a/b tests: Optimal variance reduction, bias correction, and valid inference.
\newblock {\it Bias Correction, and Valid Inference (June 18, 2023)\/} .

\bibitem[{Xiong et~al.(2023{\natexlab{a}})Xiong, Athey, Bayati, and Imbens}]{xiong2019optimal}
Xiong, Ruoxuan, Susan Athey, Mohsen Bayati, Guido~W Imbens. 2023{\natexlab{a}}.
\newblock Optimal experimental design for staggered rollouts.
\newblock {\it Management Science\/} .

\bibitem[{Xiong et~al.(2023{\natexlab{b}})Xiong, Chin, and Taylor}]{xiong2023bias}
Xiong, Ruoxuan, Alex Chin, Sean Taylor. 2023{\natexlab{b}}.
\newblock Bias-variance tradeoffs for designing simultaneous temporal experiments.
\newblock {\it The KDD'23 Workshop on Causal Discovery, Prediction and Decision\/}. PMLR, 115--131.

\bibitem[{Ye et~al.(2023{\natexlab{a}})Ye, Zhang, Zhang, Zhang, Chen, and Xu}]{ye2023cold}
Ye, Zikun, Dennis~J Zhang, Heng Zhang, Renyu Zhang, Xin Chen, Zhiwei Xu. 2023{\natexlab{a}}.
\newblock Cold start to improve market thickness on online advertising platforms: Data-driven algorithms and field experiments.
\newblock {\it Management Science\/} {\bf 69}(7) 3838--3860.

\bibitem[{Ye et~al.(2023{\natexlab{b}})Ye, Zhang, Zhang, Zhang, and Zhang}]{ye2023deep}
Ye, Zikun, Zhiqi Zhang, Dennis Zhang, Heng Zhang, Renyu~Philip Zhang. 2023{\natexlab{b}}.
\newblock Deep learning based causal inference for large-scale combinatorial experiments: Theory and empirical evidence.
\newblock {\it Available at SSRN 4375327\/} .

\bibitem[{Yuan et~al.(2021)Yuan, Altenburger, and Kooti}]{yuan2021causal}
Yuan, Yuan, Kristen Altenburger, Farshad Kooti. 2021.
\newblock Causal network motifs: identifying heterogeneous spillover effects in a/b tests.
\newblock {\it Proceedings of the Web Conference 2021\/}. 3359--3370.

\bibitem[{Yuan and Altenburger(2023)}]{yuan2023two}
Yuan, Yuan, Kristen~M Altenburger. 2023.
\newblock A two-part machine learning approach to characterizing network interference in a/b testing.
\newblock {\it arXiv preprint arXiv:2308.09790\/} .

\end{thebibliography}
